\pdfoutput=1
\RequirePackage[hyphens]{url}
\documentclass[pra,twocolumn,superscriptaddress,nofootinbib,longbibliography]{revtex4-1}
\usepackage[utf8]{inputenc}
\usepackage{xcolor}
\usepackage{hyperref}
\usepackage{amssymb}
\usepackage{amsmath}
\usepackage{pifont}
\newcommand{\cmark}{\ding{51}}
\newcommand{\xmark}{\ding{55}}
\usepackage{soul}
\usepackage{tikz}
\usepackage{tikz-cd}
\usepackage{graphicx}
\usepackage{color}
\usepackage{tensor}
\newcommand{\UU}{\mathrm{U}}
\newcommand{\LU}{\mathrm{LU}}

\newcommand{\SU}{\mathrm{SU}}

\newcommand{\eqnref}[1]{Eq.~(\ref{#1})}
\newcommand{\ket}[1]{|#1\rangle}
\newcommand{\bra}[1]{\langle #1 |}
\newcommand{\GIR}{G_{\mathrm{IR}}}
\newcommand{\GUV}{G_{\mathrm{UV}}}
\usepackage{amsthm}
\usepackage{cancel}
\newtheorem{thm}{Theorem}
\newtheorem{lemma}{Lemma}

\newcommand{\dblul}[1]{\underline{\underline{#1}}}
\def\slashchar#1{\setbox0=\hbox{$#1$}           % set a box for #1
   \dimen0=\wd0                                 % and get its size
   \setbox1=\hbox{/} \dimen1=\wd1               % get size of /
   \ifdim\dimen0>\dimen1                        % #1 is bigger
      \rlap{\hbox to \dimen0{\hfil/\hfil}}      % so center / in box
      #1                                        % and print #1
   \else                                        % / is bigger
      \rlap{\hbox to \dimen1{\hfil$#1$\hfil}}   % so center #1
      /                                         % and print /
   \fi}
   
\usetikzlibrary{decorations.markings}
\tikzset{fermisurface/.style={color=orange,very thick}}
\tikzset{fermifill/.style={fill=yellow!30}}
\tikzset{ferminode/.style={color=orange,fill=orange}}
\tikzset{brillouinzone/.style={color=gray!40,ultra thick}}
\newcommand{\fsmarking}[1]{\textcolor{orange!90!black}{\textsf{#1}}}
\newcommand{\ferminoderadius}{0.03}
\tikzset{orientedfs/.style={fermisurface,decoration={
    markings,
    mark=at position 0.5 with {\arrow{>}}},
    postaction=decorate}}
\tikzset{fermifilldouble/.style={fill=yellow!70}}

\begin{document}
\title{
Non-Fermi liquids as ersatz Fermi liquids: general constraints on compressible metals}
\author{Dominic V. Else}
\affiliation{Department of Physics, Massachusetts Institute of Technology, Cambridge, MA 02139, USA}
\author{Ryan Thorngren}
\affiliation{Center for Mathematical Sciences and Applications, Harvard University, Cambridge, MA 02138}
\author{T. Senthil}
\affiliation{Department of Physics, Massachusetts Institute of Technology, Cambridge, MA 02139, USA}
\begin{abstract}
A system with charge conservation and lattice translation symmetry has a well-defined filling $\nu$, which is a real number representing the average charge per unit cell. We show that if $\nu$ is fractional (i.e.\ not an integer), this imposes very strong constraints on the low-energy theory of the system and give a framework to understand such constraints in great generality, vastly generalizing the Luttinger and Lieb-Schultz-Mattis theorems. The most powerful constraint comes about if $\nu$ is continuously tunable (i.e.\ the system is charge-compressible), in which case we show that the low-energy theory must have a very large emergent symmetry group -- larger than any compact Lie group. An example is the Fermi surface of a Fermi liquid, where the charge at every point on the Fermi surface is conserved.
We expect that in many, if not all, cases, even exotic non-Fermi liquids will have the same emergent symmetry group as a Fermi liquid, even though they could have very different dynamics. We call a system with this property an \emph{ersatz Fermi liquid}. We show that ersatz Fermi liquids share a number of properties in common with Fermi liquids, including Luttinger's theorem (which is thus extended to a large class of non-Fermi liquids) and periodic ``quantum oscillations'' in the response to an applied magnetic field. We also establish versions of Luttinger's theorem for the composite Fermi liquid in quantum Hall systems and for spinon Fermi surfaces in Mott insulators. Our work makes connection between filling constraints and the theory of symmetry-protected topological (SPT) phases, in particular through the concept of ``'t Hooft anomalies''.
\end{abstract}
\maketitle
\tableofcontents

\section{Introduction}

In condensed matter physics, systems with prescribed microscopic degrees of freedom (usually electrons) can exhibit varied and exotic emergent behavior at low energies. In general, it is extremely difficult, either analytically or numerically, to predict the nature of the emergent low-energy behavior [as described by an ``IR (infra-red) theory''] from the properties of the microscopic degrees of freedom. For this reason, it is invaluable to have general results that constrain the nature of the IR theory, given microscopic properties of the system.

One such result is the Lieb-Schultz-Mattis theorem \cite{Lieb_1961}, which states (in the formulation of interest to us here \cite{Oshikawa_9610}) that in a system in one spatial dimension with a conserved $\UU(1)$ charge and discrete translation symmetry, if the average charge per unit cell (which we call the ``filling'') is not an integer, then either the IR theory is gapless or else it spontaneously breaks one of the symmetries. This result was later generalized to higher dimensions by Oshikawa \cite{Oshikawa_9911} and Hastings \cite{Hastings_0411}; in higher dimensions there is also the possibility that a system at fractional filling can be gapped but with non-trivial topological order \cite{Misguich_0112}.

A related result is Luttinger's theorem \cite{Luttinger_1960}, which states that if the IR theory is a Fermi liquid, then the volume enclosed by the Fermi surface (modulo the volume of the Brillouin Zone)  is determined entirely by the fractional part of the filling, and in particular is independent of the interaction strength. Originally proven perturbatively by Luttinger, the result was later established through a non-perturbative argument (but still assuming that the IR theory is describable by Fermi liquid theory) by Oshikawa \cite{Oshikawa_0002}. 
A generalized version of Luttinger's theorem even holds in a class of phases - known as Fractionalized Fermi Liquids \cite{Senthil_0209,Senthil_0305}  - which are distinct from conventional  Fermi liquids.  In  such phases a gapless Fermi liquid co-exists with non-trivial topological order. The generalization to Luttinger's theorem can then exactly be determined  from the interplay of the symmetry with the topological order \cite{Senthil_0305,Paramekanti_0406,Bonderson_1601}.

These results and others \cite{Watanabe_1505,Lu_1705,Bultinck_1808} raise the question: what is the most general statement  that one can make about the relation between the microscopic filling and properties of the IR theory? In the present work, we will answer this question, showing that the microscopic filling is always completely determined  by a few properties of the IR theory, namely: (a) its emergent symmetries; (b) the relation between the microscopic translation and $\UU(1)$ symmetries and the emergent symmetries of the IR theory; and (c) a property of the IR theory called its `` 't Hooft anomaly''. Some connections between filling and 't Hooft anomalies have previously been explored in Refs.~\cite{Cheng_1511,Cho_1705,Jian_1705, Metlitski_1707,Yao_1906,Song_1909}.

Our work has important implications for the study of ``non-Fermi liquids'', which are systems that are metallic down to zero temperature but for which the IR physics cannot be described by Fermi liquid theory. Motivated by the results on filling constraints just mentioned, we will introduce the concept of an \emph{ersatz Fermi liquid (EFL)} as a general framework to understand non-Fermi liquids. An EFL is a system which has the same \emph{kinematic} properties in the deep infrared as a Fermi liquid, though it might have very different \emph{dynamical} properties. By ``kinematic'' properties we mean properties that relate to the structure of the Hilbert space that  describes the ground state and the low-energy excitations, as opposed to ``dynamical" properties which relate to the Hamiltonian that acts in this Hilbert space. More precisely, the kinematic properties comprise properties (a), (b) and (c) described in the previous paragraph and in more detail in the next section.
Further, it will will turn out that the kinematic properties we will discuss have a strong topological flavor, so in a suitably vague sense we can say that an EFL is a system that is ``topologically equivalent'' to a Fermi liquid. (However, since the dynamical properties of an EFL can be sharply different to that of a Fermi liquid, it will generally \emph{not} be the case that an EFL can be continuously deformed into a Fermi liquid.)

We will show that many (though by no means all) well-known aspects of Fermi liquid phenomenology are, in fact, purely kinematic in nature, and therefore apply equally well to any EFL. In particular, we will show that any EFL has a Fermi surface that hosts long-lived excitations (though these excitations may not be Landau quasiparticles). In fact, our approach leads to a very general perspective on what it means for a system to have a Fermi surface.
Furthermore, we will find that this Fermi surface must obey Luttinger's theorem (or a generalization thereof, analogously to the ``fractionalized Fermi liquids''
mentioned above); and that, if the Fermi surface geometry is such that a Fermi liquid with that geometry would exhibit quantum oscillations in the dependence of physical properties on magnetic field, then any EFL with the same Fermi surface geometry is also expected to display quantum oscillations with the same periodicity. 
 
Since none of the kinematic properties that define an EFL require that the system be weakly coupled or have a description in terms of quasiparticles, we expect a wide variety of exotic non-Fermi liquid phenomena to be realizable within the class of EFLs. In fact, we will argue based on the general theory of filling constraints that any IR theory which describes a compressible  metal, i.e.\ the filling can be continuously tuned\footnote{A perhaps more standard definition of ``compressible'' would be that the partial derivative of the particle number with respcet to the chemical potential is nonzero, however we will not use this definition in this paper.} must have a very non-trivial emergent symmetry group, larger than any compact Lie group. Such a property is indeed satisfied by EFLs (due to the infinitely many conserved quantities associated with the Fermi surface); whether it could be satisfied in a different way that leads to fundamentally different kinematic properties is an important open problem.

At the very least, however, it is clear that a number of non-Fermi liquid metals can be fruitfully discussed  from the perspective developed in this paper. The simplest are non-Fermi liquid metals that arise when a Fermi surface is coupled to a critical boson that represents a fluctuating order parameter.  Within the standard  framework for such quantum critical points, they will be seen to be EFLs.
A closely related system is a Fermi surface coupled to a gapless $\UU(1)$ gauge field, which arises in the theory of composite Fermi liquid metals in the quantum Hall regime, and in some {\em insulating} quantum spin liquids with Fermi surfaces of emergent electrically neutral quasiparticles. In these states we will find that, unlike a strict EFL, the kinematic properties of the IR theories can differ from those of a Fermi liquid (indeed, unlike a Fermi liquid, these are not states in which the microscopic density can be continuously tuned) though still closely related. The main point will still stand, though, that the kinematic properties of these IR theories provide a powerful framework to thinking about their universal behavior. In particular we will show that in such systems a version of Luttinger's theorem is still satisfied.

As a final application of our results, we will examine the possibility for systems to exhibit disconnected Fermi arcs instead of a closed Fermi surface. We will find that, assuming the translational symmetry is unbroken, such a scenario is inconsistent with the IR theory being an EFL, except when the system exists on the boundary of a gapless bulk (such as a Weyl semimetal). This provides strong evidence for the impossibility of Fermi arcs. In fact, we obtain a stronger constraint: in an EFL, the Fermi surface must enclose a volume in the Brillouin zone; this volume is the generalization to EFLs of the ``electron sea'' in a weakly interacting system.

Before proceeding, let us make a final technical remark. We note that previous work\cite{coleman2005sum,powell2005depletion,huijse2011fermi} generalized the original perturbative proof of Luttinger's theorem to situations where there are a number of fermion and boson fields with $U(1)$ symmetries (either global or gauge). This proof relies on the existence of a Luttinger-Ward functional  of the exact Green's functions of the fermions/bosons from which the self-energies can be extracted by functional derivatives. These works lead to the expectation that Luttinger's theorem will be satisfied by some classes of non-Fermi liquids. However, as the quasiparticle is destroyed in the IR in such problems, perturbation theory (even to all orders) should be used with caution. In a {\em non-perturbative} context, despite the formal existence\cite{potthoff2004non} of a Luttinger-Ward functional, the conventional proof of Luttinger's theorem fails. This is dramatically illustrated by the  fractionalized Fermi liquid phases. It is thus desirable to have a more general non-perturbative argument for  Luttinger's theorem in this context, which we provide in this paper.

\section{Kinematic properties of the IR theory}
\label{sec:structural_properties}
Among all the features of of the IR theory describing the low-energy properties of a system, there are certain ones that we will refer to as ``kinematic'', and these are the subject of this section.

Before describing the first kinematic property, let us note that the microscopic system will have a group of symmetries $\GUV$. We will specifically be interested in systems where this includes a global $\UU(1)$ corresponding to charge conservation, and translation symmetries (possibly on a lattice). 

The first kinematic property of the IR theory is the set of emergent symmetries. We thus  introduce the group $\GIR$ of emergent  symmetries\footnote{For most of the paper we will focus on the most  familiar kind of emergent symmetry;  what  has retroactively been renamed ``0-form symmetry'', i.e.\ symmetries that act everywhere in spacetime at once. Recently, the importance of so-called ``higher-form symmetries'' has started to be understood \cite{Gaiotto_1412}.  Loosely speaking the restriction to the ordinary 0-form symmetries means that there are no emergent ``fractionalized excitations''. We say a point-like excitation is ``fractionalized'' if it cannot be created by a point-like creation operator, but rather is created at the endpoints of an open string by an operator supported on the string. There are similar definitions for higher-dimensional excitations. Fractionalized excitations emerge, for instance, in phases where there is topological order and/or deconfined emergent gauge fields. For such systems, there will be emergent higher-form symmetries that will play an important role, but in the interest of simplifying the exposition we will defer discussion of this point to later, and for the moment we assume that all the emergent symmetries are 0-form symmetries. We will eventually return to this point in Section \ref{sec:higher_form}.} of the IR theory.
$\GIR$ is, in general, not the same as the microscopic symmetry group $\GUV$. Nevertheless, each element $g \in \GUV$ of the microscopic symmetry group gets mapped into an element $\varphi(g) \in \GIR$ of the emergent symmetry group, such that $\varphi(g_1) \varphi(g_2) = \varphi(g_1 g_2)$ for all $g_1, g_2 \in G$. There could, of course, be elements of the emergent symmetry group $\GIR$ that do not correspond to any microscopic symmetry. Also, a microscopic symmetry $g \in \GUV$ could act trivially in the IR theory, in which case $\varphi(g) = 1$ (the identity element in $\GIR$); an example of the latter case would be for systems with a charge gap, in which case the microscopic $\UU(1)$ symmetry acts trivially in the IR. All these statements can be expressed in a compact mathematical way by saying that $\varphi$ defines a group homomorphism from $\GUV$ to $\GIR$ that need not be injective or surjective.

Note that for many of the arguments in this paper, we will specifically want $\GIR$ to represent the \emph{internal} symmetries of the system (that is, the symmetries which do not move space-time points around). In general, the IR theory will also have ``trivial'' emergent space-time symmetries such as continuous translation symmetry, which we do not include in $\GIR$. It is important to note that the microscopic translation symmetry in general will not map into these trivial translation symmetries, but rather into the internal symmetry $\GIR$. One way to think about this is that since there is a spatial rescaling transformation associated with passing to the IR theory, a microscopic translation symmetry in fact has trivial translation action in the IR limit. An alternative perspective is that we imagine that a microscopic translation symmetry acts like a product of an internal symmetry and a ``trivial'' translation symmetry, and we only worry about the internal part in defining the map $\varphi$.

 The other property of the IR theory that will be pertinent is the extent to which the full emergent symmetry group $\GIR$ can be `naturally'  realized in {\em some} realization  (not necessarily the original microscopic lattice model)  of the IR theory.  This property is formalized by the concept of the 't Hooft anomaly. Such an anomaly in $d+1$ space-time dimensions is an obstruction to UV-regularizing the theory on a lattice in $d$ spatial dimensions with the full emergent symmetry group $\GIR$ realized as an ``on-site'' microscopic symmetry \cite{Chen_1106_4752,Else_1409}.  A  powerful alternative but formal characterization of an  't Hooft anomaly \cite{tHooftAnomaly, Kapustin_1403_0617} is that the conservation law corresponding to the $\GIR$ symmetry is broken upon coupling to a background gauge field for the symmetry $\GIR$.
We will consider to the 't Hooft anomaly of the emergent symmetry $\GIR$ to also be a `kinematic' property of the IR theory. 

't Hooft anomalies are also closely related to the theory of symmetry-protected topological (SPT) phases \cite{Gu_0903,Senthil_1405,Wen_1610}.   These are gapped phases of matter in a system with an unbroken global symmetry $G$ with the following property:  the ground state cannot be continuously deformed into a trivial product-state ground state while preserving the symmetry without closing the gap if the symmetry $G$ is preserved, but can be if the symmetry $G$ is lifted.
Examples of SPT phases include the celebrated topological insulators \cite{Hasan_1002} and the Haldane phase of the spin-$1$ antiferromagnetic chain in $d = 1$ \cite{Haldane_1983a,Haldane_1983b,Pollmann_0910}.

The connection between 't Hooft anomalies and SPTs is that a $G$ SPT phase in $d+1$ spatial dimensions must have a non-trivial boundary theory, and the boundary theory carries a 't Hooft anomaly. We can say that the 't Hooft anomaly of the boundary theory is ``canceled'' by inflow from the bulk, in the sense that the conservation laws of the bulk+boundary system are preserved in the presence of background gauge fields. Thus, the classification of 't Hooft anomalies in $d$ spatial dimensions is precisely equivalent to the classification of SPT phases in $d+1$ spatial dimensions\footnote{There can also be 't Hooft anomalies associated with ``invertible'' topological phases which do not require any symmetry to protect them.}.
 Such classification has been explored at great depth, from a variety of perspectives, ranging from physical considerations to very  formal 
 ones\cite{Qi_0802,Pollmann_0909, Chen_1008, Schuch_1010, Chen_1103, Chen_1106_4752, Chen_1106_4772, Gu_1201, levin2012braiding,Lu_1205, Vishwanath_1209, fidkowski2013non,wang2014classification,wang2014interacting,Kapustin_1403_1467,Senthil_1405,metlitski2014interaction, Kapustin_1406, Else_1409, Cheng_1501, KitaevIPAM, witten2016fermion, Freed_1604, Xiong_1701, Kapustin_1701, Wang_1703, Gaiotto_1712, Wang_1811}  The aspects of SPT phases that we will need to use in this paper, however, will be simple enough, at least if we want to understand filling in systems of low spatial dimension $d \leq 2$, that a reader unfamiliar with this literature should still be able to follow our paper.
 
\subsection{Example: Luttinger liquid in 1 spatial dimension}
\label{subsec:luttinger_liquid}
Let us illustrate the above general considerations in the case of a system of spinless electrons in a lattice in one spatial dimension. Thus, the microscopic symmetry group $\GUV$ is comprised of a $\UU(1)$ symmetry generated by the total electron number $\hat{Q}$, and a $\mathbb{Z}$ symmetry generated by the lattice translation operator $\mathbb{T}$.

Let us now assume that the IR theory of the electrons is a Luttinger liquid. Thus, the low-energy physics takes place at the two Fermi points at momenta $k_L$ and $k_R$. Excitations with momentum close to $k_L$ are left-movers, and those with momentum close to $k_R$ are right-movers. At low energies, the numbers $\hat{N}_L$ and $\hat{N}_R$ of left- and right-movers are separately conserved. Therefore, they generate the emergent symmetry group $\GIR = \UU(1)_L \times \UU(1)_R$.

Next, we need to specify how the microscopic symmetry acts on the IR theory, which we can do by expressing the generators of the microscopic symmetry in terms of the generators of the emergent symmetry. Indeed, we have
\begin{align}
\hat{Q} &\sim \hat{N}_L + \hat{N}_R \label{eq:charge_embedding_luttinger} \\
\mathbb{T} &\sim \exp(-i[k_L \hat{N}_L + k_R \hat{N}_R]). \label{eq:translation_embedding_luttinger},
\end{align}
where the tilde ``$\sim$'' refers to an equivalent action on the IR theory.

The emergent symmetry group $\GIR$ in a Luttinger liquid has the  well known axial anomaly. This is an example of the  't Hooft anomaly mentioned in the previous subsection. The signature is that, if we turn on an electric field $E$ for the $\UU(1)$ symmetry generated by $\hat{N}_L + \hat{N}_R$ [which is equivalent to a microscopic electric field, by \eqnref{eq:charge_embedding_luttinger}], then left- and right-moving  charges are no longer separately conserved; instead, if we let $j\indices{_L^\mu}$, $\mu = 0,1$ be the current density for the left-moving charge, and similarly for $j\indices{_R^\mu}$, we have
\begin{align}
\partial_\mu j\indices{_L^\mu} &= -\frac{E}{2\pi}, \nonumber\\
\partial_\mu j\indices{_R^\mu} &= \frac{E}{2\pi}. \label{eq:axial_anomaly}
\end{align}
One can easily understand this equation in the case of a non-interacting Fermi gas. In that case, in the absence of electric field the electrons occupy single-particle states labeled by momentum $k$. But applying an electric field pointing to the right causes an overall flow of electrons in momentum space according to $\dot{k} = E$. This causes a charge excess to accumulate at  $k_R$ and a corresponding charge deficit at $k_L$.

As we mentioned above, there is always a bulk-boundary correspondence that relates a 't Hooft anomaly for a symmetry $G$ in $d$ spatial dimensions to an SPT in $d+1$ spatial dimensions whose boundary theory carries the 't Hooft anomaly. In this case, we have $d=1$ and $G = \UU(1) \times \UU(1)$. The corresponding SPT phase in $d=2$ is realized by a ``quantum spin Hall'' state corresponding to putting spin-up electrons in a quantum Hall state with quantized Hall conductance $\sigma_{xy} = 1$ and spin-down electrons in a quantum Hall state with $\sigma_{xy} = -1$. Here the $\UU(1) \times \UU(1)$ symmetries correspond to the separate conservation of spin-up and spin-down electrons. If we now consider a system with boundary, and apply an electric field parallel to the boundary, due to the Hall conductance this will generate a current of spin-up electrons incident onto the boundary, and a current of spin-down electrons with opposite sign. This precisely accounts for for the charge non-conservation in the boundary theory due to the 't Hooft anomaly.

Finally, let us remark that a useful way to think about SPT phases (and hence 't Hooft anomalies) is in terms of topological terms describing the response of the SPT phase to background gauge fields. For example, consider the quantum spin Hall state described above. We can theoretically couple to background gauge fields of the two $\UU(1)$ symmetries, which in the one-dimensional Luttinger liquid were interpreted as the conservation of left- and right-movers; hence, we denote the gauge fields by $A^L$ and $A^R$. The response of the quantum spin Hall state is then described by a Chern-Simons action on $(2+1)$-D spacetime:
\begin{equation}
\label{eq:mutual_chern_simons}
S[A] = \frac{1}{4\pi} \int (A^R \wedge dA^R - A^L \wedge dA^L).
\end{equation}

\section{Filling constraints}
\label{sec:filling_constraints}
Now we will turn to the question of how to understand constraints on the IR theory resulting from the microscopic filling, i.e. the average charge per unit cell, which is a real number $\nu$. Only the fractional part of $\nu$ (i.e.\ $\nu$ mod 1), should be expected to be detectible in the IR theory, because an atomic insulator (whose IR theory is completely trivial) can have any integer filling. The fundamental observation we will make is that the microscopic filling is completely fixed by the kinematic properties of the IR theory, i.e.\ the emergent symmetry group and the 't Hooft anomaly, along with the mapping from the microscopic symmetry group into the emergent symmetry group.

This is a rare example of a precise relation between a microscopic quantity (the filling) and properties of the IR theory. Such relations are extremely useful, given that it is usually very difficult to determine the IR theory from the microscopic Hamiltonian, either analytically or numerically.

The reason why such a UV-IR correspondence is possible in this case is because the UV and IR are linked through the homomorphism $\varphi$ that implements the microscopic symmetry inside the emergent IR symmetry. Naively, we can imagine arguing as follows. The lattice filling is defined in systems that have at least a microscopic $\mathrm{U}(1)$ symmetry (charge conservation), and a lattice translation symmetry $\mathbb{Z}^d$ (where $d$ is the space dimension). So we set the microscopic symmetry $\GUV = \mathrm{U}(1) \times \mathbb{Z}^d$. This  $\GUV$ is embedded into the symmetries of the IR theory, and so we can talk about the 't Hooft anomaly of the IR theory thought of as a $\GUV$-symmetric theory (which can be computed from the $\GIR$ 't Hooft anomaly in light of the homomorphism $\varphi : \GUV \to \GIR$). Then we invoke ``UV-IR anomaly matching'' to relate the filling $\nu$ to this $\GUV$ 't Hooft anomaly.

Unfortunately, this argument is not \emph{quite} right. In fact, fractional filling cannot correspond to a nontrivial 't Hooft anomaly for $\GUV$ in the usual sense, because there are no candidate SPTs with $\GUV = \UU(1) \times \mathbb{Z}^d$ symmetry which could constitute the ``bulk'', for which a system with fractional filling could be the boundary, as we explain in Appendix \ref{appendix:fillingnotspt}. Indeed, a system with fractional filling seems perfectly well-defined on its own without any need for a higher-dimensional bulk.

Instead, we will give the correct version of the argument below, in various spatial dimensions. The careful reader might notice that they still appear reminiscent of ``anomaly matching'', albeit for a 't Hooft anomaly that is trivial with respect to the microscopic symmetries. How exactly one could make such a notion precise, we leave as an open question. Nevertheless, the arguments are self-contained and  hold without any need for such an interpretation.

\subsection{One spatial dimension}
\label{subsec:1d_filling_constraints}
In one spatial dimension, we imagine putting the system on a circle and then very slowly threading $2\pi$ flux of the microscopic $\mathrm{U}(1)$ symmetry through the circle \cite{Laughlin_1981,Oshikawa_9911,Oshikawa_0002}, By a standard argument, this transforms the ground state into a low-lying excited state with a different momentum; if we label states by their eigenvalues of the translation operator $\mathbb{T}$, i.e. $\mathbb{T} \ket{\psi} = e^{-i p} \ket{\psi}$, then the momentum gets shifted by $e^{-ip} \to e^{-i(p+2\pi\nu)}$.

\newcommand{\hQIR}{\hat{Q}_{\mathrm{IR}}}
Meanwhile, since the microscopic $\mathrm{U}(1)$ symmetry, generated by $\hat{Q}$, corresponds to a $\mathrm{U}(1)$ symmetry of the IR theory, whose generator we call $\hQIR$, we can also imagine performing the $2\pi$ flux insertion in the IR theory. Now consider the IR symmetry $\tau = \varphi(\mathbb{T}) \in \GIR$ corresponding to microscopic translation $\mathbb{T}$. In the IR theory, the ground state can get transformed into a low-lying excited state with a different eigenvalue of $\tau$; that is, the eigenvalue of $\tau$ gets shifted according to
$\exp(-i\theta) \to \exp\left(-i[\theta + \alpha]\right)$ for some $
\alpha$ [defined mod $2\pi$]. One can argue that $\alpha$ depends \emph{only} on the 't Hooft anomaly of $\GIR$ and on the choice of $\tau$ and $\hQIR$; see for instance the example below. We write $\alpha = \alpha(\hQIR|\tau)$ (the dependence on the 't Hooft anomaly is kept implicit).

Now, the key point is that the processes described in the two paragraphs above are in fact the \emph{same} process, just described in two different ways. Therefore, we must equate
\begin{equation}
\label{eq:1d_equating}
\nu = \frac{\alpha(\hQIR|\tau)}{2\pi} \quad (\mathrm{mod}\,1).
\end{equation}

\subsubsection{Example: Luttinger liquid in one spatial dimension}
We can apply the general framework described above to the particular example of a Luttinger liquid in one spatial dimension, as discussed in Section \ref{subsec:luttinger_liquid}. In this case $\hQIR$ and $\tau$ are defined by the right-hand sides of Eqs.~(\ref{eq:charge_embedding_luttinger}) and (\ref{eq:translation_embedding_luttinger}).

Threading the flux of $\hat{Q}_{\mathrm{IR}}$ generates an electric field of $\hat{Q}_{\mathrm{IR}}$ by Faraday's law, and then, from the anomaly equation \eqnref{eq:axial_anomaly} we see that the $2\pi$ flux threading creates $-1$ charge of $\mathrm{U}(1)_L$ and $+1$ charge of $\mathrm{U}(1)_R$. Therefore, in light of \eqnref{eq:translation_embedding_luttinger}, the momentum transforms according to
\begin{equation}
e^{-ip} \to e^{-ip} e^{-i (k_R - k_L)}.
\end{equation}
Then, from \eqnref{eq:1d_equating} we find that
\begin{equation}
\nu = \frac{1}{2\pi}(k_R - k_L) \quad (\mathrm{mod} \, 1),
\end{equation}
which is nothing other than Luttinger's theorem for a Luttinger liquid in one spatial dimension.

\subsection{Two spatial dimensions}
\label{subsec:2d_filling_constraint}
Here the idea is to consider a ``$2\pi$ flux'' of the microscopic $\mathrm{U}(1)$ symmetry; that is, a very weak background magnetic field spread out over some very large region, such that the total flux is $2\pi$.
We can then consider how such a $2\pi$ flux transforms under translation symmetry. We claim that, in the  presence of fractional filling $\nu$, such $2\pi$ fluxes exhibit translational symmetry fractionalization; that is, acting on a $2\pi$ flux, the lattice $x$ and $y$ translations $\mathbb{T}_x$ and $\mathbb{T}_y$ obey in the limit as the $2\pi$ flux becomes infinitely spread out spatially) the magnetic algebra
\begin{equation}
\label{eq:pretrans_projective}
\mathbb{T}_x \mathbb{T}_y \mathbb{T}_x^{-1} \mathbb{T}_y^{-1} = e^{2\pi i \nu}.
\end{equation}
Heuristically, this is clear because a $2\pi$ flux sees a background charge density as an effective magnetic field. We give some more careful arguments in Appendix \ref{appendix:monopole_translations}.  This result is also closely connected to the translational symmetry fractionalization of a monopole in an insulator in 3 spatial dimensions in the presence of polarization, as discussed in Ref.~\cite{Song_1909}.

Meanwhile, in the IR theory we can consider $2\pi$ flux configurations of the IR symmetry generated by $\hQIR$ that corresponds to the microscopic $\UU(1)$ symmetry. The homomorphism $\varphi$ maps $\mathbb{T}_x$ and $\mathbb{T}_y$ into some elements $\tau_x, \tau_y \in \GIR$. In the presence of a 't Hooft anomaly, such $2\pi$ fluxes indeed can carry a projective representation, i.e.
\begin{equation}
\label{eq:trans_projective}
V(\tau_x) V(\tau_y) V(\tau_x)^{-1} V(\tau_y)^{-1} = e^{i \alpha},
\end{equation}
where $V(g)$ denotes the action of a group element $g \in \GIR$ on the $2\pi$ flux, and $\alpha = \alpha(\hQIR|\tau_x, \tau_y)$ [defined mod $2\pi$] depends on $\hQIR$, $\tau_x$ and $\tau_y$, and on the 't Hooft anomaly of $\GIR$. Therefore, we must identify
\begin{equation}
\label{eq:filling_constraint_2d}
\nu = \frac{\alpha(\hQIR|\tau_x, \tau_y)}{2\pi} \quad (\mathrm{mod} \, 1).
\end{equation}
The natural example to consider to illustrate this constraint will be a Fermi liquid. However, by contrast to the one-dimensional case, the 't Hooft anomaly of a Fermi liquid in two spatial dimensions has not previously been discussed. This is the subject of Section \ref{sec:structural_fermi_liquid}.

\subsection{General space dimension}
\label{subsec:general_dimension}

The formulation of filling constraints discussed above can be generalized to arbitrary space dimension $d$. The idea is to generalize the functions $\alpha(\hQIR| \tau)$ (in $d=1$) and $\alpha(\hQIR| \tau_x, \tau_y)$ (in $d=2$) to a function $\alpha(\hQIR|\tau_1, \cdots, \tau_d)$ that determines the filling in general spatial dimension $d$.
A convenient way to express this function is in terms of the ``topological action'' that describes the SPT phase in $d+1$ spatial dimensions which cancels the $\GIR$ 't Hooft anomaly in $d$ spatial dimensions by inflow on the boundary. We give the details in Appendix \ref{appendix:topological_action}.

\section{Consequences of the filling constraints for compressible states}
\label{sec:generic}
The systematic theory of filling constraints described in the previous section has a very important corollary. We want to consider IR theories which can exist at \emph{generic} filling; that is, they are not pinned to a particular filling but instead the filling can be continuously tuned. In other words, the IR theory represents a ``compressible'' state.
 What we will show is that in this case, for spatial dimensions $\geq 2$, the emergent symmetry $\GIR$ \emph{cannot} be a compact Lie group. In one dimension, recall that the Luttinger liquid example discussed in Section \ref{subsec:luttinger_liquid} achieves generic filling with only a compact Lie group emergent symmetry $\GIR = \UU(1) \times \UU(1)$.
  Note that compact Lie groups include finite groups as special cases, since we do not require the Lie group to be connected; also, in this paper when we refer to Lie groups, we will always assume they are finite-dimensional.
  
We wish to emphasize here that as we stated previously, in this paper when we refer to the emergent symmetry group $G_{\mathrm{IR}}$, we are referring specifically to the \emph{internal} symmetries of the IR theory. While translation symmetry is never a compact group, we do not know of any examples of theories that can sensibly arise in condensed matter systems, for which the \emph{internal} symmetry is a non-compact Lie group, and this may in fact be impossible.
Therefore, our results suggest that the emergent symmetry for IR theories that exist at generic filling must be  an infinite-dimensional group\footnote{We remind the reader again of the restriction that for the moment we assume that all the emergent symmetries are 0-form symmetries, which we re-examine in Section \ref{sec:higher_form}. In fact, theorem 1 continues to hold in the presence of \emph{finite} higher form symmetries.}. Indeed, this is the case for Fermi liquids, as we describe in the next section.

Our main result is the following:
\begin{thm}
\label{thm:mainthm}
Suppose the emergent symmetry group of the IR theory is some compact Lie group $\GIR$. Then for any spatial dimension $d \geq 2$, the filling $\nu$ is constrained to be an integer multiple of $1/N_{\GIR}$, for some finite integer $N_{\GIR}$ that depends only on the group $\GIR$ and the dimension.
\end{thm}
The proof is straightforward based on the framework of filling constraints discussed in the previous section. We give the details, and reveal what determines $N_{\GIR}$, in Appendix \ref{appendix:proof_of_mainthm}, and also give a more formal point of view in Appendix \ref{appendix:chern_simons}. 

Let us mention a simple way to understand this result for $d=2$. As mentioned in the introduction, the 't Hooft anomaly of a theory in $d$ spatial dimensions implies that the theory can be realized as the boundary of an SPT phase protected by the same symmetry $\GIR$ in $d+1$-space dimensions. Now suppose that the SPT  is such that  $n$ copies of it is trivial, with $n$ a finite positive integer. For  SPTs protected by  a compact Lie group $\GIR$ in space dimension $d+1  = 3$, this is known to be always true.  For the $d$-dimensional boundary theory of  interest, this means that for $n$ copies, there is no 't Hooft anomaly. If now we consider $n$  copies of the microscopic lattice system, we see that it has a total filling $n\nu$. Since at this filling the IR theory has no anomaly, it follows that $\nu n = p$ with $p$ an integer, which is essentially the claim of Theorem \ref{thm:mainthm}. In $d = 3$, this simple argument does not work because there are SPTs in $d + 1 = 4$ space dimensions such that there is no finite  $n$ for which $n$ copies become trivial. (This is also why the argument does not work for $d=1$). However, Theorem \ref{thm:mainthm} still holds for any $d \geq 2$ as we show in Appendices \ref{appendix:proof_of_mainthm} and \ref{appendix:chern_simons}.

\section{The kinematic properties of Fermi liquids}
\label{sec:structural_fermi_liquid}
In this section, we will return to a familiar IR theory: a Fermi liquid in two spatial dimensions (we will briefly discuss higher dimensions as well) and analyze its kinematic properties, in the language introduced previously. We will see how Fermi liquids, by virtue of having an emergent symmetry group that is ``larger'' than a compact Lie group, are able to evade the theorem of the previous section and exist at generic filling.

\subsection{Emergent symmetry group}
\label{subsec:loop_group}
The first step is to identify the emergent symmetry group $\GIR$. We invoke the following well-known property of Fermi liquids: non-forward scattering terms are irrelevant in the RG sense at low energies, so the quasiparticle number at \emph{each point} on the Fermi surface is separately conserved at low energies. Thus, Fermi liquids have a very large emergent symmetry group \cite{Haldane_0505}. Roughly, we can say that $\GIR$ = ``$\mathrm{U}(1)^{\infty}$''. However, let us be a bit more precise about how one approaches the ``$\infty$''.

We parameterize the Fermi surface by a continuous parameter $\theta$ (we do not require that $\theta$ literally represents a geometrical  angle), which is periodic, i.e.\ it lives on a circle. Imagine that we place an IR cutoff on the system (that is, place it in finite volume), which since the Fermi surface exists in momentum space, corresponds to a short-distance cutoff on $\theta$, i.e.\ $\theta$ now takes discrete values. To each such $\theta$ value, we associate a $\mathrm{U}(1)$ emergent symmetry generated by an integer-valued operator $\hat{N}_\theta$. Hence, a general symmetry operator will take the form
\begin{equation}
\exp\left(-i \sum_\theta f_\theta \hat{N}_\theta \right)
\end{equation}
where we identify $f_\theta \sim f_\theta + 2\pi$.

Now we send the IR cutoff (the spatial volume) to infinity, which corresponds to sending the spacing between discrete $\theta$ values to zero. What we want is to consider symmetry operators that do not depend too sensitively on the precise way in which the short-distance cutoff in $\theta$ gets sent to zero. In order to achieve this, we require that, in this continuum limit, the $f_\theta$ parameters become \emph{smooth} functions $f(\theta)$. Therefore, in the limit, the emergent symmetries are in one-to-one correspondence with smooth functions from the circle into $\UU(1)$. The group of all such functions is called the \emph{loop group} \cite{LoopGroupsBook} of $\UU(1)$ and we denote it by $\LU(1)$. Hence we conclude that $\GIR = \LU(1)$.

We emphasize that the group structure of $\LU(1)$ differs somewhat from naive conceptions of what ``$\UU(1)^{\infty}$'' would mean; in particular, $\LU(1)$ has only \emph{one} $\UU(1)$ subgroup, whose elements correspond to taking $f(\theta)$ to be a constant function. Physically, this is because only the \emph{total} charge on the Fermi surface is quantized to be an integer; there is no well-defined concept of the (quantized) charge at a single point on the Fermi surface, only of the linear charge density with respect to $\theta$. Accordingly,
we can represent the loop group $\LU(1)$ formally by introducing a density operator $\hat{n}(\theta)$ such that the number of quasiparticles between $\theta$ and $\theta + d\theta$ is measured by $\hat{n}(\theta) d\theta$. Technically, $\hat{n}(\theta)$ is not really an operator in itself, but an operator-valued distribution which should be integrated against a test function. The elements of the emergent symmetry group can be expressed as
\begin{equation}
\label{eq:element}
\exp\left(-i\int f(\theta) \hat{n}(\theta) \right),
\end{equation}
where $f(\theta)$ is any smooth function of $\theta$. Note that, because we identify $f(\theta) \sim f(\theta) + 2\pi$, we are allowed to consider functions $f$ with non-trivial winding number around the circle, such that $\frac{1}{2\pi} \int \partial_\theta f(\theta) d\theta$ is any integer.

\subsection{The homomorphism $\GUV \to \GIR$}
Now it should be clear how to embed the microscopic symmetries into the emergent symmetry group $\GIR = \mathrm{LU}(1)$. Indeed, if $\hat{Q}$ is the generator of the microscopic $\mathrm{U}(1)$ symmetry, we have
\begin{equation}
\label{eq:charge_embedding}
\hat{Q} \sim q \int \hat{n}(\theta) d\theta,
\end{equation}
where the integer $q$ is the charge of a Landau quasiparticle. Of course, for a Fermi liquid of electrons, $q=1$, but in principle one can imagine Fermi liquid-like states where the quasiparticles carry a different charge. (For example, the quasiparticles could be bound states of an odd number of electrons).

Meanwhile, if $k_x(\theta)$ and $k_y(\theta)$ represent the components of the lattice momentum of the point on the Fermi surface parameterized by $\theta$, then we have
\begin{equation}
\label{eq:translation_embedding}
\mathbb{T}_{\alpha} \sim \exp\left(-i\int k_\alpha(\theta) \hat{n}(
\theta) d\theta \right),
\end{equation}
where $\alpha = x,y$, and $\mathbb{T}_x$, $\mathbb{T}_y$ are the lattice translation operators.

\subsection{The 't Hooft anomaly}
\label{subsec:fermi_thooft}
Now we are in a position to discuss the 't Hooft anomaly for the emergent $\LU(1)$ symmetry. As usual, the 't Hooft anomaly can be understood by coupling to a background gauge field for the symmetry. But first, we must ask, what is a gauge field for an $\mathrm{LU}(1)$ symmetry? Since, roughly speaking, an $\mathrm{LU}(1)$ symmetry means there is a $\mathrm{U}(1)$ symmetry for each point on the circle, we can naively say that a gauge field for an $\mathrm{LU}(1)$ symmetry should be a space-time vector field $A_\mu(\theta)$ for each point $\theta$ on the circle, with gauge transformations parameterized by a scalar field $\lambda(\theta)$, and gauge transformation
\begin{equation}
\label{eq:Aspatialgauge}
A_\mu(\theta) \to A_\mu(\theta) + \partial_\mu \lambda(\theta)
\end{equation}
(where the derivative is respect to spacetime, not $\theta$).
Moreover, we require that $A_\mu(\theta)$ and $\lambda(\theta)$ be smooth functions of $\theta$.

However, there is in fact an additional ingredient that is required, that is unique to loop groups. To see this,
we can appeal to the quasiparticle picture of the Fermi liquid. The spatial components $A_i(\theta)$ describe the Aharanov-Bohm phase (which can be interpreted as a Berry phase) picked up as a spatially localized quasiparticle, localized near position $\theta$ on the Fermi surface, is transported in space. But we can also keep a quasiparticle fixed in real space, and transport it along the Fermi surface in momentum space. Therefore, the gauge field needs an additional component $A_\theta$ to describe the Berry phase associated with such a process.
We emphasize, though, that even going beyond Fermi liquids, $A_\theta$ represents an intrinsic part of what it means to couple to a gauge field for an $\LU(1)$ symmetry, regardless of any quasiparticle picture. We will return to this point in Section \ref{sec:quasi_fermi}.

The $A_\theta$ component of the gauge field transforms under gauge transformations as $A_\theta \to A_\theta + \partial_\theta \lambda$. Therefore, if we now combine the vector field $A_\mu$ with $A_\theta$, we obtain a vector field $A$ on the $D + 1$-dimensional manifold $M \times S^1$, where $M$ is the space-time manifold, $D = \mathrm{dim} M$ is the space-time dimension, and $S^1$ is the circle on which the $\theta$ variable lives. 
Moreover, $A$ transforms under gauge transformations precisely as would a $\UU(1)$ gauge field on $M \times S^1$,  Hence we arrive at our conclusion: \emph{An $\LU(1)$ gauge field on $M$ is equivalent to a $\UU(1)$ gauge field on $M \times S^1$.}
From this point on, we will denote this combined vector field as $A_\mu$, taking the convention that indices such as $\mu$ vary both over space-time directions and the $\theta$ direction.

Now we can discuss 't Hooft anomalies. One way to characterize an 't Hooft anomaly in two spatial dimensions is in terms of the topological term describing the response of the corresponding SPT in three spatial dimensions to background gauge fields. [As usual, the relation is that on a (3+1)-D spacetime $M_+$ whose boundary is a (2+1)-D manifold $M$, the SPT on $M_+$ gives rise to the 't Hooft anomaly on the space-time $M$ by anomaly inflow.]
But from the above discussion, a topological term for $\LU(1)$ gauge fields on a (3+1)-D spacetime $M_{+}$ is equivalent to a topological term for $\UU(1)$ gauge fields on the (3+1+1)-D manifold $M_{+} \times S^1$. 
 Hence, we conclude that the appropriate topological term is the 5D Chern-Simons action
\begin{equation}
\label{eq:5d_chern_simons}
S[A] = \frac{m}{24\pi^2} \int_{M_{+} \times S^1} A \wedge dA \wedge dA,
\end{equation}
where the level $m$ is quantized to be an integer\footnote{To see the quantization, observe that the total charge carried by the SPT, under the $\UU(1)$ subgroup of $\LU(1)$ corresponding to setting the function $f(\theta)$ in \eqnref{eq:element} to be independent of $\theta$, is given by the integral (over a fixed time-slice)
\[ 
\int j^{0} = \int \frac{\delta S}{\delta A_0} = \frac{m}{8\pi^2} \int dA \wedge dA = m C[A],
\]
where $C[A]$ is the second Chern number and is quantized to be an integer. Hence, $m$ must be an integer in order for the total charge of the SPT to always be an integer.}.
 Below, we will show that such a topological term, with $m= \pm 1$, indeed reproduces many known properties of spinless Fermi liquids. (Note that the sign of $m$ is only defined relative to a choice of orientation for the Fermi surface, because redefining $\theta \to -\theta$ sends $m \to -m$).

From the topological term \eqnref{eq:5d_chern_simons} we can determine the anomaly equation on the boundary by computing the current $j = \frac{\delta{S}}{\delta{A}}$ and then considering the current incident onto the boundary. We find that, on the boundary, the continuity equation is violated according to
\begin{equation}
\label{eq:theonetrueanomalyeq}
\partial_{\mu} j^{\mu} = \frac{m}{8\pi^2} \epsilon^{\lambda \sigma \tau \kappa} (\partial_\lambda A_\sigma) (\partial_\tau A_\kappa).
\end{equation}
 Here the current $j$ depends both on space-time coordinates and on $\theta$. Its space-time components describe the space-time current of the charge at position $\theta$ on the Fermi surface, i.e.\ of the symmetry generated by $\hat{n}(\theta)$. However, recalling that indices are supposed to vary over the $\theta$ direction as well as space-time directions, we have to introduce the  component $j^{\theta} = \frac{\delta S}{\delta A_\theta}$, which describes flow of charge \emph{along} the Fermi surface. An example of a case where $j^{\theta} \neq 0$ is a Fermi liquid in a magnetic field, described in the next subsection.

Finally, let us mention that there is an alternative picture to understand inflow of the 't Hooft anomaly that is sometimes  helpful. Instead of considering a system with $\LU(1)$ symmetry in a (3+1)-D space-time $M_+$ with boundary, we can consider a system with $\mathrm{L}^{\mathbb{T}^2} \UU(1)$ symmetry in a (2+1)-D space-time $M$ without boundary. Here $\mathrm{L}^{\mathbb{T}^2} \UU(1)$ is the group of smooth maps from the Brillouin zone (thought of as a torus $\mathbb{T}^2$) into $\UU(1)$. In other words, we imagine that the charge is conserved not just at each point on the Fermi surface, but also at \emph{every} $k$-point in the whole Brillouin zone. This is the case, for example, in a non-interacting Fermi gas.  In an interacting Fermi liquid, the physical Hamiltonian presumably has nonzero scattering rate for quasiparticles in the interior of the Fermi surface, but we can, \emph{theoretically}, imagine an extension of the quasiparticle Hamiltonian from the vicinity of the Fermi surface to its interior in a manner that preserves conservation of quasiparticle number at every $k$-point. This is a familiar construction in Fermi liquid theory. (The point is that the physics on the Fermi surface will ultimately not depend on the precise form of the Hamiltonian in the interior.)

Then by similar arguments to before, a gauge field for the $\mathrm{L}^{\mathbb{T}^2} \UU(1)$ symmetry is equivalent to a $\UU(1)$ gauge field on $M \times \mathbb{T}^2$. Since this is also a 5-dimensional manifold, we can write a similar Chern-Simons term to \eqnref{eq:5d_chern_simons}.
Specifically, we write
\begin{equation}
\label{eq:5d_chern_simons_alt}
S[A] = \frac{m}{24\pi^2} \int_{M \times \mathcal{D}} A \wedge dA \wedge dA,
\end{equation}
where $\mathcal{D} \subseteq \mathbb{T}^2$ is the volume ``occupied by electrons'', that is, the volume enclosed by the Fermi surface\footnote{By a particle-hole transformation which cannot modify the physics, the ``occupied volume'' could equivalently be taken to be the complement $\mathcal{D}^c$ of $\mathcal{D}$ in the Brillouin zone. Observe, however, that if we replaced $\mathcal{D}$ with $\mathcal{D}^c$ in \eqnref{eq:5d_chern_simons_alt}, the only change in the anomaly equation \eqnref{eq:theonetrueanomalyeq} on the Fermi surface would be an additional minus sign, which as we mentioned earlier just corresponds to a choice of orientation of the Fermi surface.}.
 This gives rise to the same anomaly equation \eqnref{eq:theonetrueanomalyeq} on the boundary $\partial(M \times \mathcal{D}) = M \times \partial \mathcal{D}$, where $\partial \mathcal{D}$ is the Fermi surface.
 Note that in this interpretation, the components $A_0, A_x, A_y$ of $A_\mu$ are the usual ones while the components $A_{k_x}, A_{k_y}$ are $k$-space gauge fields. Then the $5$ dimensional manifold has the interpretation of being ``phase space'' (that is, the space where points are labeled by position and momentum) plus time,  and $A_\mu$ can be considered a gauge field in phase space. A term similar to \eqnref{eq:5d_chern_simons_alt} has previously appeared in Ref.~\cite{Qi_0802}.

\subsection{Filling constraint and Luttinger's theorem}
\label{subsec:fermi_monopoles}
We consider a theory in two spatial dimensions whose 't Hooft anomaly is canceled by inflow from \eqnref{eq:5d_chern_simons}. According to the discussion of Section \ref{subsec:2d_filling_constraint}, we first need to determine how the emergent symmetry $\LU(1)$ gets represented projectively in such a theory in the presence of a $2\pi$ flux (of the microscopic charge $\UU(1)$ symmetry). This is something that can be derived from the 't Hooft anomaly that was characterized in the previous subsection. We show in Appendix \ref{appendix:monopole} that this leads to a projective representation on a $2\pi$ flux described by the commutation relations
\begin{equation}
\label{eq:kac_moody}
[\hat{n}(\theta), \hat{n}(\theta')] = -i \frac{mq}{2\pi} \delta'(\theta - \theta'),
\end{equation}
where $\hat{n}(\theta)$ are the operators introduced in Section \ref{subsec:loop_group}, and $\delta'$ is the derivative of the Dirac delta function.

 Note that \eqnref{eq:kac_moody} is precisely the so-called Kac-Moody algebra satisfied by the local density operator of a chiral fermion in one spatial dimension [such as appears, for example, at the boundary of a integer quantum Hall state in two spatial dimensions]. Indeed, this relationship is not a coincidence; one can see roughly how it comes about in the case of a Fermi liquid at the level of the semiclassical theory of electron transport. Suppose, for simplicity, that we switch off the interactions so that the Fermi liquid becomes a non-interacting Fermi gas. 
 In the presence of a spatially-dependent magnetic field $\textbf{B}(\textbf{x})$, the semiclassical equations of motion take the form \cite{Chang_9511}
\begin{align}
\frac{d\textbf{k}}{dt} &= -q \textbf{B}(\textbf{x}) \times \textbf{v}(\textbf{k}) \label{eq:semiclassical_first}\\
\frac{d\textbf{x}}{dt} &= \textbf{v}(\textbf{k}) - \Omega(\textbf{k}) \times \frac{d\textbf{k}}{dt},
\end{align}
where $\Omega(\textbf{k})$ is the Berry curvature of the Bloch states, $\textbf{v}(\textbf{k}) = \frac{\partial \mathcal{E}(\textbf{k})}{\partial \textbf{k}}$, $\mathcal{E}(\textbf{k})$ is the dispersion relation, and we have omitted terms beyond linear order in the magnetic field strength. In particular, \eqnref{eq:semiclassical_first} implies that electrons develop a circulation in momentum space along contours of constant energy, in particular along the Fermi surface. In two spatial dimensions the Fermi surface is one-dimensional and the circulation along the Fermi surface is unidirectional (set by the sign of the magnetic field), which indeed resembles a chiral fermion.

This intuitive argument, however, does not fix the coefficient of the right-hand side of \eqnref{eq:kac_moody}. 
 In Appendix \ref{appendix:semiclassical} we give a more careful derivation of \eqnref{eq:kac_moody} from the semiclassical theory of electron transport, confirming that $m= \pm 1$ for a spinless Fermi liquid. Related expressions have previously been derived in Refs.~\cite{Golkar_1602,Barci_1805}.

From \eqnref{eq:kac_moody} we can compute the projective representation of the translation symmetry in light of the embedding \eqnref{eq:translation_embedding}. We find
\begin{align}
\mathbb{T}_x \mathbb{T}_y \mathbb{T}_x^{-1} \mathbb{T}_y^{-1} &= \exp\left(i \frac{mq}{2\pi} \int k_x(\theta) \frac{dk_y(\theta)}{d\theta} d\theta\right) \\
&= \exp\left(i m q \frac{\mathcal{V}_F}{2\pi}\right),
\end{align}
where $\mathcal{V}_F$ is the volume in momentum space enclosed by the Fermi surface. (Here we have chosen a particular convention to define the orientation of the Fermi surface).
Hence, from \eqnref{eq:filling_constraint_2d}, we conclude that
\begin{equation}
\label{eq:luttingers_theorem}
\nu = m q \frac{\mathcal{V}_F}{(2\pi)^2} \quad (\mathrm{mod}\,1),
\end{equation}
which (if we set $m=q=1$) is precisely Luttinger's theorem for a spinless Fermi liquid in two spatial dimensions.

We wish to emphasize, however, that, in general, Luttinger's theorem \eqnref{eq:luttingers_theorem} follows directly from \eqnref{eq:kac_moody}, which in turn follows directly from the 't Hooft anomaly. It was not necessary to assume anything about the dynamical properties of the Fermi liquid, e.g. the existence of quasiparticles. Thus, Luttinger's theorem also holds (with a possible integer multiplicative factor $mq$) for any IR theory that has the same emergent symmetry as a Fermi liquid. 

It is interesting to reconsider \eqnref{eq:kac_moody} and its relation with the $5D$ Chern-Simons term from the viewpoint that the 5-dimensional manifold can be thought of as $4$-dimensional phase space together with the time direction (see the last two paragraphs of Section \ref{subsec:fermi_thooft}). In the presence of a static $2\pi$-strength magnetic flux in the $x, y$ components of $A_\mu$ (with the corresponding components $A_x$ and $A_y$ independent of $k_x$,$k_y$,$t$), the $5D$ Chern-Simons term \eqnref{eq:5d_chern_simons_alt} reduces to a $3D$ Chern-Simons term for the remaining components $A_I = (A_0, A_{k_x}, A_{k_y})$, assuming that they are independent of $x$ and $y$:
\begin{equation}
\label{eq:3dcs}
S_{3D} =  \frac{m}{4\pi} \int d^3x \epsilon_{IJK} A_I \partial_J A_K
\end{equation} 
where $k_x$ and $k_y$ are integrated over the volume $\mathcal{D}$, i.e. over the interior of the Fermi surface.

Now the claim is that on the Fermi surface, there is a chiral mode carrying the Kac-Moody algebra
\begin{equation}
\label{eq:kac_moody_mod}
[\hat{n}(\theta), \hat{n}(\theta')] = -i \frac{m}{2\pi} \delta'(\theta - \theta'),
\end{equation}
In particular since a $2\pi$ flux of the microscopic $\UU(1)$ corresponds to a $2\pi q$ flux of $A$ by \eqnref{eq:charge_embedding}, we recover
\eqnref{eq:kac_moody}.

In the phase space interpretation, \eqnref{eq:3dcs} describes an integer quantum Hall effect in {\em momentum space} in the interior of the Fermi surface. In other words we think of the rigid interior of the Fermi surface as hosting an integer quantum Hall state in momentum space when we apply a $2\pi$ flux in real space. The Fermi surface is the boundary in momentum space of the interior, and hence hosts a chiral edge state. 

\subsection{Extension to higher dimensions}
The description of the anomaly extends straightforwardly to higher dimensional Fermi liquids. Indeed, if $M$ is the $(d+1)$-D space-time manifold, and $F$ is a $(d-1)$-dimensional manifold parameterizing the Fermi surface, which is a codimension-1 surface in the $d$-dimensional Brillouin zone, then the emergent symmetry group is $\mathrm{L}^F \mathrm{U}(1)$, i.e. the group of smooth maps from $F$ to $\UU(1)$.
Including the components of the Berry connection on the Fermi surface promotes the $\mathrm{L}^F\mathrm{U}(1)$ gauge field to  a $\UU(1)$ gauge field on the $(2d)$-dimensional manifold $M \times F$. 

Then we can write down a topological term on a $(d+2)$-D spacetime $M_{+}$ describing the SPT whose inflow generates the anomaly of the Fermi liquid, given by the $(2d+1)$-dimensional Chern-Simons action:
\begin{equation}
\label{eq:fermi_anomaly_general_dimension}
S[A] = \frac{m}{(d+1)!(2\pi)^{d}}\int_{M_{+} \times F} A \wedge \underbrace{(dA \wedge \cdots \wedge dA)}_{\text{$d$ times}}.
\end{equation}

\section{Ersatz Fermi liquids and their phenomenology}
\label{sec:quasi_fermi}
As we showed previously (Section \ref{sec:generic}), \emph{any} IR theory that can exist at generic filling must have a very large symmetry group, ``larger'' than any compact Lie group. In the simplest cases, this emergent symmetry group will be the same as that of a Fermi liquid, i.e. $\GIR = \mathrm{LU}(1)$. We refer to an IR theory with this emergent symmetry as an \emph{ersatz Fermi liquid (EFL)}. 
We note that EFLs represent a \emph{class} of theories; there could be many different EFLs with different dynamical properties.
Here we will examine the properties that all EFLs have in common. In particular, as a  Fermi liquid is an example of an EFL, this section will provide a fresh perspective on many aspects of Fermi liquid phenomenology, showing that they arise directly from the emergent symmetry and its 't Hooft anomaly without any need to invoke the detailed dynamical properties of a Fermi liquid.

Going beyond Fermi liquids, an important example  of states which we expect will be EFLs are associated with quantum critical points in metals which are not tied to a particular electron density.  As a concrete example consider a putative quantum critical point associated with the onset of Ising nematic order\footnote{This is associated with spontaneously  breaking $C_4$ lattice rotation symmetry to $C_2$.}  in a metal in $d = 2$.  There is no particular electron density at which this transition will happen in any given system. Indeed, as microscopic parameters are changed continuously, we expect that the electron density at the transition will also change continuously without a change of universality class.  The universal critical properties of this transition are described by a theory of electrons near the Fermi surface coupled to the fluctuating order parameter modes.  (A model with similar structure also describes {\em insulating } quantum spin liquid phases with a spinon Fermi surface coupled to a dynamical $U(1)$ gauge field).   In these models the resulting IR fixed point is not a Fermi liquid\cite{holstein1973haas,lee2006doping,halperin1993theory,polchinski1994low,nayak1994non,altshuler1994low,lee2009low,metlitski2010quantum,mross2010controlled,schattner2016ising,dalidovich2013perturbative,lee2018recent}.  Nevertheless, from our point of view, these metallic quantum critical points are expected to be EFLs.  (We will discuss  neutral Fermi surfaces in insulators separately in Sec.~\ref{subsec:spinonFS}.)

We may understand why such critical points should be EFLs as  follows. In previous papers \cite{polchinski1994low,lee2009low,metlitski2010quantum,mross2010controlled} the ultimate IR fixed point was accessed  through a `patch construction' which begins by by dividing the Fermi surface into small patches.  It was then argued that  the important coupling of the fermions within a single patch is to  
boson fluctuations with momentum tangential to the local Fermi surface. This enables treating the full system by focusing on a pair of antipodal patches and ignoring their coupling to other such pairs of patches. The patch width is taken to zero at the end. In this scheme, the number of fermions within each patch is conserved. The assumption that this patch description captures the IR fixed point then implies that the linear charge density $\hat{n}(\theta)$ at each point of the Fermi surface is conserved at the fixed point.  (There is, potentially, a dangerous inter-patch BCS coupling in the pairing channel that could destroy these conservation laws.  For the Fermi surface coupled to a gauge field,  a weak pairing interaction is irrelevant at the IR fixed point\cite{metlitski2015cooper} while in the Ising nematic quantum critical point it is relevant. In the latter case the non-Fermi liquid metallic fixed point is preempted by the superconducting instability. Our discussion then applies at a scale above this instability.)

 We also expect that more complex quantum critical points 
associated with the death of a Fermi surface can be subsumed under the umbrella of EFLs.  Such critical points have been argued\cite{senthil2008critical} to possess a critical Fermi surface even in the absence of Landau quasiparticles.  As we will see below, the EFL description, if it indeed applies to such quantum critical points, enables inferring many of their general properties. 

\subsection{General properties of EFLs}
Let us return to a general EFL. Then, translation symmetry must embed into $\GIR$ somehow: that is, we have
\begin{equation}
\mathbb{T}_\alpha \sim \exp\left(-i\int k_\alpha(\theta)  \hat{n}(\theta) d\theta \right).
\label{eq:quasifermi_translation}
\end{equation}
for some $\mathrm{U}(1)$-valued functions $k_\alpha(\theta)$. The values of $k_\alpha(\theta)$ can be interpreted as momenta in the Brillouin zone, so this defines a codimension-1 surface in the Brillouin zone. We can take this to be the \emph{definition} of a Fermi surface in a general EFL. Of course, this ``Fermi surface'' may or may not have any signature in, say, angle-resolved photoemission spectroscopy (ARPES) measurements, and if it does the signature might be different in character from that of a Fermi liquid.

 Nevertheless, the fact that $\hat{n}(\theta)$ is a conserved operator for every $\theta$ tells us that, at any point $\theta$ on the Fermi surface, there are infinitely long-lived excitations which are forbidden from scattering away from that point. To see this, first observe that since $\hat{n}(\theta)$ is conserved, excitations can be labeled by their corresponding eigenvalue $n(\theta)$. An excitation at point $\theta$ on the Fermi surface is characterized by
\begin{equation}
\label{eq:fermi_surface_quantum}
n(\theta') = N \delta(\theta' - \theta),
\end{equation}
for some $N$. One can show that $N$ is quantized to an integer, because by the definition of $\mathrm{LU}(1)$, the operator
\begin{equation}
\hat{N} = \int \hat{n}(\theta) d\theta
\end{equation}
must have integer eigenvalues, since $\exp(-2\pi i \hat{N})$ is the identity operator. Henceforth, we will refer to such excitations as \emph{Fermi surface quanta}. From \eqnref{eq:quasifermi_translation} and \eqnref{eq:fermi_surface_quantum}, we see that a single Fermi surface quantum carries momentum $\mathbf{k}(\theta)$.
 Fermi surface quanta are the generalization of Landau quasiparticles to a general ersatz Fermi liquid. Their dynamics, however, could be very different from Landau quasiparticles. Moreover, apart from the fact that their number is quantized, the Fermi surface quanta need not have particularly ``quasiparticle-like'' properties. Nevertheless, the Berry's phase of a spatially localized Fermi surface quantum as it is moved in space or over the Fermi surface is still a well-defined quantity, which supplies the general interpretation of the gauge field $A_\mu$ (including the $A_\theta$ component) discussed in Section \ref{subsec:fermi_thooft}.

Next, we can consider how the microscopic $\mathrm{U}(1)$ symmetry embeds into $\mathrm{LU}(1)$. The quantization of charge, i.e.\ the requirement that $e^{i(2\pi) \hat{Q} }= 1$ (where $\hat{Q}$ is the microscopic charge operator), constrains the embedding to be of the form
\begin{equation}
\hat{Q} \sim q \hat{N},
\end{equation}
for some integer $q$. We can interpret $q$ as the charge of a single Fermi surface quantum.

Finally, we can consider the 't Hooft anomaly of the $\mathrm{LU}(1)$ symmetry. There is not much freedom, since as we saw in Section \ref{sec:structural_fermi_liquid} the 't Hooft anomalies are just classified by the level $m$ of the 5D Chern Simons term \eqnref{eq:5d_chern_simons} which is quantized to be an integer. In summary, once we have fixed the shape of the Fermi surface, the  kinematic properties of an ersatz Fermi liquid (in the sense of Section \ref{sec:structural_properties}) are captured by the integers $q$ and $m$. (Note that there is a  freedom that sends $q \to -q$ and $m \to -m$ simultaneously by redefining the generators of $\LU(1)$; moreover, as we mentioned earlier, choosing the reverse orientation of the Fermi surface sends $m \to -m$ while leaving $q$ fixed).

Next we observe the arguments of Section \ref{subsec:fermi_monopoles} can be applied in any EFL. Therefore, we immediately conclude that the Fermi surface in any EFL satisfies Luttinger's theorem in the form \eqnref{eq:luttingers_theorem}.

In the remainder of this section we will consider various aspects of Fermi liquid phenomenology and argue that they hold equally well in any EFL.

\subsection{Response to electric fields}
\label{subsec:electric_field}
A property of Fermi liquids is that if a uniform electric field is applied, the system responds in essentially the same way as it would in a non-interacting Fermi gas, which is to say that the momenta of quasiparticles get shifted (assuming the electric field is in the $x$ direction) according to $k_x \to k_x + a_x E t$, where $a_x$ is the unit cell size in the $x$ direction\footnote{This appears because we have normalized our momenta to be dimensionless.}. This causes the density operators $\hat{n}(\theta)$ to be no longer conserved in the presence of the electric field. Specializing for simplicity to a Fermi liquid in two spatial dimensions, the total charge on a segment $[\theta,\theta+d\theta]$ of the Fermi surface gets shifted according to
\begin{equation}
\label{eq:nshift}
\hat{n}(\theta) d\theta \to \hat{n}(\theta) d\theta + q a_x E t \frac{dk_y(\theta)}{(2\pi)^2/(L_x L_y)}
\end{equation}
where $dk_y(\theta) = k_y(\theta + d\theta) - k_y(\theta)$, and $L_x$ and $L_y$ are the linear dimensions of the system (normalized by $a_x$ and $a_y$ respectively, the unit cell dimensions) in the $x$ and $y$ direction. We have to divide by the denominator in the second term in \eqnref{eq:nshift} to take into account the density of single-particle states in momentum space in a finite-size system.

We will now show that \eqnref{eq:nshift} indeed follows from the 't Hooft anomaly. In order to do that, we will want to reformulate \eqnref{eq:nshift} in a way that removes the explicit dependence on $L_y$. We write \eqnref{eq:nshift} as
\begin{equation}
\label{eq:some_equation}
\frac{d}{dt} \hat{n}(\theta) = \eta(\theta; L_y) := qE \frac{a_x L_x L_y}{(2\pi)^2} \frac{dk_y(\theta)}{d\theta}
\end{equation}
Next we identify the \emph{difference}
\begin{align}
\Delta \eta (\theta) &:= \eta(\theta; L_y + 1) - \eta(\theta; L_y) \\ &= qE \frac{a_x L_x}{(2\pi)^2} \frac{d k_y(\theta)}{d\theta}
\label{eq:shiftratetransflux}
\end{align}
with the density shift rate associated with applying a uniform electric field in the $x$ direction in the presence of a ``flux of $y$ translation'' symmetry around the $y$ direction.

To compute this shift from the 't Hooft anomaly, we start with the anomaly equation \eqnref{eq:theonetrueanomalyeq} and consider a configuration where
\begin{align}
A_x &= E t, \\
A_y &= -k_y(\theta) / L_y, \\
A_t &= 0 \\
A_{\theta} & \mbox{\, is independent of $x,y,t$}
\end{align}
The motivation for our choice of $A_y$ is that it ensures that the $\LU(1)$ flux around the $y$ direction is given by
\begin{equation}
\exp\left(-i\int k_y(\theta) \hat{n}(\theta) d\theta\right),
\end{equation}
which coincides with the expression \eqnref{eq:translation_embedding} for the translation symmetry operator $\mathbb{T}_y$. Substituting into \eqnref{eq:theonetrueanomalyeq} and integrating over $x$ and $y$ gives
\begin{equation}
\partial_t n(\theta) + \partial_\theta J^\theta = \frac{mq}{(2\pi)^2} E \frac{dk_y(\theta)}{d\theta} a_x L_x,
\end{equation}
where $J^{\theta} = \int j^{\theta} dx dy$, and we can identify $\int j^{t} dx dy$ with the expectation value $n(\theta) = \langle \hat{n}(\theta) \rangle$. Hence, provided that $\partial_\theta J^{\theta} = 0$, we recover \eqnref{eq:shiftratetransflux} if we set $m=1$.

Our assumption that $\partial_\theta J^{\theta} = 0$, even in the presence of a background electric field, requires a bit more explanation. $J^\theta$ represents a charge circulation along the Fermi surface. As we mentioned in Section \ref{subsec:fermi_monopoles}, in a Fermi liquid a magnetic field induces a chiral flow of quasiparticles along the Fermi surface, in which case $J_\theta \propto B n(\theta)$ and generally $\partial_\theta J_\theta \neq 0$. So we need to address why the same thing could not happen with an electric field. Of course, in a Fermi liquid, one can easily convince oneself that it does not, but we want an argument that holds more generally in any EFL.

In a quantized formulation of the IR theory, the current operator $\hat{J}_\theta$ is defined by
\begin{equation}
\label{eq:J_in_terms_of_H}
\hat{J}_\theta = \int \frac{\delta \hat{H}}{\delta A_\theta}  d^2 x
\end{equation}
where $\hat{H}$ is the Hamiltonian of the IR theory. Then we know that, in the absence of an electromagnetic field, $\hat{J}_\theta$ must satisfy (as an operator identity) the conservation law
\begin{equation}
\frac{d}{dt} \hat{n}(\theta,t) + \partial_\theta \hat{J}_\theta(\theta,t) = 0
\end{equation}
On the other hand, in the absence of an external electromagnetic field,  by assumption $\hat{n}(\theta)$ should be conserved for each $\theta$, since it is the generator of the emergent symmetry. Therefore, in the absence of electromagnetic field, the current operator identically satisfies 
\begin{equation}
\label{eq:partialJ}
\partial_\theta \hat{J}_\theta = 0.
\end{equation}
 Now observe that because the electric field $E$ enters into the Hamiltonian ``temporally'', i.e.\ through the time-derivative of $A$, it is not possible for $\hat{J}_\theta$ defined by \eqnref{eq:J_in_terms_of_H} to depend on $E$ because the Hamiltonian is defined on a single time slice of space-time. Therefore, $\hat{J}_\theta$ continues to satisfy \eqnref{eq:partialJ} even with an applied electric field. On the other hand, it \emph{is} possible for $\hat{J}_\theta$ to depend on the applied magnetic field, as the Fermi liquid example demonstrates.

\subsection{Quantum oscillations}
\label{subsec:quantum_oscillations}
Recall that Fermi liquids display ``quantum oscillations'' when a weak magnetic field $B$ is applied; that is, various physical properties are periodic in $1/B$ \cite{QuantumOscillationsBook}.
In the case that the physical property under consideration is resistivity, for example, this is known as the Shubnikov-de Haas effect. (What we mean by ``weak magnetic field'' is that the magnetic flux per unit cell should be much less than 1. To observe the oscillations at finite temperature $T$, it is also necessary that the cyclotron energy $E_c$ should satisfy $E_c \gtrsim T$, which for $T > 0$ also places a lower bound on the magnetic fields for which the oscillations are observable).

It is sometimes  stated that observing quantum oscillations is evidence for a quasiparticle description, i.e.\ of a Fermi liquid. Here, however, we will show that there are very general reasons to expect any EFL to display the same periodicity of the quantum oscillations. We will not, however, make any statement about the \emph{amplitude} of the quantum oscillations, and it remains possible that this amplitude and its dependence on parameters such as temperature will still allow for Fermi liquids and non-Fermi liquids to be distinguished.

Our task is complicated by the fact that the quantum oscillations are non-perturbative in the magnetic field, and therefore, it is not clear that one expects them to be directly describable in terms of the deep IR theory. Instead, we will find it necessary to make appeal to a UV completion. 

\subsubsection{Two spatial dimensions}
Let us first consider the case of a system in two spatial dimensions which microscopically has continuous translation symmetry, with a microscopic charge density $\rho$. In such a case, in the presence of a magnetic field the $x$ and $y$ translation generators $P_x$ and $P_y$ fail to commute, but we can define a ``magnetic unit cell'' of volume $b_x \times b_y$, such that the flux per magnetic unit cell is $2\pi$ and the discrete translations $\mathbb{T}_x = \exp(-ib_x P_x)$ and $\mathbb{T}_y = \exp(-ib_y P_y)$ do commute. Then we can define the magnetic filling $\nu_M = (b_x b_y) \rho = 2\pi\rho/B$, which is the charge per magnetic unit cell. Then we can apply all the usual results on filling with respect to this discrete translation symmetry; in particular, if $\nu_M$ is not an integer then the Lieb-Schultz-Mattis-Oshikawa-Hastings theorem \cite{Lieb_1961,Oshikawa_9610,Oshikawa_9911,Hastings_0411} forbids the system from having a trivial (i.e. not topologically ordered) gapped ground state, whereas such a state is permitted for integer $\nu_M$. Moreover, if we assume that the ground state at a given $\nu_M$ is itself an EFL with a single Fermi surface, then Luttinger's theorem for EFLs implies that at integer $\nu_M$ the Fermi volume must fill all of the Brillouin zone corresponding to the magnetic unit cell, i.e.\ the Fermi surface becomes degenerate, presumably leading to an instability. Most generally, the point is that since $\nu_M \, \mathrm{mod}\, 1$ always reflects properties of the IR theory according to our general framework of filling constraints, the nature of the IR theory must vary with $\nu_M \, \mathrm{mod} \, 1$.

These considerations motivate our assumption that, in general, observable properties of the ground state vary periodically with $\nu_M$, with period 1 (The periodicity refers to the behavior for $\nu_M \gg 1$, and where $\nu_M$ varies over an interval $\Delta \nu_M \ll \nu_M$. On longer scales there will be some envelope function governing the amplitude of the oscillations.) This behavior is well known in the Fermi liquid case, where integer $\nu_M$ (non-integer $\nu_M$) corresponds to fully filled (partially filled) Landau levels respectively.
Our periodicity assumption immediately  implies that, if we keep $\rho$ fixed and vary $B$, then the periodicity with respect to $1/B$ is
\begin{equation}
\label{eq:quantum_oscillations_periodicity}
\Delta(1/B) = 1/(2\pi\rho).
\end{equation}

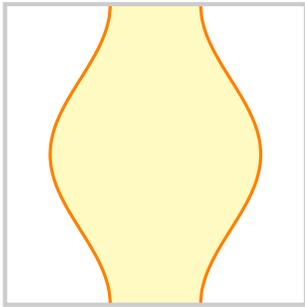
\begin{figure}

\begin{tikzpicture}[scale=4]
\draw[fermisurface,fermifill] plot[smooth,domain=0:360,variable=\t] ({0.1*cos(\t) + 0.25},{\t/360}) -- 
      plot[smooth,domain=360:0,variable=\t] ({-0.1*cos(\t) + 0.75}, {\t/360});
      \draw[brillouinzone] (0,0) -- (0,1) -- (1,1) -- (1,0) -- cycle;
\end{tikzpicture}

\caption{\label{no_oscillations}A Fermi liquid with a Fermi surface that wraps non-trivially around the Brillouin torus does not exhibit quantum oscillations.}
\end{figure}

Next we can consider a system in two spatial dimensions that microscopically only has discrete translation symmetry. Here we recall that, even for Fermi liquids, there is a condition on the Fermi surface in order to observe quantum oscillations. For a Fermi surface as depicted in Figure \ref{no_oscillations}, no quantum oscillations will be observed; physically one can think of this as coming from the fact that the semiclassical orbits of this Fermi surface in a magnetic field are not closed in position space. The general condition for a Fermi liquid to display quantum oscillations is that each connected component of the Fermi surface should have trivial winding number on the Brillouin torus.

Let us therefore examine what we can say about a general EFL that satisfies this condition. To avoid some subtleties in the argument, we will assume that $q=1$.
The key point is that when the Fermi surface does not have nontrivial winding number, then there is a consistent way to define the Fermi momentum $\textbf{k}(\theta)$  without any mod $2\pi$ ambiguity. Then we can define the filling that is relevant for quantum oscillations as
\begin{equation}
\nu_o = \frac{m}{(2\pi)^2} \mathcal{V}_F, \quad \mathcal{V}_F = \int k_x(\theta) \frac{dk_y(\theta)}{d\theta} d\theta,
\end{equation}
where $\mathcal{V}_F$ is the volume enclosed by the Fermi surface, and this formula holds without any $\mathrm{mod} \, 1$ equivalence. Note that $\nu_o$ is not necessarily the same as the microscopic filling $\nu$, although by Luttinger's theorem it must differ from it by an integer. Moreover, we can define emergent symmetry generators
\begin{equation}
\hat{P}_\alpha = a_\alpha^{-1} \int k_\alpha(\theta) \hat{n}(\theta), \\
\end{equation}
with $a_x$, $a_y$ the dimensions of the unit cell, such that $\exp(-ia_\alpha P_\alpha) \sim \mathbb{T}_\alpha$.

This leads us to believe (although we will not give a rigorous proof) that, even though the actual microscopic theory may have had only a discrete translation symmetry, there \emph{exists} a UV completion of the IR theory with continuous translation symmetry whose action on the IR theory is generated by $\hat{P}_\alpha$. Assuming that this is the case,  invoking the continuous translations version of Luttinger's theorem shows that the density of $\UU(1)$ charge in this UV completion must be $\rho = \nu_o/(a_x a_y)$. Then with $\rho$ so defined we again find that the periodicity of quantum oscillations is given by \eqnref{eq:quantum_oscillations_periodicity}.

\subsubsection{Two spatial dimensions: multiple Fermi surfaces}
As a warm-up to going to three spatial dimensions, we will extend the results of the previous section to the case where we have several disconnected components of the Fermi surface, instead of just one. For simplicity we will consider the case 
of two components, although the arguments can easily be generalized.
In that case, the emergent symmetry is $L^F \UU(1)$, i.e. the group of smooth maps from $F$ into $\UU(1)$, where $F = S^1 \sqcup S^1$, i.e. the disjoint union of two circles. We can represent this by writing the generators of the emergent symmetry as $\hat{n}^{(\lambda)}(\theta)$, where $\theta \in S^1$ and $\lambda \in \{1,2\}$ labels the two components of the Fermi surface. Similarly we write the momentum of the Fermi surface as $\textbf{k}^{(\lambda)}(\theta)$.
 We assume that each component satisfies the condition of no winding described earlier, which allows us to define the momentum operator
\begin{equation}
\hat{P}_\alpha = \sum_\lambda a_\alpha^{-1} \int k_\alpha^{(\lambda)}(\theta) \hat{n}^{(\lambda)}(\theta) d\theta.
\end{equation}
The crucial point is that the fact that there are two connected components of the Fermi surface allows us to define two different $\UU(1)$ emergent symmetries, generated by $\hat{Q}^{(1)}$ and $\hat{Q}^{(2)}$, with
\begin{equation}
\hat{Q}^{(\lambda)} = \int \hat{n}^{(\lambda)}(\theta) d\theta,
\end{equation}
In the same spirit as the argument of the previous section, therefore, we will postulate that there is a UV completion with both continuous translation symmetry and $\hat{Q}^{(1)}$ and $\hat{Q}^{(2)}$ realized microscopically. Defining $\rho^{(1)}$ and $\rho^{(2)}$ to be the corresponding charge densities, by similar arguments as earlier we have a generalized Luttinger's theorem
\begin{equation}
\rho^{(\lambda)} = \frac{m}{(2\pi)^2 a_x a_y} \mathcal{V}_F^{(\lambda)}, \quad \mathcal{V}_F^{(\lambda)} = \int k_x^{(\lambda)} (\theta) k_y^{(\lambda)} (\theta) d\theta.
\end{equation}
Now, in the presence of magnetic field, if we define magnetic fillings $\nu_M^{(\lambda)} = 2\pi\rho^{(\lambda)}/B$, we expect that any observable property should be periodic in each of $\nu_M^{(1)}
$ and $\nu_M^{(2)}$ with period 1. In other words, any observable property $O$ can be expressed as
\begin{equation}
O = f(\nu_M^{(1)}, \nu_M^{(2)}),
\end{equation}
where $f(\nu_1, \nu_2)$ is some function such that $f(\nu_1 + 1, \nu_2) = f(\nu_1, \nu_2 + 1) = f(\nu_1, \nu_2)$. This implies that if we vary $B$ while keeping $\rho^{(\lambda)}$ fixed, and assuming that the ratio $\rho^{(1)}/\rho^{(2)}$ is an irrational number, then $O$ varies \emph{quasiperiodically} in $1/B$ with base frequencies $2\pi \rho^{(1)}$ and $2\pi\rho^{(2)}$.

\subsubsection{Three spatial dimensions}
\label{subsubsec:quantum_oscillations_3d}
Now we can consider a 3D system. By dimensional reduction, we can consider a 3D system as a 2D system; but the Fermi surface of the 2D system will consist of infinitely many components corresponding to taking slices through the 3D Fermi surface with fixed $k_z$. By generalizing the previous discussion to $N$ components, and then taking the limit as $N \to \infty$, we find that any observable property should be expressible as
\begin{equation}
O = \mathcal{F}[\nu_\perp],
\end{equation}
where $\mathcal{F}$ is some functional of $\nu_\perp$, and $\nu_\perp$ is a function into $\mathbb{R}/\mathbb{Z}$ (i.e.\ the real line with $\nu \sim \nu + 1$ identified), defined by
\begin{align}
\nu_\perp(k_z) &= 2\pi \rho_{\perp}(k_z) \frac{1}{B} \quad \mathrm{mod} \, 1 \\
\rho_{\perp}(k_z) &= \frac{m\mathcal{V}_{\perp}(k_z)}{(2\pi)^2 a_x a_y}.
\end{align}
where $\mathcal{V}_{\perp}(k_z)$ is the two-dimensional area enclosed by the intersection of the Fermi surface with a plane of fixed $k_z$.

We show in Appendix \ref{appendix:oscillations} that this implies, so long as the functional $\mathcal{F}$ is sufficiently regular, that the dependence of $O$ on $B$ at small $B$ is dominated by the extremal Fermi surface cross-sections, and we obtain that $O$ varies periodically or quasiperiodically with $1/B$, with base frequencies $2\pi \rho_{\perp}(k_z^{*})$, $k_z^{*} \in \Sigma$, where $\Sigma$ is the set of solutions to $\frac{d}{dk_z} \mathcal{V}_{\perp}(k_z) = 0$. This is the same result as for a Fermi liquid.

\subsection{Anomalous Hall effect}
The anomalous Hall effect refers to a Hall conductance that exists in a time reversal broken system in zero external magnetic field.  In general there can be many different effects that contribute to this effect. For a clean free electron system with a Fermi surface, there is an interesting contribution that arises due to the net Berry curvature $\Phi_B$ of the filled Fermi sea.    In $d = 2$, we have 
\begin{equation}
\sigma_{xy} =  \frac{\Phi_B}{2\pi}
\end{equation}
This contribution can alternately be re-expressed as a Fermi surface property in terms of the Berry gauge connection: 
\begin{equation}
\Phi_B = \int d\theta A_\theta
\end{equation}
Note that the integrand (as opposed to the integral) on the right hand side -- while only involving the Fermi surface -- is not by itself invariant under $k$-space gauge transformations. Going beyond non-interacting systems, for an interacting Fermi liquid it has been argued \cite{Chen_1604} that there will be additional Fermi surface contributions to the anomalous Hall conductance that, by contrast to the Berry phase contribution above, involve an integral over gauge invariant quantities defined in local patches of the Fermi surface. 

Next we show that the `t Hooft anomaly of an EFL  directly implies the Berry phase contribution to the  anomalous Hall effect. To that end we start from \eqnref{eq:5d_chern_simons_alt} and consider a configuration where the gauge field $A$ has a flux $\Phi_B$ through the interior of the Fermi surface in the $(k_x, k_y)$ plane, and take the corresponding gauge field components $(A_{k_x}, A_{k_y})$ to be independent of $(t,x, y)$. Similar to the discussion surrounding \eqnref{eq:3dcs}, we now find that the $5D$ Chern-Simons term reduces (assuming that $A_0, A_x, A_y$ are independent of $k_x$ and $k_y$) to:
\begin{equation}
\label{eq:ahaction}
S_{AH} =  \frac{m \Phi_B}{8\pi^2} \int dt dx dy \, \epsilon_{IJK} A_I \partial_J A_K
\end{equation} 
for the $I, J, K = (0, x, y)$. This then directly corresponds to a contribution to the Hall response $\sigma_{xy} = \frac{\Phi_B}{2\pi}$.  We can write this directly in terms of the boundary theory at the Fermi surface as $\Phi_B = \int d\theta A_\theta$. In the free fermion case this is exactly the Berry phase contribution to the anomalous Hall effect discussed above.  In a general EFL, we should regard this as an unavoidable contribution due to the 't Hooft anomaly of the IR theory. Like in the interacting Fermi liquid, the full measured anomalous Hall effect may include other contributions that are `local' on the Fermi surface.  Note that \eqnref{eq:ahaction}  has the structure of  a $3D$ Chern-Simons term but with a coefficient that is not quantized to be an integer multiples of $\frac{1}{4\pi}$. This is allowed here because the boundary gauge fields are coupled to the gapless modes associated with the Fermi surface. The full boundary action that includes both the Fermi surface modes and the unquantized Chern-Simons action of \eqnref{eq:ahaction} will be properly gauge invariant (including under large gauge transformations).

\subsection{Chiral magnetic effect}
The \emph{chiral magnetic effect} refers to a phenomenon in a three-dimensional Fermi liquid where if the Chern number of the Berry curvature on the two-dimensional Fermi surface is nonzero (which, in a non-interacting Fermi gas, would occur when the Fermi surface encloses a Weyl point), then the total charge becomes nonconserved in the presence of both an electric field and magnetic field, with $\textbf{E} \cdot \textbf{B} \neq 0$ \cite{Son_1203}. Recall that that the Berry gauge field on the Fermi surface is still defined in a general EFL. Here we will show that when this Berry gauge field has non-trivial Chern number in a 3D EFL, then the system exhibits the chiral magnetic effect.

We start from the anomaly equation that is the analog for an EFL in three spatial dimensions of \eqnref{eq:theonetrueanomalyeq} (which was stated for an EFL in two spatial dimensions); for convenience we write it in an index-free form as:
\begin{equation}
\label{eq:3d_anomaly}
d(\ast j) = \frac{m}{48\pi^3} F \wedge F \wedge F,
\end{equation}
where $d$ is the exterior derivative on forms, ``$\ast$'' is the Hodge star operator, and $J$ is the current expressed as a 1-form, and $F$ is the 2-form gauge curvature (which we can write locally as $F = dA$, although $A$ may not be globally defined).
 Note that all these forms live on a 6-dimensional space, $M \times \mathcal{F}$, where $M$ is the 4-dimensional space-time manifold and $\mathcal{F}$ is the two-dimensional manifold that parameterizes the Fermi surface. 
 
 Now suppose that we define the total charge current $j_{\mathrm{EM}}$, which is a 1-form on $M$, by integrating $j$ over the Fermi surface $\mathcal{F}$. Suppose furthermore that we write $F = \pi_F^{*}(F_{\mathrm{EM}}) + \pi_M^{*}(F_{\mathrm{Berry}})$, where $F_{\mathrm{EM}}$ are 2-forms on $M$ and $\mathcal{F}$ respectively, and
 $\pi_\mathcal{F}^{*}$ and $\pi_M^{*}$ are the pull-back operators associated with the projections $\pi_\mathcal{F} : \mathcal{F} \times M \to \mathcal{F}$ and $\pi_M : \mathcal{F} \times M$ respectively. In index notation, this would just be saying that $F$ is the sum of two terms, each of which only depends on, and only has components in, $M$ and $\mathcal{F}$ respectively. Then by integrating \eqnref{eq:3d_anomaly} we find that
\begin{equation}
\label{eq:chiral_magnetic_anomaly}
d(\ast j_{\mathrm{EM}}) = \frac{C}{8\pi^2} F_{\mathrm{EM}} \wedge F_{\mathrm{EM}},
\end{equation}
where
\begin{equation}
C = \frac{1}{2\pi} \int_{\mathcal{F}} F_{\mathrm{Berry}}
\end{equation}
is the Chern number of the Berry curvature. We can also write \eqnref{eq:chiral_magnetic_anomaly} as
\begin{equation}
\partial_{\alpha} j_{\mathrm{EM}}^{\alpha} = \frac{C}{8\pi^2} \mathbf{E} \cdot \mathbf{B},
\end{equation}
indicating that the charge is not conserved when the right-hand side is nonzero.

\section{Extension to spinful systems}
\label{sec:spinful}
Throughout this paper, we have assumed that the only microscopic internal symmetry is $\UU(1)$, the charge conservation symmetry. Physically, one often wants to consider systems which also have an $\SU(2)$ spin rotation symmetry. In that case, the full internal symmetry group is $\UU(2)$, which acts by matrix multiplication on the vector $(\psi_{\uparrow}^{\dagger}, \psi_{\downarrow}^{\dagger})$ of microscopic spin-up and spin-down electron creation  operators. 
 Here we will briefly discuss how the general considerations above get extended in that case. (One is still allowed to use the results of previous sections by simply ignoring the spin rotation symmetry, but taking into account the additional symmetry will lead to stronger constraints.)
 For simplicity we will focus on the case of $d=2$ spatial dimensions.
 
 Firstly we observe that $\UU(2)$ has a subgroup $\UU(1)_\uparrow \times \UU(1)_\downarrow \leq \UU(2)$ corresponding to the diagonal unitary matrices. We write the generators of $\UU(1)_\uparrow$ and $\UU(1)_\downarrow$ as $\hat{Q}_\uparrow$ and $\hat{Q}_\downarrow$; they measure the total number of up-spin and down-spin electrons respectively. The total charge is the sum $\hat{Q} = \hat{Q}_\downarrow + \hat{Q}_\uparrow$. They both have corresponding fillings $\nu_{\uparrow}$ and $\nu_{\downarrow}$ in the ground state; however, $\UU(2)$ invariance of the ground state immediately implies that $\nu_\uparrow = \nu_\downarrow := \nu$. The total charge density is then $\rho = 2 \nu$. 
 Now, if we consider a spinful version of a Fermi liquid,  the quasiparticle charge at each Fermi surface point will be conserved, as will the {\em total} quasiparticle spin. Note in particular that the quasiparticle spin at each Fermi surface point is not separately conserved due to the presence of Landau interactions in the spin channel. Thus, for a spinful Fermi liquid, the emergent symmetry group $\GIR$ is a quotient of $\LU(1) \times \UU(2)$, where the quotient corresponds to identifying two $\UU(1)$ subgroups: the $\UU(1)$ subgroup of $\UU(2)$ corresponding to unitary matrices $e^{i\theta} \mathbb{I}_2$, where $\mathbb{I}_2$ is the $(2\times 2)$ identify matrix; and the $\UU(1)$ subgroup of $\LU(1)$ generated by $\hat{N} = \int \hat{n}(\theta) d\theta$, where $\hat{n}(\theta)$ are the generators of $\LU(1)$ as in previous sections. Therefore, we define a spinful EFL to be a system which has the same emergent symmetry group $\GIR$. We assume that the translation symmetry continues to embed into $\LU(1)$ according to \eqnref{eq:translation_embedding}, which defines the Fermi surface of the EFL. Moreover, the mapping of the microscopic $\UU(2)$ into $\GIR$ is induced from a homomorphism $\varphi : \UU(2)  \to \UU(2)_{\mathrm{IR}}$, where the target $\UU(2)_{\mathrm{IR}}$ is the corresponding subgroup of $\GIR$. One can show (for example, by considering the induced homomorphism on the Lie algebras), that such maps are characterized by an odd \footnote{The restriction to odd integers did not come up in Section \ref{sec:quasi_fermi}, but we would have found the same restriction if we assumed that microscopically the fermion parity is $(-1)^{\hat{Q}}$ where $\hat{Q}$ is the total charge (as is the case in an electronic system), and moreover required (as is the case is all known examples) that an elementary Fermi surface quantum be fermionic. In the spinful case, the restriction to odd $q$ comes directly from assumed form of $\GIR$, although in principle one could get even $q$ if we replaced $\UU(2)$ with $\mathrm{SO}(3) \times \UU(1)$ in the definition of $\GIR$.}
   integer $q$, which again we interpret as the charge of a Fermi surface quantum (which in a Fermi liquid would be a Landau quasiparticle). Then, we have, in particular that the microscopic charge $\hat{Q}$ corresponds in the IR theory to
   \begin{equation}
   \label{eq:spinful_charge_embedding}
   \hat{Q} \sim q(\hat{N}_{\uparrow} + \hat{N}_{\downarrow}),
   \end{equation}
   where $\hat{N}_{\uparrow}$ and $\hat{N}_{\downarrow}$ are the generators of the $\UU(1)_{\uparrow,\mathrm{IR}} \times \UU(2)_{\downarrow,\mathrm{IR}} \leq \UU(2)_{\mathrm{IR}}$ symmetry corresponding to the diagonal matrices in $\UU(2)_{\mathrm{IR}}$.

Now we can discuss the form that Luttinger's theorem must take in a spinful EFL. We imagine inserting a $2\pi$ flux of $\UU(1)_{\uparrow,\mathrm{IR}}$ or $\UU(1)_{\downarrow,\mathrm{IR}}$; this leads to a projective representation of $\LU(1)$ described by \eqnref{eq:kac_moody_mod} with some integer anomaly coefficient $m_\uparrow$ or $m_\downarrow$ respectively. $\UU(2)$ invariance again implies that $m_\uparrow = m_\downarrow := m$. Therefore, repeating the argument for Luttinger's theorem, and taking into account \eqnref{eq:spinful_charge_embedding}, we find
 \begin{equation}
 \frac{mq}{(2\pi)^2} \mathcal{V}_F = \nu \quad (\mathrm{mod} \, 1).
 \end{equation}
 where $\mathcal{V}_F$ is the volume enclosed by the Fermi surface. With respect to the total charge density $\rho = 2\nu$, on the other hand, we have
  \begin{equation}
  \label{eq:spinful_luttinger}
 \frac{2mq}{(2\pi)^2} \mathcal{V}_F = \rho \quad (\mathrm{mod} \, 2).
 \end{equation}
 The extra factor of two takes into account the two possible spin values, and agrees (setting $m = 1$) with the usual result for spinful Fermi liquids. (Thus, if we had defined $m$ with respect to the total charge as we did in Section \ref{sec:quasi_fermi}, we would have found $m=2$ for a spinful Fermi liquid/EFL).
 
Our discussion goes through with little modification for a Kondo lattice in which itinerant electrons couple to a local spin-1/2 moment in each unit cell. In this case, we can define a $\UU(1)_\uparrow$ symmetry generated by $\hat{Q}_\uparrow^{\mathrm{tot}} := \hat{S}^z + N/2+ \hat{Q}_\uparrow$, where $\hat{S}^z$ is the total spin component of the local moments, $\hat{Q}_\uparrow$ is defined as before with respect to the itinerant electrons, and $N$ is the total number of unit cells. Assuming no spin ordering of the local moments, we have that $\langle S^z \rangle = 0$ and hence the total filling of $\hat{Q}_\uparrow^{\mathrm{tot}}$ is $\nu^{\mathrm{tot}}_\uparrow = 1/2 + \nu_\uparrow$, where $\nu_\uparrow$ is the contribution from the itinerant electrons. If we now apply similar arguments to before, we find that Luttinger's theorem for a putative IR Fermi liquid requires a ``large'' Fermi surface that counts the local moments as part of the Fermi sea, as previously discussed with the flux threading argument in Ref.~\cite{Oshikawa_0002}.
 
For an alternative point of view, we can think of the  Kondo lattice model as being obtained from an Anderson lattice model with a correlated half-filled band of electrons coupled to a separate weakly correlated partially filled band.  Then, by definition, the charge $\rho$ per unit cell includes \emph{all} the microscopic electrons, in particular the contribution from the correlated band. If the IR theory is a spinful Fermi liquid (or spinful EFL), then Luttinger's theorem in the form \eqnref{eq:spinful_luttinger} shows that there Fermi surface must be ``large'' as in the previous paragraph. Since our discussion (as well as that of Ref. \onlinecite{Oshikawa_0002}) of Luttinger's theorem does not rely on perturbation theory in the interaction strength,  it will hold even in the Kondo limit of the Anderson model.

Note that the arguments of this section assume that all the relevant symmetries -- charge conservation, spin rotation and translation symmetry -- are preserved and not spontaneously broken, and that the IR theory is a spinful Fermi liquid or at least a spinful EFL with the same emergent symmetries. If any of these conditions are violated then the appropriate statement of Luttinger's theorem will be modified and in such cases, a ``small'' Fermi surface could be permitted in the Kondo lattice model. The case of `Fractionalized Fermi Liquids' which have a small Fermi surface while preserving all the symmetries  requires separate discussion but can be easily understood within the framework of this paper. We illustrate this for simple class of such phases in Section \ref{sec:higher_form}.

\section{Fermi surfaces coupled to dynamical gauge fields}
In this section, we will consider cases where the IR theory comprises a Fermi surface coupled to a dynamical gauge field. As we will see, these do not constitute EFLs in the strict sense of Section \ref{sec:quasi_fermi}, but they are still closely related and it still turns out to be helpful to think of the behavior of these systems in terms of their ``kinematic'' properties.

We will consider three main examples. The first two examples relate to a system of interacting electrons in two spatial dimensions with continuous translational symmetry in a magnetic field $B$, such that the magnetic filling $\nu = 2\pi \rho/B$ (where $\rho$ is the electron density) is equal to $\nu = 1/2$. Such a system of electrons is believed to form a ``composite Fermi liquid'' with a Fermi surface of ``composite fermions'' which are distinct from the microscopic electrons. There are two competing proposals for the IR theory of the system, which we consider separately: the one proposed by Halperin, Lee and Read (HLR) \cite{halperin1993theory}; and the one proposed by Son \cite{Son_1502}. In the HLR and Son theories, the system is not quite a conventional Fermi liquid, even when expressed in terms of the composite fermions,  because in the IR theory the composite fermions still couple to a dynamical gauge field.

The final example relates to a Mott insulator in two spatial dimensions with discrete translation symmetry and no magnetic field. In certain circumstances, it is believed that the IR theory of such a system can have a ``spinon Fermi surface'' consisting of emergent fermions coupled to a dynamical gauge field.

\subsection{HLR theory of composite Fermi liquid}
\label{subsec:hlr_theory}
We begin by discussing the microscopic symmetries of electrons in $2d$  in a uniform magnetic field.  It is well known that in the presence of such a magnetic field, the translation operators $\hat{P}_x$ and $\hat{P}_y$ do not commute. Rather they satisfy 
\begin{align}
\label{eq:heisenberg_first}
[\hat{P}_x, \hat{P}_y] &= iB \hat{Q}, \\
[\hat{P}_x, \hat{Q}] &= 0, \\
[\hat{P}_y, \hat{Q}] &= 0, \label{eq:heisenberg_last}
\end{align}
where $\hat{Q}$ is the operator for the total electric charge (generator of the global $U(1)$ symmetry).  This symmetry algebra will have to be matched by any putative IR theory. 
In an infinite system, we think of these operators $\hat{Q}$,$\hat{P}_x$,$\hat{P}_y$ as the generators of the symmetry action on local observables; hence, there is an ambiguity in their definition in that one is free to shift any of them by a constant without changing the action on local observables. Thus, one is free to fix them to have zero eigenvalue in the ground state.
 
Let us now consider the HLR theory of the composite Fermi liquid ground state at $\nu = \frac{1}{2}$. 
Our starting point is the Lagrangian for the HLR theory, which can be written as \cite{halperin1993theory, Seiberg_1606}
\begin{equation}
L = L[\chi,a] - \frac{2}{4\pi} b \wedge db + \frac{1}{2\pi}(A - a) \wedge db
\end{equation}
where $a$ and $b$ are dynamical $\UU(1)$ gauge fields, and $A$ is the background gauge field for the microscopic $\UU(1)$ symmetry. Here $\chi$ is a fermion field carrying unit gauge charge of $a$, and $L[\chi,a]$ is the Lagrangian describing a Fermi surface of these fermions coupled to the gauge field $a$. This is the corrected version (see e.g., Ref.~\cite{Seiberg_1606}) of the Lagrangian initially written by HLR \cite{halperin1993theory}, which suffered from an improperly quantized Chern-Simons term.

Varying with respect to the temporal components $a_0$ and $b_0$ gives the constraints
\begin{align}
\hat{n}_\chi - \hat{B}_b/(2\pi) &= 0, \label{eq:hlr_n_chi} \\
-2\hat{B}_b + (B - \hat{B}_a) &= 0
\end{align}
where $\hat{n}_\chi$ is the density of the $\chi$ fermions, and $\hat{B}_a$ and $\hat{B}_b$ are the magnetic fields of the $a$ and $b$ gauge fields respectively. Meanwhile, varying with respect to $A_0$ allows us to identify the microscopic charge density $\hat{\rho}$ as
\begin{equation}
\label{eq:hlr_rho}
\hat{\rho} = \hat{B}_b/(2\pi).
\end{equation}

We define the excess charge according to
\begin{equation}
\hat{N} = \int (\hat{n}_\chi - n_\chi) d^2 \textbf{x},
\end{equation}
where $n_\chi$ is the charge density in the ground state. We also define the effective magnetic flux that the $\chi$ fermions experience (they only couple directly to $a$) according to
\begin{equation}
\hat{\Phi} = \int \hat{B}_a d^2 \textbf{x},
\end{equation}
Using $2 n_\chi = B$ and the constraints above, it follows that 
\begin{equation}
\label{eq:hlr_effective_flux}
2\hat{N} = -\frac{\hat{\Phi}}{2\pi}.
\end{equation}

Now we need to consider what form the IR symmetry takes. 
Let us first consider what happens if we treat the gauge fields $a$ and $b$ at the mean-field level. In that case, the $\chi$ fermions do not experience a magnetic field in the ground state and therefore form a Fermi liquid for which the charge at each point of the Fermi surface is conserved.
Thus, the emergent symmetry group is still something resembling $\mathrm{LU}(1)$ with generators $\hat{n}(\theta)$. It is not exactly $\mathrm{LU}(1)$, however, because if we consider low-energy excitations that carry a net charge $\hat{N} = \int \hat{n}(\theta) d\theta \neq 0$, where $\hat{n}(\theta)$ are the generators of the emergent symmetry group, then from \eqnref{eq:hlr_effective_flux} we see that the $\chi$ fermions feel a net magnetic flux. From the discussion of Section \ref{subsec:fermi_monopoles} we expect
   the magnetic flux to induce a projective representation of $\LU(1)$ \footnote{One might ask why the magnetic flux does not also induce a chiral flow of quasiparticles along the Fermi surface, violating conservation of $\hat{n}(\theta)$. The point is that in that for a low-energy excitation, the variation of the local charge density, and therefore of the flux of $a$, should be spread over an arbitarily large volume, and hence the effective local magnetic field, which determines the velocity of the chiral flow, actually goes to zero.}. Hence, the
\newcommand{\wtLU}{\mathrm{\widetilde{L}U}}
emergent symmetry group is $\wtLU(1)$, where $\wtLU(1)$ is obtained from $\LU(1)$ by replacing the commutation relation $[\hat{n}(\theta), \hat{n}(\theta')] = 0$ obeyed by the generators $\hat{n}(\theta)$ of $\LU(1)$ with:
\begin{equation}
\label{eq:hlr_commutator}
[\hat{n}(\theta), \hat{n}(\theta')] = 2i\frac{1}{2\pi} \delta'(\theta - \theta') \int \hat{n}(\theta) d\theta,
\end{equation}
and with those elements $U_f \in \mathrm{LU}(1)$ corresponding to functions $ f : S^1 \to \UU(1)$ with non-trivial winding number $W_f = \frac{1}{2\pi} \int \partial_\theta f(\theta)$ excluded. [Because otherwise \eqnref{eq:hlr_commutator} would imply that $U_f \hat{N} U_f^{-1} = (1+2W_f)\hat{N}$, which is mathematically inconsistent unless $W_f = 0$.]
From Eqs.~(\ref{eq:hlr_n_chi}) and (\ref{eq:hlr_rho}), it is clear that the microscopic charge $\hat{Q}$ embeds into the IR symmetry group according to $\hat{Q} \sim \hat{N}$. Meanwhile, the continuous translation symmetry, with generators $\hat{P}_x$ and $\hat{P}_y$, should embed according to  
\begin{equation}
\label{eq:hlr_P}
\hat{P}_\alpha \sim \int k_\alpha(\theta) \hat{n}(\theta) d\theta,
\end{equation}
which defines the Fermi surface $\textbf{k}(\theta)$.

We will not attempt to extend the general discussion of topological terms as in Section \ref{subsec:fermi_thooft} to this more non-trivial symmetry group $\wtLU(1)$. However, we can at least argue for a version of Luttinger's theorem for the composite Fermi liquid. Unlike in the case of no magnetic field, it will turn out not to be necessary to invoke the 't Hooft anomaly of the IR theory; the result actually follows directly from the structure of the emergent symmetry group.

From the action of translations on the low-energy theory, we find that
\begin{equation}
\label{eq:translation_commutator_IR}
[\hat P_x, \hat P_y] \sim \frac{2\mathcal{V}_F}{2\pi} \hat{N},
\end{equation}
where
\begin{equation}
\label{eq:cfl_luttinger}
\mathcal{V}_F = \int k_x(\theta) \frac{dk_y(\theta)}{d\theta} d\theta
\end{equation}
is the volume enclosed by the Fermi surface. We emphasize that here $\hat{N}$ is not just a number but an \emph{operator} which acts non-trivially on the states of the IR theory.

 On the other hand, microscopically we know that we must have the magnetic translation algebra of Eqn. \ref{eq:heisenberg_first} with which we must impose compatibility of \eqnref{eq:translation_commutator_IR}. This  gives
\begin{equation}
\label{eq:composite_luttinger}
\mathcal{V}_F = \pi B,
\end{equation}
which is Luttinger's theorem for the composite Fermi liquid.

Now we should discuss to what extent we expect these considerations to survive once we include gauge fluctuations. The only property of the IR theory that we needed to derive Luttinger's theorem was that the emergent symmetry group is $\wtLU(1)$. The question is whether this emergent symmetry survives the inclusion of gauge fluctuations. Let us argue that, at any rate, the emergent symmetry group cannot be a compact Lie group, which suggests that the full $\wtLU(1)$ symmetry is preserved. We cannot apply the results of Section \ref{sec:generic} to show this, because if we define the magnetic unit cell such that the discrete translations with respect to this unit cell commute, then the filling with respect to this unit cell is supposed to be $\nu = 1/2$ for the composite Fermi liquid, a rational number.

However, we can instead argue as follows. Suppose that the emergent symmetry group $\GIR$ is a compact Lie group. Then let $\hat{p}_x, \hat{p}_y, \hat{q}$ be the elements of the Lie algebra of $\GIR$ corresponding to the microscopic translation symmetries $\hat{P}_x$ and $\hat{P}_y$ and the microscopic charge conservation symmetry $\hat{Q}$. Then since the microscopic symmetries obey the algebra of Eqs.~(\ref{eq:heisenberg_first}--\ref{eq:heisenberg_last}). 
 it follows that the same algebra must be satisfied by $\hat{p}_x, \hat{p}_y, \hat{q}$.
Eqs.~(\ref{eq:heisenberg_first}--\ref{eq:heisenberg_last}) constitute the Heisenberg algebra and cannot be implemented inside of a compact Lie group unless $\hat{q}$ is identically zero. To see this, suppose that $\GIR$ is a compact Lie group. Then it admits a finite-dimensional faithful unitary representation. Since $\hat{p}_x, \hat{p}_y$ and $\hat{q}$ generate a subgroup of $\GIR$, this induces a finite-dimensional representation of the Heisenberg algebra Eqs.~(\ref{eq:heisenberg_first}--\ref{eq:heisenberg_last}). Now we invoke the famous fact that the Heisenberg algebra does not admit any finite-dimensional representations in which $\hat{q}$ acts non-trivially\footnote{To see this, just decompose into irreducible representations. In any such irreducible representation, by Schur's Lemma $\hat{q}$ acts proportionally to the identity $\hat{q} = Q \mathbb{I}$ for some scalar $Q$. Then taking the trace of $[\hat{p}_x, \hat{p}_y] = iB\hat{q}$ shows that $Q = 0$.}.

Therefore, we conclude that, if $\GIR$ is a compact Lie group then the charge conservation symmetry acts trivially on the IR theory; in other words, the system has a charge gap and is an electrical insulator. This would be inconsistent with the physics expected of a composite Fermi liquid. Thus, we expect that the emergent symmetry group remains $\wtLU(1)$ in the presence of gauge fluctuations. Of course, in principle it could go to some totally different group that is larger than any compact Lie group, but this does not seem very plausible, and moreover, as we we discussed in Section \ref{sec:quasi_fermi}, the charge is still conserved at every point in the Fermi surface in the conventional description of the fixed point for a Fermi surface coupled to dynamical gauge field. Thus, assuming that the emergent symmetry group remains $\wtLU(1)$, we find that the composite Fermi liquid indeed obeys Luttinger's theorem. Note that, though it was suggested based on numerics in Ref.~\cite{Balram_1506} that Luttinger's thoerem could be violated in the composite Fermi liquid, more extensive numerics showed this to be a finite size effect \cite{Balram_1707}.

\subsection{Son theory of composite Fermi liquid}
We start from the Lagrangian \cite{Son_1502,Seiberg_1606}
\begin{equation}
\mathcal{L} =i\bar \chi \cancel{D} _{a}\chi - {\frac{2}{4\pi}} b \wedge db + {\frac{1}{2\pi}}a \wedge db - {\frac{1}{2\pi}}A \wedge db + \cdots,
\end{equation}
where $a$ and $b$ are dynamical $\UU(1)$ gauge fields, and the first term represents
 the Lagrangian for a massless relativistic Dirac fermion field $\chi$ coupled to the gauge field $a$ \footnote{We follow the regularization convention adopted in Ref.~\cite{Seiberg_1606}, where the action for the massless Dirac fermion implicitly includes the so-called ``$\eta$ invariant''. Physically this amounts to assuming that the system also contains an additional massive Dirac fermion in addition to the massless one explicitly represented in the Lagrangian.}. The ``$\cdots$'' represents higher-order terms. This is the corrected version \cite{Seiberg_1606} of the Lagrangian originally written by Son \cite{Son_1502}.

Proceeding similarly to Ref.~\cite{Seiberg_1606}, from this Lagrangian one obtains the constraints
\begin{align}
\hat{N} - \frac{1}{4\pi} \hat{\Phi}_a + \frac{1}{2\pi} \hat{\Phi}_b &= 0, \\
-2 \hat{\Phi}_b + \hat{\Phi_a} &= 0, \\
\hat{Q} &= -\frac{1}{2\pi} \hat{\Phi}_b, \label{eq:sonQ}
\end{align}
where $\hat{N}$ is the charge of the $\chi$ fermions, $\hat{\Phi}_a$ and $\hat{\Phi}_b$ are the total magnetic fluxes of the $a$ and $b$ fields respectively, $\hat{Q}$ is the total microscopic charge, and all of $\hat{N}, \hat{Q}, \hat{\Phi}_a, \hat{\Phi}_b$ represent the excess values compared to the ground state. From these equations we find that $\hat{N} = 0$ and $\hat{\Phi}_a = 2\hat{\Phi}_b$.

At the mean-field level the $\chi$ fermions form a Fermi liquid, so the charge at each point on the Fermi surface is conserved. However, we have to impose the constraint that $\hat{N} = \int \hat{n}(\theta) d\theta = 0$. Assuming for the moment that the fermions feel no magnetic flux, this would imply that the emergent symmetry group is $\Omega \UU(1) := \LU(1)/\UU(1)$.

Now if we take into account that the fermions feel an effective magnetic flux given by $\hat{\Phi}_a = 2\hat{\Phi}_b$, we must have the commutation relation
\begin{equation}
\label{eq:son_commutator}
[\hat{n}(\theta), \hat{n}(\theta')] = -2i\frac{1}{2\pi} \delta'(\theta - \theta') \frac{\hat{\Phi}_b}{2\pi}.
\end{equation}
Since the $\hat{n}(\theta)$'s are conserved, it must also be the case that the magnetic flux $\hat{\Phi}_b$ is conserved and (since it is quantized in units of $2\pi$) generates a $\UU(1)_{\mathrm{flux}}$ symmetry.
We can also observe that, for a general function $ f : S^1 \to \UU(1)$, \eqnref{eq:son_commutator} implies that the corrsponding group element $U_f $ satisfies $U_f \hat{N} U_f^{-1} = \hat{N} + 2W_f \hat{\Phi}_b$, where $W_f = \frac{1}{2\pi} \int \partial_\theta f(\theta) d\theta$ is the integer winding number. Since we are imposing that $\hat{N}$ is identically zero (and $\hat{\Phi_b}$ is not), we must exclude $U_f$ with $W_f \neq 0$ from the symmetry group, as in the HLR case.

 Hence, the full emergent symmetry group $\GIR$ is a central extension of $\Omega_0 \UU(1)$ by $\UU(1)_{\mathrm{flux}}$ where $\Omega_0 \UU(1) = \mathrm{L}_0 \UU(1)$, and $\mathrm{L}_0 \UU(1) / \UU(1)$ is the subgroup of $\LU(1)$ corresponding to functions $f : S^1 \to \UU(1)$ with trivial winding number. We note that this is a different group compared to the HLR case, which reflects a physical difference between the Son and HLR theories.

\eqnref{eq:sonQ} tells us how the microscopic $\UU(1)$ symmetry embeds into $\GIR$, and as before the continuous translation should embed according to \eqnref{eq:hlr_P}, which defines the Fermi surface. By similar arguments to the HLR case one again finds Luttinger's theorem \eqnref{eq:composite_luttinger} for the composite Fermi liquid. Since, again, the only property of the IR theory that we used was the nature of the emergent symmetry group, we expect that this result continues to hold in the presence of gauge fluctuations.

Note that our argument for Luttinger's theorem in the Son theory is very different from the one given by Son \cite{Son_1502}, who noted that the density of composite fermions is given by $\rho_{CF} = B/(4\pi)$, and then invoked Luttinger's theorem for a Fermi liquid to relate this to the volume of the Fermi surface. Of course, the composite Fermi liquid is \emph{not} a Fermi liquid beyond mean-field theory, but the arguments of this paper have shown that Luttinger's theorem is in fact far more general.
However, from this point of view it remains unclear whether one ought to have expected a Luttinger's theorem to hold with respect to $\rho_{CF}$, since the composite fermions do not carry charge under any global $\UU(1)$ symmetry. The argument we gave above does not suffer from these difficulties.

\subsection{Spinon Fermi surface}
\label{subsec:spinonFS}
Finally, let us consider a system in 2 spatial dimensions where the IR theory consists of a \emph{spinful} Fermi surface coupled to a dynamical $\UU(1)$ gauge field. This should apply, for example, to Mott insulators exhibiting a ``spinon Fermi surface''.

If we treat the dynamical gauge field at the mean-field level, then the fermions form a spinful Fermi liquid described by the considerations of Section \ref{sec:spinful}. Thus, we have conserved Fermi surface density $\hat{n}(\theta)$. In particular, if we use the observation of that section that a spinful Fermi surface has anomaly coefficient $m=2$ with respect to the total charge $\UU(1)$, we conclude that the $\hat{n}(\theta)$'s must satisfy
 \begin{equation}
 \label{eq:spinon_loop_commutator}
 [\hat{n}(\theta), \hat{n}(\theta')] = -2i \frac{1}{2\pi} \delta'(\theta - \theta') \frac{\hat{\Phi}_b}{2\pi},
 \end{equation}
 where $\hat{\Phi}_b$ is the operator that measures the magnetic flux of the dynamical gauge field. Moreover, we have to impose that the total gauge charge is zero, so $\int \hat{n}(\theta) d\theta = 0$. We see that the emergent symmetry group $\GIR$ takes the form $\GIR = G_{\mathrm{IR,charge}} \times \mathrm{SO}(3)$, where the charge part $G_{\mathrm{IR,charge}}$ is identical to the emergent symmetry group of the Son theory as described above [although for an electrical insulator the microscopic charge $\UU(1)$ will act trivially on the IR theory instead of embedding into the emergent $\UU(1)_{\mathrm{flux}}$], while the $\mathrm{SO}(3)$ part accounts for to the spin-rotation symmetry of the spinons\footnote{It is $\mathrm{SO}(3)$ rather than $\mathrm{SU}(2)$ or $\UU(2)$ because of the constraint that the total gauge charge, i.e. the net spinon number, is zero for any state in the IR theory.}.
 
 Unlike in the quantum Hall systems discussed above, for a Mott insulator one microscopically has only a discrete (commuting) translation symmetry. The discrete translations must embed into $\GIR$ in the usual way:
 \begin{equation}
 \mathbb{T}_\alpha \sim \exp\left(-i\int k_\alpha(\theta) \hat{n}(\theta) \right),
 \end{equation}
 which defines the Fermi surface. From \eqnref{eq:spinon_loop_commutator}, the requirement that $\mathbb{T}_x$ and $\mathbb{T}_y$ commute give a non-trivial constraint:
 \begin{equation}
 \label{eq:spinon_luttinger_pre}
 \frac{\mathcal{V}_F}{(2\pi)^2} = \mbox{$0$ or $\frac{1}{2}$} \quad (\mathrm{mod} \, 1),
 \end{equation}
 where $\mathcal{V}_F$ is the volume enclosed by the Fermi surface.
 
 We can obtain more information by considering the 't Hooft anomaly of the IR theory and invoking the filling constraints of Section \ref{subsec:2d_filling_constraint}. We assume that the system microscopically has spin-rotation symmetry, so that the full microscopic internal symmetry group is $\UU(2)$, and in particular there is a $\UU(1)_\uparrow \times \UU(1)_\downarrow$ subgroup as discussed in Section \ref{sec:spinful}, and we can define the corresponding fillings $\nu_\uparrow$ and $\nu_\downarrow$ which satisfy $\nu_\uparrow = \nu_\downarrow := \nu$ and $\rho = 2\nu$, where $\rho$ is the filling of the total charge. For an insulating state without symmetry fractionalization (see Section \ref{sec:higher_form}), the microscopic charge $\UU(1)$ acts trivially on the IR theory and hence from the general theory of filling constraints we find that $\rho$ must be an integer. On the other hand, for a band insulator, $\rho$ must be an even integer (taking into account the two spin components). Since a Mott insulator by definition is an insulator that exists at fillings where band insulators are impossible, we conclude that $\rho$ for a Mott insulator must be an odd integer and hence $\nu = 1/2 \,(\mathrm{mod} \, 1)$.
 
From the general discussion of Section \ref{subsec:2d_filling_constraint}, we now see that if we apply a flux $\Phi_\uparrow$ of $\UU(1)_\uparrow$ and flux $\Phi_\downarrow$ of $\UU(1)_\downarrow$, it must be the case that translations act projectively on it,
\begin{equation}
\label{eq:spinon_translation_commutator}
\mathbb{T}_x \mathbb{T}_y \mathbb{T}_x^{-1} \mathbb{T}_y^{-1} = (-1)^{(\Phi_\uparrow + \Phi_\downarrow)/(2\pi)}
\end{equation}
 Now we need to see how this comes about in the IR theory. We return to \eqnref{eq:spinon_loop_commutator} but take into account that, at the mean-field level where we can treat the dynamical gauge field as a background, the fermions now feels an effective $\UU(1)_\uparrow$ flux of $\hat{\Phi}_b + \Phi_\uparrow$ and an effective $\UU(1)_\downarrow$ flux of $\hat{\Phi}_b + \Phi_\downarrow$. Hence, \eqnref{eq:spinon_loop_commutator} gets generalized to\footnote{We remark that if we set $\Phi_\uparrow = \Phi_\downarrow = 2\pi$, then the effect on \eqnref{eq:spinon_thooft} can be absorbed by redefining $\hat{\Phi}_b \to \hat{\Phi}_b - 2\pi$. Since $\hat{\Phi}_b$ always has such an ambiguity in terms of its action on a background flux (such ambiguities are the reason why fluxes are able to carry projective representations in the first place), we conclude that $\Phi_\uparrow = \Phi_\downarrow = 2\pi$ does not carry a non-trivial projective representation. This corresponds to the statement that the projective representations of $G_{\mathrm{IR,charge}}$ are $\mathbb{Z}_2$ classified, whereas the projective representations of $\LU(1)$ (see Section \ref{subsec:fermi_monopoles}) were $\mathbb{Z}$ classified. These results are consistent with the fact that the total microscopic charge $\UU(1)$ acts trivially on the IR theory since it is an electrical insulator.}
 \begin{equation}
 \label{eq:spinon_thooft}
 [\hat{n}(\theta), \hat{n}(\theta')] = -i \frac{1}{2\pi} \delta'(\theta - \theta') \frac{2\hat{\Phi}_b + \Phi_\uparrow + \Phi_\downarrow}{2\pi},
 \end{equation}
 which indeed gives \eqnref{eq:spinon_translation_commutator} provided that
 \begin{equation}
 \label{eq:spinon_luttinger}
 \frac{\mathcal{V}_F}{(2\pi)^2} = \frac{1}{2} \quad (\mathrm{mod} \, 1),
 \end{equation}
 which is Luttinger's theorem for a spinon Fermi surface
 
 As usual, we must consider to what extent we expect these results to hold beyond mean-field theory. The result \eqnref{eq:spinon_luttinger_pre} only depends on the structure of $\GIR$, so we must ask whether $\GIR$ remains the same upon including gauge fluctuations. Since we are assuming the system is a Mott insulator, with integer filling of the microscopic $\UU(1)$ charge, we cannot use the results of Section \ref{sec:generic}. Still, as we mentioned in the introduction to Section \ref{sec:quasi_fermi}, the charge on each patch of the Fermi surface remains conserved in the usual approach to Fermi surfaces coupled to a dynamical gauge field, suggesting that $\GIR$ is indeed robust to gauge fluctuations and that \eqnref{eq:spinon_luttinger_pre} is satisfied generally. To get the stronger result \eqnref{eq:spinon_luttinger}, we also have to assume that the 't Hooft anomaly of the IR theory, captured by \eqnref{eq:spinon_thooft}, remains unchanged by gauge fluctuations. However, the classification of 't Hooft anomalies is discrete -- more concretely, there is no way to continuously deform \eqnref{eq:spinon_thooft} without leading to inconsistencies -- so it seems unlikely that the 't Hooft anomaly would be affected by gauge fluctuations. Moreover from \eqnref{eq:spinon_luttinger_pre} the only possibility other than \eqnref{eq:spinon_luttinger} would be $\mathcal{V}_F/(2\pi)^2 = 0 \, (\mathrm{mod} \,1)$, in which case the Fermi surface is singular and presumably unstable.

\section{Impossibility of Fermi arcs and other constraints on the Fermi surface}
\subsection{Fermi arcs}
\label{subsec:fermi_arcs}
In the celebrated pseudogap normal phase of the cuprate high temperature superconducting  materials, ARPES measurements observe ``Fermi arcs'' rather than a closed Fermi surface (see, eg, Ref.~\cite{keimer2015quantum} for a review). One possible explanation is that these materials still have a closed Fermi surface, but for some reason parts of the Fermi surface are not easily visible to the ARPES probe.
On the other hand, 3-dimensional Weyl semimetals exhibit Fermi arcs at their 2-dimensional boundaries \cite{Wan_2011}, but these can only exist at the boundary of a gapless 3-dimensional bulk.
It is an important fundamental question to ask whether a system with true Fermi arcs (not related to the above mechanisms) can ever exist.

Specifically, in this section we will ask whether it is possible to extend the framework of EFLs to substitute a closed Fermi surface with a Fermi arc. We will find that there is an obstruction to doing so. Therefore, in light of the results of Section \ref{sec:generic}, the only remaining route to Fermi arcs, assuming that they can exist in systems at generic filling, would be to find some other way to have an emergent symmetry group larger than any compact Lie group, different from what occurs in EFLs.

Let us consider a theory in $d$ spatial dimensions described by a suitable generalization of the EFL class discussed in Section \ref{sec:quasi_fermi}. Specifically, we imagine that the Fermi surface is a parameterized by a $(d-1)$-dimensional manifold $F$; but, instead of requiring that $F$ be a closed manifold as before, we allow it to be a manifold with boundary. Then we assume that the emergent symmetry group is $L^F \UU(1)$, the space of smooth maps from $F$ to $\UU(1)$. We can discuss 't Hooft anomalies as before, and again we find that the 't Hooft anomaly should be described by inflow from a Chern-Simons term on a higher-dimensional manifold, for example for $d=2$:
\begin{equation}
S[A] = \frac{m}{24\pi^2} \int_{M_{+} \times F} A \wedge dA \wedge dA
\end{equation}
where $M_{+}$ is an extension of the $d+1$-dimensional space-time manifold to a $d+2$-dimensional manifold. Now, however, we immediately encounter a problem if $F$ has a boundary $\partial F$, because the Chern-Simons term is not gauge-invariant -- that is, the system has a 't Hooft anomaly -- in the presence of boundaries (unless $m=0$). At the boundary of $M_+$, the gauge variation is canceled by the 't Hooft anomaly of the EFL, but since the interior of $M_{+}$ is supposed to correspond to a gapped SPT bulk, there is no way to cancel the anomaly at $M_{+} \times \partial F$.

Should we be worried about this anomaly? In general, an uncanceled 't Hooft anomaly signals a violation of charge conservation. We can see this concretely in spatial dimension $d=2$ if we assume a particular model of a Fermi arc which has Fermi liquid-like quasiparticles. Then in the presence of a magnetic field, the quasiparticles have a chiral flow along the Fermi surface as discussed in Section \ref{subsec:fermi_monopoles}; but when they reach the end of the Fermi arc they have no choice but to simply disappear.
This violates conservation of microscopic charge (if the microscopic charge $q$ of a quasiparticle is nonzero)  and conservation of momentum (even when $q=0$).
Beyond this simple model, the 't Hooft anomaly indicates the same issue with lack of charge/momentum conservation, and therefore we must impose that, if $F$ has a boundary, then $m=0$. However, by the discussion of Section \ref{subsec:fermi_monopoles} this would imply the filling $\nu = 0 \, \mathrm{mod} \, 1$, so it cannot correspond to a state that exists at generic filling.

Of course, this is only a proof of the impossibility of Fermi arcs if we assume that the system is described by the EFL framework. We cannot make more general statements, except to recall from Section \ref{sec:generic} that any system that can exist at generic filling yet is not an EFL needs to find some other way of having an emergent symmetry group that is larger than any compact Lie group.

Finally, we should address one physical mechanism that one might think of to generate Fermi arcs, where the system couples to a critical boson at some wave-vector $\mathbf{q}$.  Then, in the quasiparticle picture described above, one could imagine that a quasiparticle, upon reaching the end of a Fermi arc at momentum $\mathbf{k}$, jumps to the beginning of another arc at momentum $\mathbf{k} - \mathbf{q}$ while emitting a boson of momentum $\mathbf{q}$. In the language of emergent symmetries, this would correspond to introducing another operator $\hat{N}_{\mathrm{boson}}$ that generates a $\mathrm{U}(1)$ group, and then writing the translation symmetry as
\begin{equation}
\mathbb{T}_\alpha = \exp\left(-i\int k_\alpha(\theta) \hat{n}(\theta) d\theta\right) \exp\left(-iq_\alpha \hat{N}_{\mathrm{boson}}\right).
\end{equation}
This would amount to postulating a slightly different emergent symmetry group, namely a modification of $\mathrm{LU}(1)$ where one allows the corresponding functions $f : S^1 \to \mathrm{U}(1)$ to have at least a finite number of discontinuities.
 The question is whether there is a meaningful non-trivial anomaly that one could write for such a group, as is necessary if it is to describe a system with non-integer microscopic filling. It seems likely that there is not, for reasons that we will explain at the end of the next subsection.

\subsection{More general constraints}

\begin{figure}

\newcommand{\mycmark}{\textcolor{green!70!black}{\cmark}}
\newcommand{\allowedconfig}{\node at (0.9,0.1) {\LARGE \mycmark};}
\newcommand{\myxmark}{\textcolor{red}{\xmark}}
\newcommand{\disallowedconfig}{\node at (0.9,0.1) {\LARGE  \myxmark};}
\newcommand{\mylabel}[1]{\node at (0.08,0.92) {(#1)};}
\begin{tikzpicture}[scale=4]
\newcommand{\fsarcradius}{0.35}

\mylabel{a}

\draw [orientedfs,domain=15:75] plot ({0.5+\fsarcradius*cos(\x)},{0.5+\fsarcradius * sin(\x)});
\draw [orientedfs,domain=105:165] plot ({0.5+\fsarcradius*cos(\x)},{0.5+\fsarcradius * sin(\x)});
\draw [orientedfs,domain=195:255] plot ({0.5+\fsarcradius*cos(\x)},{0.5+\fsarcradius * sin(\x)});
\draw [orientedfs,domain=195:255] plot ({0.5+\fsarcradius*cos(\x)},{0.5+\fsarcradius * sin(\x)});
\draw [orientedfs,domain=285:345] plot ({0.5+\fsarcradius*cos(\x)},{0.5+\fsarcradius * sin(\x)});

\node at (0.19,0.19) {\fsmarking{+1}};
\node at (0.81,0.19) {\fsmarking{+1}};
\node at (0.19,0.81) {\fsmarking{+1}};
\node at (0.81,0.81) {\fsmarking{+1}};

\draw[brillouinzone] (0,0) -- (0,1) -- (1,1) -- (1,0) -- cycle;

\disallowedconfig

\newcommand{\subfigshift}{1.1}
\begin{scope}[shift={(\subfigshift,0)}]
\newcommand{\fsarcradiustwo}{0.3}

\mylabel{b} 

\draw [orientedfs,domain=0:180] plot ({0.5+\fsarcradiustwo*cos(\x)},{0.5+\fsarcradiustwo * sin(\x)});
\draw [orientedfs,domain=180:360] plot ({0.5+\fsarcradiustwo*cos(\x)},{0.5+\fsarcradiustwo * sin(\x)});

\node at (0.5,0.87) {\fsmarking{+1}};
\node at (0.5,0.13) {\fsmarking{+2}};

\draw[ferminode] (0.5-\fsarcradiustwo,0.5) circle(\ferminoderadius);
\draw[ferminode] (0.5+\fsarcradiustwo,0.5) circle(\ferminoderadius);

\draw[brillouinzone] (0,0) -- (0,1) -- (1,1) -- (1,0) -- cycle;

\disallowedconfig
\end{scope}
\begin{scope}[shift={(0,-\subfigshift)}]
\mylabel{c}
\draw[orientedfs] plot[smooth,domain=0:1,variable=\t] ({0.5+0.1*cos(360*\t)},{\t});
\draw[brillouinzone] (0,0) -- (0,1) -- (1,1) -- (1,0) -- cycle;
\node at (0.5,0.53) {\fsmarking{+1}};
\disallowedconfig
\end{scope}

\begin{scope}[shift={(\subfigshift,-\subfigshift)}]
\mylabel{d}
\fill[fermifill] plot[smooth,domain=0:1,variable=\t] ({0.1*cos(360*\t) + 0.25},{\t}) -- 
      plot[smooth,domain=1:0,variable=\t] ({-0.1*cos(360*\t) + 0.75}, {\t});
      \draw[brillouinzone] (0,0) -- (0,1) -- (1,1) -- (1,0) -- cycle;
      
      \draw[orientedfs] plot[smooth,domain=0:1,variable=\t] ({0.1*cos(360*\t) + 0.25},{\t});
      \draw[orientedfs] plot[smooth,domain=1:0,variable=\t] ({-0.1*cos(360*\t) + 0.75}, {\t});
      \node at (0.2,0.85) {\fsmarking{+1}};
            \node at (0.8,0.85) {\fsmarking{+1}};
            \allowedconfig
\end{scope}

\begin{scope}[shift={(0,-2*\subfigshift)}]
\mylabel{e}
\draw[orientedfs,fermifill] (0.5,0.5) circle(0.3);
\draw[brillouinzone] (0,0) -- (0,1) -- (1,1) -- (1,0) -- cycle;
\node at (0.75,0.2) {\fsmarking{+1}};
\allowedconfig
\end{scope}

\begin{scope}[shift={(\subfigshift,-2*\subfigshift)}]
   
\mylabel{f}
\fill[fermifilldouble] (0.5,0.5) circle(0.35);
\draw[brillouinzone] (0,0) -- (0,1) -- (1,1) -- (1,0) -- cycle;

\begin{scope}
\clip (0.5,0.5) circle(0.35);
\draw[fermisurface,fermifill, decoration={
    markings,
    mark=at position 0.755 with {\arrow{<}}},
    postaction=decorate] (0.5,0.9) circle(0.35);
\end{scope}

\draw[fermisurface,decoration={
    markings,
    mark=at position 0.255 with {\arrow{>}}},
    postaction=decorate] (0.5,0.5) circle(0.35);
\draw[fermisurface,decoration={
    markings,
    mark=at position 0.75 with {\arrow{>}}},
    postaction=decorate] (0.5,0.5) circle(0.35);
    
\node at (0.5,0.09) {\fsmarking{+2}};
\node at (0.5,0.49) {\fsmarking{+1}};
\node at (0.5,0.91) {\fsmarking{+1}};
\allowedconfig
\end{scope}
\end{tikzpicture}

\caption{\label{fig:fermi}Allowed (\mycmark) and disallowed (\myxmark) Fermi surface configurations in a two-dimensional Brillouin zone. Each Fermi surface segment has an associated anomaly coefficient $m$, shown by the number next to the segment. An arbitary choice of orientation of the segment (denoted with arrows) determines the convention for the sign of $m$. For allowed configurations, it is always possible to define an integer function $n(\textbf{k})$ -- depicted here by the different shading colors of different regions -- such that the boundary between two regions carries the anomaly coefficient determined by the difference of $n(\textbf{k})$ between the two regions. In a Fermi gas, $n(\textbf{k})$ represents the number of occupied bands at the point $\textbf{k}$.}
\end{figure}
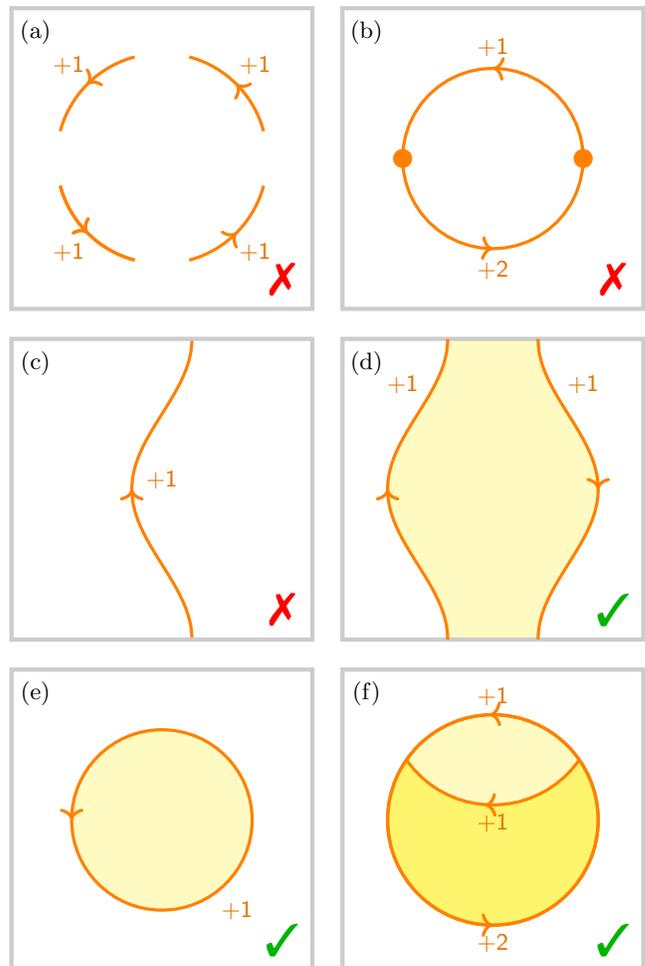
Returning to EFLs, we can can also obtain some more general constraints than just the impossibility of Fermi arcs. Indeed, the most general statement will be that the Fermi surface must ``enclose'' a volume in a suitable generalized sense. This obviously fails for Fermi arcs, but it also rules out Fermi surfaces which are closed but nevertheless fail to enclose a volume in the Brillouin zone, e.g. see Figure \ref{fig:fermi}(c).

Let us imagine generally that we construct the Fermi surface out of a collection of ``patches'' which might be glued together along their boundaries and might not all have the same values of the anomaly coefficient $m$. Then the requirement that the Chern-Simons term be gauge-invariant places a constraint on how patches are allowed to be glued together. For example, in two spatial dimensions the rule is that the sum of the anomaly coefficients for all the segments of Fermi surface (taking into account the orientation of the segment) intersecting at a given point must be zero. This is the most general sense in which the Fermi surface must be ``closed''. Thus, the configuration shown in Figure \ref{fig:fermi}(f) is allowed, but the one shown in Figure \ref{fig:fermi}(b) is not.

In fact, we can make an even stronger statement: configurations such as the one shown in Figure \ref{fig:fermi}(c), even though the Fermi surface is closed, are also disallowed (except possibly in the case where the Fermi surface quanta are uncharged, i.e. $q=0$). In this case, the Chern-Simons term is gauge-invariant, but there is still an issue with the implementation of translation symmetry. Recall that for $d=2$ the translation symmetry $\mathbb{T}_x$ and microscopic charge operator $\hat{Q}$ can be expressed in terms of the $\LU(1)$ generators $\hat{n}(\theta)$ as
\begin{align}
\mathbb{T}_x &\sim \exp\left(-i \int k_x(\theta) \hat{n}(\theta) d\theta\right) \\
\hat{Q} &\sim q\int \hat{n}(\theta) d\theta.
\end{align}
In light of the commutation relations \eqnref{eq:kac_moody}, when acting on a $2\pi$ flux we have
\begin{equation}
\mathbb{T}_x \hat{Q} \mathbb{T}_x^{-1} \sim \hat{Q} + mq W_x,
\end{equation}
where
\begin{equation}
W_x = \frac{1}{2\pi} \int \frac{dk_x(\theta)}{d\theta} d\theta
\end{equation}
is a winding number of the Fermi surface on the Brillouin torus.
However, microscopically the translation operator should commute with the charge, so we must impose that $mqW_x = 0$. Another way to see this (which generalizes more easily to higher dimensions) would be impose conservation of microscopic charge on the electric field response discussed in Section \ref{subsec:electric_field}.

From such considerations, we can obtain the most general constraint on the Fermi surface. We find that it must be the case that each point
$\mathbf{k}$ in the Brillouin zone can be assigned an integer $n(\textbf{k})$, such that the boundary between two regions with different $n(\textbf{k})$ carries a Fermi surface with anomaly coefficient $m$ given by the difference $\Delta n$.
In a Fermi gas, $n(\textbf{k})$ has a natural interpretation: it describes the number of bands that are occupied at the point $\textbf{k}$. What we have found is that such an $n(\textbf{k})$ can always be defined in any consistent EFL. Note that in general the Fermi surface only determines $n(\textbf{k})$ up to addition of a $\textbf{k}$-independent integer.
Repeating the analysis of Section \ref{subsec:fermi_monopoles} in this more general setting, we can also now formulate the most general version of Luttinger's theorem:
\begin{equation}
\frac{q}{(2\pi)^d} \int d^2 \textbf{k} n(\textbf{k}) = \nu \quad (\mathrm{mod} \, 1),
\end{equation}
where the integral is over the whole Brillouin zone.

We remark that that a compact mathematical way to state the above results is if we think of the Fermi surface, together with associated anomaly coefficients $m$ of the different patches, as defining a \emph{chain} (in the homology theory sense) $\omega \in C_{d-1}(\mathbb{T}^d, \mathbb{Z})$, where $\mathbb{T}^d$ represents the Brillouin zone. Then the requirement that the Chern-Simons term be gauge invariant implies that this chain is closed, i.e. it has trivial boundary, $\partial \omega = 0$; whereas, from the winding number arguments, we find that this chain must be \emph{exact}, $\omega = \partial \kappa$ for some $\kappa$.

Finally, we note that in retrospect the result should have been clear, because however the microscopic translations act in the IR theory (which is specified by specifying the Fermi surface), it must be possible to compute a filling $\nu$ by considering the anomaly of the emergent symmetry group. For Fermi surfaces that do not enclose a volume, there is evidently no meaningful way to define an associated $\nu$. It is for this reason that we do not expect there to be any non-trivial anomalies associated with the modified symmetry group described in the last paragraph of Section \ref{subsec:fermi_arcs}, again because it would lead to Fermi surfaces that do not enclose a volume and therefore have no meaningful $\nu$ associated with them.
\section{The role of fractionalization}
\label{sec:higher_form}
The analysis of Sections \ref{sec:filling_constraints} and \ref{sec:generic} is not sufficient to describe systems which admit ``fractionalized'' excitations, i.e.\ localized excitations which cannot be created locally. To see this, consider a specifc and familiar example, namely a system of bosons on a lattice  in two spatial dimensions with microscopic filling $\nu = 1/2$. A symmetry preserving ground state that is allowed is one where the IR theory  is equivalent to a $\mathbb{Z}_2$ gauge theory\cite{read1991large,wen1991mean,senthil2000z}. 
Let us first recall the physical properties of this state and how it is allowed at filling $\nu = 1/2$ filling. The key property is that the ground state is topologically ordered and admits fractionalized excitations, which we can label as $1, e, m, f$, where $e$ is the $\mathbb{Z}_2$ gauge charge, $m$ is the $\mathbb{Z}_2$ gauge flux, and $f$ is the composite of $e$ and $m$. Moreover, these excitations can exhibit ``symmetry fractionalization'' of the microscopic symmetries, which roughly means that they carry fractional quantum numbers. Specifically, in the $\nu = 1/2$ Mott insulator, one of the excitation types (say $e$) carries \emph{half-quantized} charge under the microscopic $\UU(1)$ symmetry. This is allowed because  these excitations are fractionalized and can only be created in pairs. Moreover, the $m$ particle also experiences translational symmetry fractionalization\cite{jalabert1991spontaneous,senthil2000z,sachdev1999translational}, which is to say that $\mathbb{T}_x \mathbb{T}_y \mathbb{T}_x^{-1} \mathbb{T}_y^{-1} = -1$ when acting on a single $m$ excitation. Again, this does not contradict the fact that translations commute microscopically, because the global number of $m$ particles is always even. Physically, one can think of the translational symmetry fractionalization as saying that the ground state has an $e$ particle in each unit cell; the $-1$ phase factors then comes from the $-1$ braiding phase for $m$ going around $e$. But then since the $e$ particle has $1/2$ charge under the microscopic $\UU(1)$, these $e$ particles contribute a filling $\nu = 1/2$.  It is also instructive to ask how this state is compatible with the momentum balance constraints from flux threading introduced by Oshikawa \cite{Oshikawa_0002}.  The point is that adiabatic threading of a $2\pi$ flux through one cycle of a torus shifts the ground state from one topological sector to another which is at a different total crystal momentum\cite{Misguich_0112}. This change of the ground state momentum is then able to match what is expected from flux threading. 

How should we think about the emergent IR symmetry of this state?  Deep in the IR, at scales below the gap to all quasiparticle excitations, we can describe it as a purely topological theory. Naively the only emergent symmetry of the resulting theory might seem to be $\GIR = \mathbb{Z}_2$, generated by the ``electromagnetic duality'' which exchanges a $\mathbb{Z}_2$ gauge charge with a gauge flux; however, for the theory under consideration, neither translations nor the microscopic charge $\UU(1)$ map into this duality symmetry. Therefore, following the analysis of Section \ref{sec:filling_constraints} we would immediately conclude that $\nu = 0$ mod 1. How then, are we to make sense, in the context of our general framework, of the fact that such a IR theory actually can exist at $\nu = 1/2$?

We have to think about what went wrong in the argument of Section \ref{subsec:2d_filling_constraint}. The assumption we made there is that a $2\pi$ flux of the microscopic symmetry generated by the microscopic charge $\hat{Q}$ will correspond in the IR theory to a $2\pi$ flux of $\hat{Q}_{\mathrm{IR}}$, the generator of the corresponding IR symmetry. In any system with a charge gap, we will have that $\hat{Q}_{\mathrm{IR}}$ is identically zero, so we would conclude that such an object would be completely trivial, leading to a trivial filling. However, in the $\mathbb{Z}_2$ gauge theory example, a $2\pi$ flux of the microscopic symmetry is evidently a non-trivial object since acting on it we have $\mathbb{T}_x \mathbb{T}_y \mathbb{T}_x^{-1} \mathbb{T}_y^{-1} = -1$. In fact, as a defect of the IR theory it is equivalent to an $m$ particle. In general, an interpretation of non-trivial symmetry fractionalization is precisely the statement that certain flux configurations of the microscopic symmetry look like topologically non-trivial excitations of the IR theory. Therefore, it will always be necessary to take into account symmetry fractionalization when computing the filling.

Let us first recall a few cases from previous works where symmetry fractionalization of the microscopic symmetry affects the filling.
The first case generalizes the deconfined $\mathbb{Z}_2$ gauge theory for bosons at $\nu = 1/2$ (as discussed above) to situations where the IR theory is an arbitrary gapped (2+1)-D gapped topological phase. The computation of $\nu$ has been discussed in Refs.~\cite{Paramekanti_0406,Bonderson_1601}. The basic idea is that the translational symmetry fractionalization is described by a class in $H^2(\mathbb{Z} \times \mathbb{Z}, A) \cong A$, where $A$ is the group of Abelian anyons. Roughly, we can imagine that there is a background $a$ particle sitting in each unit cell. Meanwhile, the symmetry fractionalization of $\UU(1)$ is described by an Abelian anyon $b$, which we can think of as the anyonic charge carried by a $2\pi$ flux. Because of the background of $a$ particles, if $b$ particles have non-trivial braiding phase $S_{ab}$ with $a$ particles, then they carry a projective representation of translation symmetry. Thus, we find that the filling is given by
\begin{equation}
\nu = \frac{1}{2\pi i} \log S_{ab} \quad (\mathrm{mod} \, 1).
\end{equation}
A nice way to interpret this equation is that since a $2\pi$ flux has non-trivial braiding $S_{ab}$ with an $a$ particle, then the latter must carry fractional charge $q_a = \frac{1}{2\pi i} \log S_{ab} \; (\mathrm{mod} \, 1)$ under the $\UU(1)$ symmetry. Therefore, since we have an $a$ particle in each unit cell, the average charge per unit cell is $\nu = q_a$.

Building on the previous case, we can imagine a scenario where the IR theory is equivalent to stacking a gapless theory $\mathcal{C}_{\mathrm{gapless}}$ without any topologically non-trivial excitations and a gapped topological theory $\mathcal{C}_{\mathrm{gapped}}$. In that case, the microscopic filling $\nu$ is simply the sum of the contributions from each component:
\begin{equation}
\nu = \nu_{\mathrm{gapless}} + \nu_{\mathrm{gapped}},
\end{equation}
where $\nu_{\mathrm{gapless}}$ is computed according to Section \ref{sec:filling_constraints} and $\nu_{\mathrm{gapped}}$ is computed according to the previous paragraph.
An example of such a case is the  FL* Fermi liquid discussed in \cite{Senthil_0209, Senthil_0305}.

If we want to construct a more general theory beyond these examples, we might start to become a bit more uncertain about what exactly we mean in general by a ``topologically non-trivial excitation''.  Instead, we will introduce an alternative formalism that is more readily generalizable \cite{Etingof_0909,Barkeshli_1410,Benini_1803,Hsin_1904}. To motivate this formalism, note that in the $\mathbb{Z}_2$ gauge theory example,  $\GIR = \mathbb{Z}_2$ is not in fact the only emergent symmetry of the IR theory.  Indeed, gapped topological phases in two spatial dimensions, such as $\mathbb{Z}_2$ gauge theory, have emergent \emph{1-form} symmetries \cite{Gaiotto_1412}. This is a general feature of gapped topological phases. In fact one can think about the 1-form symmetry operators, which act on closed strings when $d=2$, as becoming the open string operators that create topologically non-trivial excitations at the endpoints when the string is broken.

Once one has introduced the emergent higher-form symmetries, one can phrase symmetry fractionalization entirely in terms of these symmetries without needing to talk about excitations at all. In general, a $d$-dimensional system can have emergent 0-form, 1-form, ..., up to $(d-1)$-form symmetries. These symmetries combine into a structure called a ``$d$-group'' \cite{Baez_0307,Gukov_1307,Kapustin_1309,Benini_1803}, which we denote $\dblul{G}_{\mathrm{IR}}$. Meanwhile, the microscopic symmetries can also be thought of as a forming a $d$-group $\dblul{G}_{\mathrm{UV}}$, although normally we will consider the case where the microscopic symmetries are all 0-form and there are no non-trivial higher-form symmetries. Then the correct way to describe the action of the microscopic symmetries on the IR theory is through an $n$-group homomorphism
\begin{equation}
\dblul{\rho} : \dblul{G}_{\mathrm{UV}} \to \dblul{G}_{\mathrm{IR}}.
\end{equation}
One might think if the microscopic symmetries are all 0-form, then we do not need to worry about the higher-form part of $\dblul{G}_{\mathrm{IR}}$, but it turns out that even in this case the homomorphism $\dblul{\rho}$ contains not just the data of how the microscopic symmetries map into emergent 0-form symmetries, but also the data of the symmetry fractionalization of the microscopic symmetries, which can be interpreted as relating to the interplay between the microscopic symmetries and the emergent higher-form symmetries.
%\newtext{As we describe in Appendix, the relation between the microscopic filling and the 't Hooft anomaly of the emergent symmetry group $\dblul{G}_{\mathrm{IR}}$ can indeed be described in terms of the homomorphism $\dblul{\rho}$.}

It is instructive to consider another case, in addition to gapped topological phases, where higher-form symmetries modify the relation between microscopic filling and the properties of the IR theory, namely superfluids, which spontaneous break the charge $\mathrm{U}(1)$ symmetry. A superfluid in $d$ spatial dimensions has an emergent $\mathrm{U}(1)$ $(d-1)$-form symmetry, since the winding number of the order parameter along any closed curve is conserved in the IR theory. Moreover, a superfluid is a compressible phase which can exist at any filling. One can show that the filling is encoded in the symmetry fractionalization of the translation symmetry by the $(d-1)$-form symmetry. Physically this is reflected in the fact that a vortex sees an ``effective magnetic field'' given by the charge density \cite{lannert2001quantum,balents2005putting}.

Finally, it is necessary for us to address the question of whether symmetry fractionalization can allow for the spirit of Theorem \ref{thm:mainthm} to be  bypassed; that is, is it possible for a system whose emergent symmetry is described by a compact Lie $d$-group to exist at \emph{generic} filling? The superfluid example above shows that we have to restrict to systems in which the microscopic $\UU(1)$ and translation symmetries are not spontaneously broken. In that case, we conjecture that a compact Lie $d$-group symmetry is incompatible with generic filling. To motivate this, firstly, we expect that in order to have generic filling, at the least the system needs to have an emergent continuous $(d-1)$-form symmetry (e.g. a 1-form symmetry in $d=2$ or a 2-form symmetry in $d=3$). The usual way to get emergent continuous higher form symmetries would be to have an emergent deconfined gauge field. However, in $d=3$, for example (we return to the more subtle $d=2$ case below), we know that a deconfined $\UU(1)$ gauge theory has two 1-form symmetries (electric and magnetic) but no 2-form symmetry. We can contrast this with gapped topological phases in $d=3$, which typically have both a 1-form symmetry and 2-form symmetry.

We also recall that a deconfined $\UU(1)$ gauge theory in $d=2$ spatial dimensions is dual to a superfluid, and that the only way for such a theory to be stable (other than by coupling to charged gapless fermions, which would explicitly break the 1-form symmetry) is if the monopole events (in the superfluid language, charge creation events) are forbidden by one of the microscopic symmetries, in which case this symmetry will be spontaneously broken.

Therefore it appears, at least in these well-known mechanisms for obtaining emergent higher-form symmetries, that there is no way to obtain generic filling without spontaneously breaking either the microscopic $\UU(1)$ or translation symmetry. Whether there is some less familiar theory in $d$ spatial dimensions with an emergent $(d-1)$-form symmetry that can achieve this, we leave as an open question.

\section{Discussion: relation to ``violations of Luttinger's theorem''}
In this work, we have argued that Luttinger's theorem, or a suitable generalization, in fact should hold in great generality, both for Fermi liquids and at least a certain class of non-Fermi liquids. Therefore, we must contend with past works that have reported a ``violation of Luttinger's theorem'' in various settings.

In certain cases \cite{Dzyaloshinskii__2003,Stanescu_0508,Yang_0602}, attempts have to been made to formulate Luttinger's theorem in a way that could apply beyond perturbation theory in the interaction strength by expressing it in terms of the volume enclosed both by traditional Landau Fermi surfaces on
 which the electron Green's function $G(\mathbf{k},\omega)$ has a pole at $\omega=0$ and by ``Luttinger surfaces'' in which $\mathrm{Re} G(\mathbf{k},\omega)$ changes sign via a zero of $G(\mathbf{k},0)$ instead of a pole. However, it has been shown that no such result holds in general \cite{Rosch_0602,Dave_1207}. In our work, we have defined a Fermi surface in a general system by its emergent symmetries, rather than by any particular properties of the electron Green's function. Our version of Luttinger's theorem is therefore very different from the one disproven by Refs.~\cite{Rosch_0602,Dave_1207}.
 
A potential violation of Luttinger's theorem has been discussed in the context of holographic models \cite{Hartnoll_1011,HartnollBlackHoles,Iqbal_1112,Hashimoto_1203}. For a system that is holographically dual to a gravitational theory that includes a black hole, the charge hidden behind the event horizon does not contribute to the volume enclosed by the Fermi surface. It is not clear, however, whether such a scenario can ever actually occur at zero temperature, since an event horizon would be associated with a nonzero entropy density, seemingly in tension with the Third Law of Thermodynamics. If somehow the event horizon could be stabilized at zero temperature, however, one could envision it being associated with some exotic emergent symmetry group that is larger than any compact Lie group but different from that of a Fermi liquid, in which case a violation of Luttinger's theorem would be compatible with the considerations of this paper.

Next we discuss the numerical study of the $t-J$ model in Ref.~\cite{Putikka_9803}. By an analytic continuation of a high-temperature series expansion, Ref.~\cite{Putikka_9803} computed the properties of the $t-J$ model down to around temperature $T \sim 0.2J$ and found that, for several different definitions of the ``Fermi surface'' at finite temperature, the volume inside the Fermi surface appears to be plateauing at a value less than the value predicted by Luttinger's theorem. However, our arguments in favor of Luttinger's theorem only applied at zero temperature. Since the ground state of the $t-J$ model is unknown and could well exhibit exotic physics such as high-temperature superconductivity, the significance of these finite-temperature observations remains unclear, especially given the ambiguity of the definition of the Fermi surface at finite temperature and the fact that the results of Ref.~\cite{Putikka_9803} also contain hints of Fermi arcs (see Section \ref{subsec:fermi_arcs}). 

Finally, we discuss Ref.~\cite{Cappelluti_9906}, which studied an $\SU(N)$ generalization of the $t-J$ model in a $1/N$ expansion where a Fermi liquid mean-field solution becomes exact in the limit as $N \to \infty$. Ref.~\cite{Cappelluti_9906} found that the fluctuations appear to cause a violation of Luttinger's theorem to leading order in $1/N$. According to our general arguments, a violation of Luttinger's theorem must correspond to a change in the emergent symmetry group (such that the system is no longer an EFL). However, the mean-field solution is an EFL, in particular a Fermi liquid, and it seems unlikely that fluctuations could give rise to a \emph{new} emergent symmetry not already present while remaining perturbative. Therefore, assuming there is no error in the calculations of Ref.~\cite{Cappelluti_9906}, it seems that we must interpret their results as signaling a breakdown of perturbation theory, such that the radius of convergence in $1/N$ is zero. In that case, there is no reason to believe that Luttinger's theorem is violated at all.

\section{Conclusion and outlook}
In this work, we have made significant progress towards an understanding of strongly correlated metals, going beyond both Fermi liquid theory and perturbation theory in the interaction strength. We have highlighted the essential role of the emergent symmetry of the IR theory and its 't Hooft anomaly.

Going forward, we believe that our work leads to a new perspective on studies of non-Fermi liquids: from our point of view,  a crucial question one should ask about any potential non-Fermi liquid state is what  its emergent symmetry group is, and if it is supposed to represent a compressible state, in what way the emergent symmetry group and its 't Hooft anomaly enable the system to satisfy the filling constraints developed in this paper for a continuously tunable filling $\nu$. (In particular, in this case the emergent symmetry group must be larger than any compact Lie group).

What one might expect to find is that in fact, \emph{every} non-Fermi liquid is simply an ersatz Fermi liquid (EFL) as defined in Section \ref{sec:quasi_fermi} or a slight variant of it (for example including fractionalization as discussed in Section \ref{sec:higher_form}). It is a very important open question to determine whether there might be a fundamentally different way to realize compressible metallic states, with a totally different emergent symmetry group (which must still be larger than any compact Lie group). If such a possibility were to be realized, we would still expect that consideration of the emergent symmetry group and its 't Hooft anomaly will still be a powerful way to determine properties of the system, as we found for EFLs.

In particular, it is quite energizing to imagine how our results apply to the strange metal phase in cuprates \cite{keimer2015quantum}, assuming that it is possible for it to extend all the way down to a critical point at zero temperature, and assuming that it can exist in a translationally invariant system without disorder.
The filling at which the critical point occurs does not appear to be near any particular rational filling. Therefore, it is likely that the filling is generic and can be continuously tuned depending on the microscopic parameters, which as we have noted, leads to very strong constraints on the IR theory. Although the strange metal has extremely exotic properties, there does not appear to be any obvious reason why it could not still be an EFL (In particular, we emphasize again that the ``Fermi surface quantum'' excitations of an EFL do \emph{not} need to have the nature of quasiparticles, which are certainly expected to be absent in the strange metal.)
Identification of the strange metal as being in the EFL class, or even a significant generalization of it, would have further profound implications which we explore in a subsequent work \cite{Else_2010}.

\begin{acknowledgments}
We thank Zhen Bi, Meng Cheng, Andrey Gromov, Shamit Kachru, Ethan Lake, John McGreevy, Max Metlitski, Chong Wang, and Liujun Zou for helpful discussions. D.V.E.\ was supported by the EPiQS Initiative of the Gordon and Betty Moore Foundation, Grant No. GBMF8684. T.S.\ is supported by a US Department of Energy grant DE- SC0008739, and in part by a Simons Investigator award from the Simons Foundation. This work was also partly supported by
the Simons Collaboration on Ultra-Quantum Matter, which is a grant from the Simons Foundation (651440, TS).

\end{acknowledgments}
\appendix
\section{Filling does not correspond to the boundary of an SPT}
\label{appendix:fillingnotspt}
Here we will explain why, contrary to what one might initially assume, filling in $d$ spatial dimensions does not correspond to a boundary of a bulk SPT in $d+1$ spatial dimensions, when viewed in terms of the microscopic symmetries.
The relevant symmetry group is $\GIR = \mathbb{Z}^d \times \UU(1)$. One can show \cite{Kitaev_0901,Thorngren_1612,Else_1910} that the SPT classification for such a symmetry group in  $D$ spatial dimensions is given by
\begin{equation}
\label{eq:weak_classification}
\mathcal{C}_D \times \prod_{r=1}^d \mathcal{C}_{D-r}^{\times \begin{pmatrix}d \\r \end{pmatrix}},
\end{equation}
where $\mathcal{C}_{m}$ is the classification of $\UU(1)$ SPTs in $m$ spatial dimensions,
 $\mathcal{C}^{\times k}$ denotes the product of $k$ copies of $\mathcal{C}$, and $\begin{pmatrix} d \\ r\end{pmatrix}$ is the binomial coefficient. Physically, $\mathcal{C}_D$ correspond to ``strong'' SPT phases, i.e. those protected by $\UU(1)$ alone, while the remaining terms correspond to ``weak'' SPT phases that can be understood by stacking layers of lower-dimensional strong SPT phases.
 
If filling in $d$ spatial dimensions could somehow correspond to an SPT in $D= d+1$ spatial dimensions, then it would have to be an SPT that inherently relies on all $d$ translation symmetries. Therefore, the only term in \eqnref{eq:weak_classification} that could be relevant is the $r=d$ term, $\mathcal{C}_1$. The problem is that there \emph{are} no 1-dimensional SPTs protected by $\UU(1)$ alone, in either bosonic or fermionic systems. Therefore, $\mathcal{C}_1 = 0$. So we conclude that there are no SPTs in $d+1$ spatial dimensions that could correspond to filling in $d$ dimensions.

\section{Translation symmetry fractionalization on a $2\pi$ flux in 2-D systems at fractional filling}
\label{appendix:monopole_translations}
In this section, we will give a more detailed argument for why a $2\pi$ flux in a 2-D system with fractional filling should transform projectively under the translation symmetry, as discussed in Section \ref{subsec:2d_filling_constraint}.

First of all, let us recall that in the presence of a magnetic field one should introduce gauge-invariant translation operators in the following way. For simplicity of notation, let us assume that the system is defined in continuous space  (although with only a discrete translation symmetry); similar considerations can be made in a tight-binding model.
We introduce the generator of gauge-invariant continuous translations according to
 \begin{equation}
\widetilde{\mathbf{P}} = \mathbf{P} - \int \mathbf{A}(\mathbf{x}) \hat{\rho}(\mathbf{x}) d^2 \textbf{x},
\end{equation}
where $\mathbf{A}(\textbf{x})$ is the magnetic vector potential, $\hat{\rho}(\mathbf{x})$ is the local density operator, and $\mathbf{P}$ is the generator of translations when $\mathbf{A} = 0$. These translation operators satisfy the commutation algebra
\begin{equation}
\label{eq:continuous_commutator}
[\widetilde{P}_\alpha, \widetilde{P}_\beta] =  i \varepsilon_{\alpha \beta} \int B(\mathbf{x}) \hat{\rho}(\mathbf{x}) d^2 \mathbf{x},
\end{equation}
where, where $B(\textbf{x}) = \varepsilon^{\alpha \beta} \partial_\alpha A_\beta(\textbf{x})$ is the magnetic field.
We then define the discrete translation operators $\mathbb{T}_x$ and $\mathbb{T}_y$ by
\begin{equation}
\label{eq:discrete_in_terms_of_continuous}
\mathbb{T}_\alpha = \exp(-i \mathbf{a}^{(\alpha)} \cdot \widetilde{\mathbf{P}}),
\end{equation}
where $\mathbf{a}^{(\alpha)}$ is the lattice translation vector in the $\alpha$ direction.

Now consider any state $\ket{\psi}$ that obeys the ``cluster property'' 
\begin{equation}
\langle \hat{\rho}(\mathbf{x}) \hat{\rho}(\mathbf{x}') \rangle - \langle \hat{\rho}(\mathbf{x})\rangle \langle \hat{\rho}(\mathbf{x}') \rangle \leq C |\mathbf{x} - \mathbf{x'}|^{-\eta}
\end{equation}
for some constants $C$ and $\eta > 0$, and $\langle \cdot \rangle = \bra{\psi} \cdot \ket{\psi}$
denotes expectation values with respect to $\ket{\psi}$. Then it follows from \eqnref{eq:continuous_commutator} that, if we fix that $\int B(\textbf{x}) d^2 \textbf{x} = 2\pi$ but make the magnetic field more and more spread out, then the variance
\begin{equation}
\langle [\widetilde{P}_\alpha, \widetilde{P}_\beta]^2 \rangle - \langle [\widetilde{P}_\alpha, \widetilde{P}_\beta]\rangle^2
\end{equation}
goes to zero in the limit where the magnetic field is infinitely spread out. Therefore, when acting on such a state $\ket{\psi}$, and in this limit, we can replace the commutator by its expectation value:
\begin{equation}
[\widetilde{P}_\alpha, \widetilde{P}_\beta] \sim  \langle [\widetilde{P}_\alpha, \widetilde{P}_\beta] \rangle = i \epsilon_{\alpha \beta} \int B(\textbf{x}) \rho(\textbf{x}) d^2 \textbf{x} \sim \frac{2\pi \nu}{v} i \epsilon_{\alpha \beta}, \label{eq:doing_stuff}
\end{equation}
where $v$ is the volume of one unit cell and $\rho(\textbf{x}) = \langle \hat{\rho}(\textbf{x}) \rangle$. To get to the last identification, we used the fact that $B(\textbf{x})$ is very slowly varying, and that the integral of $\rho(\textbf{x})$ over one unit cell is $\nu$. Finally, from \eqnref{eq:doing_stuff} and the definition \eqnref{eq:discrete_in_terms_of_continuous}, we find that
\begin{equation}
\mathbb{T}_x \mathbb{T}_y \mathbb{T}_x^{-1} \mathbb{T}_y^{-1} \sim e^{2\pi i \nu},
\end{equation}
which is \eqnref{eq:pretrans_projective}.

\section{Filling constraints in general dimension from the topological action of SPTs}
\label{appendix:topological_action}
Here we will discuss the sense in which the microscopic filling can be computed from the IR theory in general spatial dimension $d$, generalizing the discussion of Sections \ref{subsec:1d_filling_constraints} and \ref{subsec:2d_filling_constraint}. First of all, we know that there is a one-to-one correspondence between the 't Hooft anomaly of the IR theory and an SPT in $d+1$ spatial dimensions. 
Since the SPT phase is topological, at long wave-lengths a $\GIR$ SPT phase in $d+1$ spatial dimensions is described by a (Euclidean) topological action $S[M,A]$ that depends on the $d+2$-dimensional space-time manifold $M$ and a gauge field $A$ for the $\GIR$ symmetry. This action can be evaluated on any oriented space-time manifold $M$. (In some cases, there may be some additional structures required on $M$, for example fermionic systems require a spin structure or $\mathrm{spin}_c$ structure.) This action is purely imaginary (i.e.\ has no real part), does not depend on a metric on $M$ and, provided that $\GIR$ is a compact Lie group, depends only on the SPT phase and is independent of the microscopic details of the SPT ground state. In general, the action is only defined modulo $2\pi i$ since this does not affect the amplitude $e^{-S}$.

Let $\tau_1, \cdots, \tau_d \in \GIR$ be the elements of the IR symmetry group corresponding to microscopic translations. We can define a generalization of the function $\alpha$ from the previous subsections in the following way. We set the space-time manifold $M$ to be the $d+2$-dimensional torus $T^{d+2}$. We require that over the first $d$ of the $d+2$ directions of the torus, the gauge field $A$ is flat, with the fluxes through the associated 1-cycles given by $\tau_1, \cdots, \tau_d$. Then over the final two directions on the torus, we require that there be a $2\pi$ flux of the $\UU(1)$ subgroup of $\GIR$ (generated by $\hQIR$) corresponding to the microscopic $\UU(1)$ symmetry. Finally, we define $\alpha(\hQIR|\tau_1, \cdots, \tau_d) = -i S[M,A]$. One can show that this agrees with the previous definitions of $\alpha$ in the cases $d=1,2$.

Thus, our hypothesis for the form of the filling constraint is
\begin{equation}
\nu = \frac{1}{2\pi} \alpha(\hQIR|\tau_1, \cdots, \tau_d) \quad (\mathrm{mod} \, 1).
\end{equation}
We have not devised a physical derivation analogous to the $d=1$ and $d=2$ cases; however we expect that this should be possible, since $\alpha$ should relate to the anomalous action of translations on a $2\pi$ flux, which in $d$ spatial dimensions is a $(d-2)$-dimensional object.

We can extend this definition to include higher form symmetries as follows. As we have discussed in Section \ref{sec:higher_form}, in this case we can include the possibility of symmetry fractionalization by considering the map $\rho:G_{\rm UV} \to G_{\rm IR}$ to be a map of higher groups, where $G_{\rm IR}$ contains all higher form symmetries. We will show that there is a $G_{\rm IR}$ gauge field $A$ on $T^{d+2}$ associated to $\rho$ which generalizes the $A$ above, so that $\nu = - i S[T^{d+2},A]/2\pi$. First, we equip the $d+2$ torus $T^{d+2}$ with a $\mathbb{Z}^d \times \UU(1)$ gauge field $A'$ defined as follows. For $j = 1, \ldots, d$, the $j$th 1-cycle $C_j$ of $T^{d+2}$ has $\int_{C_j} A'$ equal to the generator of the $j$th $\mathbb{Z}$ factor. On the last two coordinates, which form a $T^2$, we place a $2\pi$ flux in the $\UU(1)$ factor. We then take $A = \rho(A')$.

Note that, as we mentioned, in certain cases the topological action could depend on other data (such as a spin structure for fermionic systems). However, we expect that if the action still depends non-trivially on this additional data for the particular space-time manifold $M$ and gauge field $A$ considered here, this corresponds to cases where the microscopic symmetry $\GUV = \mathbb{Z}^d \times \UU(1)$ still has a non-trivial 't Hooft anomaly in $d$ spatial dimensions. In this case, the system cannot be realized as lattice system in $d$ spatial dimensions with translation symmetry and on-site $\UU(1)$ symmetry, only as a boundary of an SPT in $d+1$ spatial dimensions. In such systems, the filling of the $d$-dimensional system does not need to be well-defined.

\section{Projective representation of loop group on monopole from 't Hooft anomaly}
\label{appendix:monopole}
In this Appendix, we will show that \eqnref{eq:kac_moody} follows from the anomaly equation \eqnref{eq:theonetrueanomalyeq}.

Let us work on a 4-dimensional space parameterized by coordinates $(t,x,y,\theta)$. We start by considering a configuration where the $A_x$ and $A_y$ components of the gauge field depend only on $x$ and $y$, and the $A_t$, $A_\theta$ components depend only on $z$,$t$. Then, we integrate \eqnref{eq:theonetrueanomalyeq} over $x$ and $y$. We obtain
\begin{equation}
\label{eq:dimredaneq}
\partial_\mu \widetilde{j}^{\mu} = \frac{m\phi}{2\pi} (\partial_t A_\theta - \partial_\theta A_t),
\end{equation}
where the index $\mu$ now varies only over $t$ and $\theta$, 
and we have defined $\phi = \frac{1}{2\pi} \int (\partial_x A_y - \partial_y A_x) dx dy$, which is the number of flux quanta passing through the system, and $\widetilde{j}^{\mu} = \int j^{\mu} dx dy$. Effectively, what we have done is compactify our original system in 2 spatial dimensions to one in 0 spatial dimensions. This 0-D system has a 't Hooft anomaly that depends on $\phi$, and is canceled by inflow from a 2-D SPT. Recall that such 't Hooft anomalies correspond to projective representations of the symmetry group. If we set $\phi = 1$, this corresponds to the projective representation of the symmetry group acting on a $2\pi$ flux.

In our case, where the symmetry group is $G = \LU(1)$, we can express the projective representation through this central extension of the commutation algebra of the symmetry generators $\hat{n}(\theta)$, i.e.
\
\begin{equation}
[\hat{n}(\theta), \hat{n}(\theta')] = c(\theta, \theta'),
\end{equation}
where $c(\theta, \theta')$ is a $c$-number (i.e.\ not an operator) to be determined.

The projective representation should be independent of the choice of Hamiltonian, so we are free to choose $H = \int f(\theta) \hat{n}(\theta) d\theta$ for some arbitrary test function $f(\theta)$, which amounts to coupling the zero Hamiltonian to the gauge field $A_t = f(\theta)$. Then we find that Heisenberg evolution $\hat{n}(\theta,t) = e^{itH} \hat{n}(\theta) e^{-itH}$ satisfies
\begin{equation}
\label{eq:density_derivative}
\frac{d}{dt} \hat{n}(\theta,t) = i \int c(\theta', \theta) f(\theta') d\theta'.
\end{equation}
We should compare this with the anomaly equation \eqnref{eq:dimredaneq}, in which we choose $A_\theta$ to be independent of time, we have $\widetilde{j}^{\theta} = \frac{\delta H}{\delta A_\theta} = 0$,  and we identify the time-component $\widetilde{j}^t$ with the expectation value of $\hat{n}(\theta,t)$. We find that, to be consistent with \eqnref{eq:density_derivative}, it must be that
\begin{equation}
i\int c(\theta', \theta) f(\theta') d \theta' = -\frac{m \phi}{2\pi} \frac{d}{d\theta} f(\theta),
\end{equation}
which is true for an arbitrary test function $f$ only if 
\begin{equation}
c(\theta,\theta') = -i \frac{m\phi}{2\pi} \delta'(\theta - \theta').
\end{equation}
Finally, to obtain the extra factor of $q$ in the right-hand side of \eqnref{eq:kac_moody}, we recall that \eqnref{eq:charge_embedding}
implies that $\phi$, which we can interpret as the number of flux quanta of the IR $\UU(1)$ symmetry generated by $\int \hat{n}(\theta) d\theta$, is $q$ times the number of flux quanta of the microscopic $\UU(1)$ symmetry.

\section{Projective representation of loop group  in a Fermi liquid}
\label{appendix:semiclassical}
In this Appendix, we derive \eqnref{eq:kac_moody} in a Fermi liquid by working in terms of the semiclassical theory of band electrons. Strictly speaking, we will only treat non-interacting electrons, but since our arguments are purely geometric and the result ultimately can be reduced to a contribution coming solely from the Fermi surface, we expect that the same results will hold in a Fermi liquid in light of the quasiparticle picture.

For clarity, in this appendix we temporarily work in units where $\hbar \neq 1$. Then we can write \eqnref{eq:kac_moody} as
\begin{equation}
[\hat{n}(\theta), \hat{n}(\theta')] = -i \frac{\Phi mq}{(2\pi)^2 \hbar} \delta'(\theta - \theta')
\end{equation}
where $\Phi = \int B(\textbf{x}) d^2 \textbf{x}$ is the total magnetic flux. Therefore, in the semiclassical theory what we want to show is the Poisson bracket
\begin{equation}
\{ n(\theta), n(\theta') \} = -\frac{\Phi mq}{(2\pi \hbar)^2} \delta'(\theta - \theta').
\end{equation}
Note that for each $\theta$, $n(\theta)$ is a semiclassical observable, i.e.\ a real-valued function on the semiclassical phase space.

What we will actually prove is the following equivalent statement:
\begin{equation}
\label{eq:result_to_prove}
\{ n_f, n_g \} = -\frac{\Phi q m}{(2\pi \hbar)^2} \int f(\theta) \frac{dg(\theta)}{d\theta} d\theta,
\end{equation}
for any smooth functions $f$ and $g$, where we have defined   
\begin{equation}
n_f = \int f(\theta) n(\theta) d\theta.
\end{equation}
 We can relate $n_f$ and $n_g$ to the density $n(\textbf{k}, \textbf{x})$ of semiclassical particles according to:
\begin{equation}
n_f = \int f_{+}(\textbf{k}) n(\textbf{k}, \textbf{x}) d^2 \textbf{x} d^2 \textbf{k},
\end{equation}
where $f_{+}(\textbf{k})$ is any smooth extension of $f$ away from the Fermi surface such that $f_{+}(\textbf{k}_F(\theta)) = f(\theta)$, where $\textbf{k}_F(\theta)$ is the momentum of the Fermi surface as a function of $\theta$.
Note that by $n(\textbf{k}, \textbf{x})$ we mean that for \emph{each} $\textbf{k}, \textbf{x}$,  $n(\textbf{k}, \textbf{x})$ is a function of phase space variables. If there are $N$ particles in total then the phase space variables are labeled by $(\mathbf{x}^{(1)}, \cdots, \mathbf{x}^{(N)}, \mathbf{k}^{(1)}, \cdots, \mathbf{k}^{(N)})$, and we can write
\begin{multline}
n(\textbf{k}, \textbf{x})(\mathbf{x}^{(1)}, \cdots, \mathbf{x}^{(N)}, \mathbf{k}^{(1)}, \cdots, \mathbf{k}^{(N)}) \\
= \sum_{i=1}^N \delta^2(\mathbf{k} - \mathbf{k}^{(i)})  \delta^2(\mathbf{x} - \mathbf{x}^{(i)}).
\end{multline}
Here $\textbf{k}^{(j)}$ denotes the gauge-invariant momentum. That means that the Poisson brackets of the $\textbf{x}^{(j)}$'s and $\textbf{k}^{(j)}$'s get modified due to the magnetic field, as we discuss below.

Now we can compute $\{ n_f, n_g \}$. We will first do it for a single-particle phase space ($\mathbf{k}^{(1)}, \mathbf{x}^{(1)})$. Then we have
\begin{align}
n_f(\textbf{k}^{(1)}, \textbf{x}^{(1)}) = f_{+}(\textbf{k}^{(1)})
\end{align}
and hence
\begin{align}
\{ n_f, n_g \} &= \frac{\partial f_{+}(\textbf{k}^{(1)})}{\partial k^{(1)}_\alpha} \frac{\partial g_{+}(\textbf{k}^{(1)})}{\partial k^{(1)}_\beta} \{ k^{(1)}_\alpha, k^{(1)}_\beta \} \\
&= -B(\mathbf{x}^{(1)}) [\nabla_{\textbf{k}^{(1)}} f_{+}(\textbf{k}^{(1)})] \times [\nabla_{\textbf{k}^{(1)}} g_{+}(\textbf{k}^{(1)})],
\end{align}
where we used the fact that the momenta $\textbf{k}^{(1)}$ obey the magnetic algebra
\footnote{In general there there can be a Berry curvature correction to this equation \cite{Duval_0506}, but it enters only at next-leading order in the magnetic field strength, so we disregard it here.}
\begin{equation}
\{ k^{(1)}_\alpha, k^{(1)}_\beta \} =  -qB(\textbf{x}^{(1)}) \varepsilon_{\alpha \beta}
\end{equation}
(here we have chosen the convention that $\textbf{k}$ carries units of momentum, not inverse length.)
For general $N$, we have
\begin{equation}
\label{eq:apoisson}
\{ n_f, n_g \} = -q\int B(\textbf{x}) n(\textbf{x}, \textbf{k}) [ \nabla_{\textbf{k}} f_{+}(\textbf{k}) ] \times [ \nabla_{\textbf{k}} g_{+}(\textbf{k}) ] d^2 \textbf{x} d^2 \textbf{k},
\end{equation}
where in this equation both the left- and right-hand sides have an implicit dependence on the phase space variables $(\textbf{k}^{(1)}, \cdots, \textbf{k}^{(N)}, \textbf{x}^{(1)}, \cdots, \textbf{x}^{(N)})$, on the right-hand side through $n(\textbf{x}, \textbf{k})$.

Now let us write $n(\textbf{x}, \textbf{k}) = n_0(\textbf{k}) + \delta n(\textbf{x}, \textbf{k})$, where $n_0(\textbf{k})$ has no implicit dependence on the phase space variables and represents the equilibrium distribution of particles in the ground state (keeping in mind the Pauli exclusion principle). For low-energy configurations, the contribution of $\delta n$ to \eqnref{eq:apoisson} will be negligible compared to $n_0$. Hence, we can replace $n$ with $n_0$ in \eqnref{eq:apoisson} and we find
\begin{align}
\{ n_f, n_g \} &= -q\Phi \int n_0(\textbf{k}) [\nabla_{\textbf{k}} f_{+}(\textbf{k}) ] \times [ \nabla_{\textbf{k}} g_{+}(\textbf{k}) ] d^2 \textbf{k}, \\
&= -q\Phi \int n_0(\textbf{k}) \nabla_{\textbf{k}} \times [ f_{+}(\textbf{k}) \nabla_{\textbf{k}} g_{+}(\textbf{k})] d^2 \textbf{k}.
\end{align}
Now we observe that $n_0(\textbf{k}) = 1/(2\pi\hbar)^2$ in the occupied region of the Brillouin zone, and $0$ outside. So by Stokes' theorem we obtain
\begin{align}
\{ n_f, n_g \} &= -\frac{qm\Phi}{(2\pi \hbar)^2} \oint_{\mbox{\footnotesize Fermi surface}} [f_{+}(\mathbf{k}) \nabla_{\textbf{k}} g_{+}(\textbf{k})] \cdot d\textbf{k} \\
&= -\frac{qm\Phi}{(2\pi \hbar)^2} \int f(\theta) \frac{dg(\theta)}{d\theta} d\theta,
\end{align}
where $m = \pm 1$ (relative to an arbitrary choice of orientation of the Fermi surface, $m = +1$ or $-1$ corresponds to which side of the Fermi surface the occupied region of the Brillouin zone is on). This is the result we wanted to show.

\section{Proof of Theorem \ref{thm:mainthm}}
\label{appendix:proof_of_mainthm}
In this appendix, we give a proof of Theorem \ref{thm:mainthm} from Section \ref{sec:generic}. The result in fact follows directly from the following mathematical property of the function $\alpha$ defined in Section \ref{subsec:1d_filling_constraints} for spatial dimension $d=1$, in Section \ref{subsec:2d_filling_constraint} for $d=2$, and in Appendix \ref{appendix:topological_action} for general $d$.
\begin{lemma}
\label{lemma:linearity}
Let $\mathcal{A}$ be an Abelian subgroup of $\GIR$.
The function $\alpha(\mathcal{Q}|\tau_1, \cdots, \tau_d)$ defined in Section \ref{sec:filling_constraints} is linear in each argument when restricted to $\mathcal{A}$. That is, for $a_1, \cdots, a_d, a_1' \in \mathcal{A}$, and where $\mathcal{Q}$ is the infinitesimal generator of a $\UU(1)$ subgroup of $\mathcal{A}$, then we have
\begin{multline}
\alpha(\mathcal{Q}|a_1 a_1', a_2, \cdots, a_d) = \alpha(\mathcal{Q}|a_1, a_2, \cdots, a_d) \\+ \alpha(\mathcal{Q}|a_1', a_2, \cdots, a_d) \quad (\mathrm{mod}\, 2\pi),
\end{multline}
and similarly for each of the arguments after the ``$|$''.
\begin{proof}
It is easy to verify this given the concrete definitions of $\alpha$ in the cases $d=1$ and $d=2$. In general $d$, where $\alpha$ is defined as discussed in Appendix \ref{appendix:topological_action}, the result follows from the hypothesis that the topological action is a cobordism invariant \cite{Kapustin_1403_1467,Kapustin_1406,Freed_1604}.
\end{proof}
\end{lemma}

Lemma \ref{lemma:linearity} is all we need to give:
\begin{proof}[Proof of Theorem \ref{thm:mainthm}]
Let $\tau_1, \dots, \tau_d \in \GIR$ be the IR symmetries corresponding to the microscopic translation symmetries and let $\hQIR$ be the infinitesimal generator of the IR symmetry corresponding to the microscopic $\UU(1)$ rotation symmetry. Since the microscopic translations commute with each other and with charge, $\tau_1, \cdots, \tau_d$ and $\hQIR$ must all commute. Let $\mathcal{A}$ be a maximal Abelian subgroup of $\GIR$ (that is, a subgroup that cannot be enlarged as a subgroup of $\GIR$ while remaining Abelian) that contains $\tau_1, \cdots, \tau_d$ and the subgroup generated by $\hQIR$. One can show\footnote{By using the Closed Subgroup Theorem for Lie groups \cite{ManifoldBook}.} that $\mathcal{A}$ is a compact Abelian Lie subgroup of $\GIR$. Let $C_{\mathcal{A}}$ be the number of connected components of $\mathcal{A}$, and let $\mathcal{A}_0 \leq \mathcal{A}$ be the connected component of the identity element. Then since $[a]^{C_\mathcal{A}} = 1$ for any element $[a] \in \mathcal{A} /\mathcal{A}_0$, we conclude that $a^{C_{\mathcal{A}}} \in \mathcal{A}_0$ for any $a \in \mathcal{A}$.

Now since $\mathcal{A}$ is a compact Abelian Lie group, the connected component $\mathcal{A}_0$ must be isomorphic to $\UU(1)^r$ for some $r$. Hence, we can write elements of $\mathcal{A}_0$ as $(\theta_1, \cdots, \theta_r)$ for angular variables $\theta_1, \cdots, \theta_r$. We now return to the function $\alpha(\hQIR|\tau_1, \cdots, \tau_d)$ that computes the filling. Henceforth, we will pick a given $\hQIR$ and hold it fixed, therefore we write simply $\alpha(\tau_1, \dots, \tau_d)$. The linearity of $\alpha$ established by Lemma \ref{lemma:linearity} implies that $\alpha$ is continuous in each argument, and, moreover, that if $a_1 = (\theta_1, \cdots, \theta_r) \in \mathcal{A}_0$, then
\begin{equation}
\alpha(a_1, \cdots, a_d) = \sum_{j=1}^r \theta_j m_j(a_2, \cdots, a_d) \quad (\mathrm{mod} \, 2\pi),
\end{equation}
for some continuous integer function $m_j : \mathcal{A}^{\times (d-1)} \to \mathbb{Z}$, where we have suppressed the dependences on $\hQIR$, which we hold fixed. But now continuity implies that $m_j$ is zero if any of its arguments are in $\mathcal{A}_0$. We therefore conclude that $\alpha(a_1, \cdots, a_d) = 0 \quad (\mathrm{mod} \, 2\pi)$ if any \emph{two} of its arguments are in $\mathcal{A}_0$. Therefore:
\begin{multline}
(C_\mathcal{A})^2 \alpha(\tau_1, \cdots, \tau_d)= \alpha(\tau_1^{C_{\mathcal{A}}}, \tau_2^{C_{\mathcal{A}}}, \tau_3, \cdots, \tau_d)  = 0\\
(\mathrm{mod} \, 2\pi),
\end{multline}
and hence the filling $\nu$ is an integer multiple of $1/(C_{\mathcal{A}})^2$.

Finally, to see that we can take $N_{\GIR}$ finite for any given compact Lie group $\GIR$, we note\footnote{This follows \cite{MathOverflow} from the ``neighboring subgroup theorem'' of Ref.~\cite{Montgomery_1942}} that for any such $\GIR$,  there are only finitely many conjugacy classes of maximal Abelian Lie subgroups $\mathcal{A} \leq \GIR$.
\end{proof}

For an alternative perspective that gives more precise bounds on $N_{\GIR}$, see the next appendix, Appendix \ref{appendix:chern_simons}.

\section{Bounds on Filling Denominators for Compact $\GIR$}
\label{appendix:chern_simons}

Theorem \ref{thm:mainthm} says that a compact $\GIR$ is inconsistent with an irrational filling fraction. In this appendix, we will discuss how, given an assumed $\GIR$, in a fixed space dimension $d$, one can compute a denominator $N_{\GIR}$ such that $\GIR$ is consistent only with rational filling fractions of the form $m/N_{\GIR}$, for $m \in \mathbb{Z}$. We will show the following

\begin{thm}
In $d>1$ space dimensions, for any compact emergent symmetry $\GIR$ (which could be a subgroup of the full group of emergent symmetries) containing the translation generators $\tau_1,\ldots,\tau_d$ and the subgroup $U(1)$ corresponding to microscopic particle-number conservation, the maximum denominator $N_{\GIR}$ divides the largest order of elements in the torsion part of the (co)bordism group $\Omega_{d+2}(\GIR)^{\rm tors} = \Omega^{d+2}(\GIR)^{\rm tors}$, ie. the group of $d+1$-space-dimensional SPT phases such that for some $n$, a stack of $n$ copies of the SPT is trivial.
\end{thm}

\begin{thm}
In $d = 2$ space dimensions, assuming a compact emergent symmetry $\GIR$ satisfying the conditions of Theorem 2 and furthermore of the form $\GIR = H \times U(1)$, the maximum denominator $N_{\GIR}$ divides $|\pi_0(\GIR)| \cdot |\pi_1(\GIR)^{\rm tors}|$, where $\pi_0(\GIR)$ is the set of connected components of $\GIR$ and $\pi_1(\GIR)^{\rm tors}$ is the torsion part of its fundamental group. If there is also a finite 1-form symmetry $K$, then $N_{\GIR}$ divides $|\pi_0(\GIR)| \cdot |\pi_1(\GIR)^{\rm tors}| \cdot |K|$.
\end{thm}

\begin{thm}
In $d = 3$ space dimensions, assuming a compact emergent symmetry $\GIR$ satisfying the conditions of Theorem 2 and furthermore of the form $\GIR = H \times U(1)$, the maximum denominator $N_{\GIR}$ divides $|\pi_0(\GIR)| \cdot gcd(|\pi_1(\GIR)^{\rm tors}|,|\pi_0(\GIR)|)$, where $\pi_0(\GIR)$ is the set of connected components of $\GIR$ and $\pi_1(\GIR)^{\rm tors}$ is the torsion part of its fundamental group and gcd denotes the greatest common divisor. In particular, for connected $\GIR$, $N_{\GIR} = 1$, so a compact connected emergent symmetry is incompatible with any fractional filling. With a finite 1-form symmetry $K$, then $N_{\GIR}$ divides $|\pi_0(\GIR)| \cdot gcd(|\pi_1(\GIR)^{\rm tors}|,|\pi_0(\GIR)|) \cdot gcd(|K|,\pi_0(\GIR))$. If there is furthermore a finite 2-form symmetry $J$, then $N_{\GIR}$ divides $|\pi_0(\GIR)| \cdot gcd(|\pi_1(\GIR)^{\rm tors}|,|\pi_0(\GIR)|) \cdot gcd(|K|,|\pi_0(\GIR)|) \cdot |J|$.
\end{thm}

We note that since $\Omega_{d+2}(\GIR)^{\rm tors}$ is finite, Theorem \ref{thm:mainthm} follows from Theorem 2.

\emph{Proof of Theorem 2:} As argued in Appendix \ref{appendix:topological_action}, the filling fraction associated to $\GIR$ may be computed from the emergent $\GIR$ 't Hooft anomaly as well as the data of how $U(1)$ and translation symmetries are realized in $\GIR$. In particular, the translation generators $\tau_1, \ldots, \tau_d$ and the $U(1)$ generator $\hat Q$ define a $\GIR$ gauge field on $T^{d+2} = T^d \times T^2$, where the $\tau_j$ define the fluxes through the $d$ coordinate 1-cycles of $T^d$ and $\hat Q$ defines a $2\pi$ flux of $T^2$. The emergent 't Hooft anomaly is associated with a cobordism invariant of $\GIR$ gauge fields on $d+2$ manifolds (the partition function of the $((d+1)+1)$D-dimensional SPT associated with the 't Hooft anomaly), and that cobordism invariant evaluated on this particular background equals $e^{2\pi i \nu}$.

Let us denote the group of cobordism invariants by $\Omega^{d+2}(\GIR)$. In the notation of \cite{Kapustin_1406}, it is $\Omega^{d+2}_{\rm spin}(B(\GIR/\mathbb{Z}_2^F),\xi),$ where $\xi$ is an orientable bundle whose 2nd Stiefel-Whitney class classifies the extension of $\GIR/\mathbb{Z}_2^F$ by fermion parity $\mathbb{Z}_2^F$. The group $\Omega^{d+2}(\GIR)$ sits in a (split) short exact sequence \cite{Freed_1604}
\[
\begin{tikzcd}
{\rm Hom}(\Omega_{d+3}(\GIR),\mathbb{Z}) \arrow{r}  &  \Omega^{d+2}(\GIR) \arrow{d} \\
 & {\rm Hom}(\Omega_{d+2}(\GIR)^{\rm tors},U(1))
\end{tikzcd}
\]
where $\Omega_{d+2}(\GIR)$ and $\Omega_{d+3}(\GIR)$ are the associated bordism groups of manifolds and $\Omega_{d+2}(\GIR)^{\rm tors}$ indicates the torsion subgroup of $\Omega_{d+2}(\GIR)$, elements $x$ for which there exists an $n$ such that $n\cdot x = 0$.

The kernel of this extension, ${\rm Hom}(\Omega_{d+3}(\GIR),\mathbb{Z})$, is characterized by generalized Chern-Simons invariants, described by Chern-Weil forms evaluated on a bounding $d+3$-manifold $M$ with $\partial M = T^{d+2}$. We will show that for $d > 1$ these invariants do not contribute to the filling fraction. More precisely, the map $\Omega^{d+2}(\GIR) \to U(1)$ given by evaluating the cobordism invariant on such a $T^{d+2}$ background as described above sends the kernel to zero. Therefore, the maximum denominator $N_{\GIR}$ divides $|{\rm Hom}(\Omega_{d+2}(\GIR)^{\rm tors},U(1))| = |\Omega_{d+2}(\GIR)^{\rm tors}|$, which is finite. This also gives an alternative proof of Theorem \ref{thm:mainthm}. In fact, $N_{\GIR}$ divides the largest order of elements in $\Omega_{d+2}(\GIR)^{\rm tors}$, which we will use below. This group is isomorphic to its cobordism partner $\Omega^{d+2}(\GIR)^{\rm tors}$ because the short exact sequence above splits.

Now we will show that for $d > 1$ the generalized Chern-Simons invariants do not contribute to the filling fraction. For even $d$, there are no generalized Chern-Simons invariants and we are done, so let us consider odd $d$ instead. We will use a flat $M$ so that there are no metric contributions to the Chern-Simons invariant. We will also use the fact that the remaining Chern-Simons invariants are defined for $\GIR$ bundles by taking trace in a unitary representation $R:\GIR \to U(n)$ and is therefore equal to a Chern-Simons invariant $U(n)$ bundle associated to it by the map $R$. Then, we use the fact that $\tau_1,\ldots, \tau_d$ and the $U(1)$ generator $\hat Q$ commute in $\GIR$, hence their images in $U(n)$ also commute and can be simultaneously diagonalized. This reduces the calculation to an Abelian Chern-Simons invariant for the subgroup $U(1)^n$ of diagonal matrices in $U(n)$. Taking the trace reduces this to a calculation for a $U(1)$ bundle.

We can extend our $U(1)$ bundle from $T^{d+2}= T^d \times T^2$ to $D^2 \times T^{d-1} \times T^2$, where $D^2$ is a disc whose boundary is the first coordinate cycle of $T^d$. The extended gauge field $A$ is defined so that it has curvature only along $D^2$ (integrating to ${\rm Tr}_R (\tau_1)$ as required by the extension) and along $T^2$ (where it has some Chern number). The Chern-Simons invariant is proportional to
\[\int_{D^2 \times T^{d-1} \times T^2} (dA)^{\frac{d+3}{2}}.\]
However, since the 2-form $dA$ has only nonzero components in 4 directions, $D^2$ and $T^2$, its power in the integrand above vanishes if $d+3>4$, ie. if $d > 1$ \footnote{We note that for $d = 1$, the Chern-Simons invariant is an integer multiple of ${\rm Tr}_R(\tau_1)$.}. This proves the claim. \qed

\emph{Proof of Theorem 3:} We can get a handle on the maximum denominator $N_{\GIR}$, which we have proven divides $|\Omega_{d+2}(\GIR)^{\rm tors}|$, using the Atiyah-Hirzebruch spectral sequence for computing the latter. This gives us an upper bound
\[N_{\GIR} \quad {\rm divides} \quad \prod_{j = 0}^{d+2} | H_j(B\GIR^b,\Omega_{d+2-j}^{\rm spin})^{\rm tors}|,\]
where $\GIR^b = \GIR/\mathbb{Z}_2^F$ and $\Omega_k^{\rm spin}$ is the usual spin bordism group, which in the low degrees of physical interest is
\[
\begin{gathered}
\Omega_0^{\rm spin} = \mathbb{Z} \\
\Omega_1^{\rm spin} = \mathbb{Z}_2 \\
\Omega_2^{\rm spin} = \mathbb{Z}_2 \\
\Omega_3^{\rm spin} = 0 \\
\Omega_4^{\rm spin} = \mathbb{Z} \\
\Omega_5^{\rm spin} = 0.
\end{gathered}
\]
This upper bound can be further improved by taking into account the known differentials of the dual spectral sequence for cobordism $\Omega^{d+2}(\GIR)$, which has isomorphic torsion, $\Omega^{d+2}(\GIR)^{\rm tors} = \Omega_{d+2}(\GIR)^{\rm tors}$. See \cite{thorngren2018anomalies} for a review.

Another thing we can do is consider the special case $\GIR = H \times U(1)$, where $U(1)$ and $H$ is compact. Because $\tau_1,\ldots,\tau_d$ and $\hat Q$ must commute, they are always contained in such a group. In this case, by taking considering $d+2$-manifolds of the form $M \times T^2$ with a arbitrary $\GIR$ bundle on $M$ and a $U(1)$ bundle with Chern number 1 on $T^2$, our $d+2$-dimensional cobordism invariant defines a $d$-dimensional cobordism invariant for $M$. By similar arguments as above, this gives us the further constraint that $N_{\GIR}$ divides the element of highest order in $\Omega_d(\GIR)^{\rm tors} = \Omega^d(\GIR)^{\rm tors}$.

For $d = 2$, using the Atiyah-Hirzebruch spectral sequence for $\Omega^2(\GIR)$, we find in the case $\GIR = H \times U(1)$, $N_{\GIR}$ divides
\[|\Omega_2^{\rm spin}| \cdot |H^1(B\GIR^b,\Omega_1^{\rm spin}) | \cdot | H^3(B\GIR^b,\Omega_0^{\rm spin})^{\rm tors} |.\]
One can show that since the fermion parity is a subgroup of the $U(1)$ factor, that the first two factors do not contribute to the spin bordism because the spectral sequence has a nonzero differential there. The third part is $|\pi_0(\GIR)| \cdot |\pi_1(\GIR)^{\rm tors}|$ torsion, so
\[N_{\GIR} \quad {\rm divides} \quad |\pi_0(\GIR)| \cdot |\pi_1(\GIR)^{\rm tors}|.\]
We note we can often get even better bounds if we better understand the cobordism invariants for $\GIR$. For example, if $\GIR = \mathbb{Z}_2^r \times U(1)$, one can show the maximum denominator is $N_{\GIR} = 2$, even though this group has $2^r$ components.

To include a finite 1-form symmetry $K$, there is one more term which contributes to the spectral sequence, $H^3(B^2K,\mathbb{Z}) = K$, which is $|K|$ torsion.\qed

\emph{Proof of Theorem 4:} For $d = 3$, $\GIR = H \times U(1)$, the contributions to the Atiyah-Hirzebruch spectral sequence for $\Omega^3(\GIR)$ are
\[|H^1(B\GIR^b,\Omega_2^{\rm spin}) | \cdot |H^2(B\GIR^b,\Omega_1^{\rm spin}) | \cdot | H^4(B\GIR^b,\Omega_0^{\rm spin})^{\rm tors} |.\]
As in $d = 2$, the first factor does not contribute. The second factor has one piece that can contribute, coming from the map $H^2(BU(1),\mathbb{Z}) \to H^2(B\GIR^b,\mathbb{Z}_2)$ but it turns out to give a spin Chern-Simons term for the $U(1)$ factor and we can disregard it, ie. it does not contribute to the torsion in $\Omega^3(\GIR)$. The third piece can be bounded using the Serre spectral sequence for the extension
\[\GIR^{b,0} \to \GIR^b \to \pi_0(\GIR^b) = \pi_0(\GIR),\]
from which we find it is $|\pi_0(\GIR)| \cdot gcd(|\pi_1(\GIR)^{\rm tors}|,|\pi_0(\GIR)|)$ torsion, so
\[N_{\GIR} \quad {\rm divides} \quad |\pi_0(\GIR)| \cdot gcd(|\pi_1(\GIR)^{\rm tors}|,|\pi_0(\GIR)|).\]
Note that when $\GIR$ is connected, we re-establish the result $N_{\GIR} = 1$: connected, compact $\GIR$ are incompatible with fractional filling in $d = 3$. 

When there is also a finite 1-form symmetry $K$, it contributes a term to $H^4(B\GIR^b,\Omega_0^{\rm spin})$ via $H^1(B\pi_0(\GIR^b),H^3(B^2K,\mathbb{Z}))$, which is $gcd(|\pi_0(\GIR^b)|,|K|)$ torsion.

When there is also a finite 2-form symmetry $J$, it contributes a term to $H^4(B\GIR^b,\Omega_0^{\rm spin})$ via $H^4(B^2J,\mathbb{Z}) = J$.
\qed

Finally we note that the bounds so derived hold for \emph{every subgroup} $\GIR$ of the full emergent symmetry, so long as it contains the translation generators $\tau_1,\ldots,\tau_d$ and the $U(1)$ particle-number symmetry. Thus, one does not need to know the entire emergent symmetry to obtain a bound. Further, by varying $\GIR$, one can sometimes find a better bound for the problem at hand.

\section{Asymptotic analysis for quantum oscillations}
\label{appendix:oscillations}
In this appendix, we will derive the asymptotic form for quantum oscillations stated in Section \ref{subsubsec:quantum_oscillations_3d}. Abstracting out from the specific details, the situation is that we have some functional $\mathcal{F}[f]$, where $f$ is a function of a single variable into $\mathbb{R}/\mathbb{Z}$ that can be expressed as
\begin{equation}
f(x) = t g(x),
\end{equation}
and we we wish to extract the asymptotic dependence of $\mathcal{F}[f]$ as $t \to \infty$ while keeping the function $g$ fixed.

First we imagine approximating $f$ by its values on the discrete points $x_1, \cdots, x_k$. Then we can expand $\mathcal{F}$ as a Fourier series
\begin{equation}
\mathcal{F}[f] = \sum_{\mathbf{n} \in \mathbb{Z}^k} a_{\mathbf{n}} e^{2\pi it ( n_1 g(x_1) + \cdots + n_k g(x_k))},
\end{equation}
where the sum is over integer vectors $\mathbf{n} \in \mathbb{Z}^k$.
We can rewrite this as
\begin{equation}
\mathcal{F}[f] = \mathcal{S}^{(0)} + \sum_{n \neq 0} \mathcal{S}_n^{(1)} + \sum_{n_1 \neq 0, n_2 \neq 0}^{\infty} \mathcal{S}_{n_1, n_2}^{(2)} + \cdots,
\end{equation}
where $\mathcal{S}^{(0)} = a_{\textbf{0}}$ and
\begin{align}
\mathcal{S}^{(1)}_n&= \sum_{j=1}^k a_{(j:n)} e^{2\pi itm g(x_j)} \\
\mathcal{S}^{(2)}_{n_1,n_2} &= \sum_{j_1 \neq j_2}  a_{(j_1: n_1, j_2: n_2)} e^{2\pi it (n_1 g(x_1) + it n_2 g(x_2))} \\
\mathcal{S}^{(3)}_{n_1, n_2, n_3} &= \cdots,
\end{align}
where $(j_1 : m_1, j_2 : m_2)$, for example, is the vector $\mathbf{n}$ obtained by setting $n_{j_1} = m_1$ and $n_{j_2} = m_2$ and $n_{j} = 0$ for $j \notin \{ m_1, m_2 \}$. By taking the limit as $k \to \infty$ and the $x_j$'s become dense, we find that $\mathcal{S}_\ell \to \mathcal{I}_{\ell}$, where

\begin{align}
\mathcal{I}^{(1)}_{n} &= \int dx \, a_n(x) e^{2\pi it n g(x)} \\
\mathcal{I}^{(2)}_{n_1, n_2} &=  \int dx_1 dx_2 a_{n_1, n_2}(x_1, x_2) e^{2\pi it(n_1 g(x_1) + n_2 g(x_2))}, \\
\mathcal{I}^{(3)}_{n_1, n_2, n_3} &= \cdots
\end{align}

Now we invoke the theory of stationary phase integrals which tells us that, for a function $h(\mathbf{x})$ of a $q$-dimensional variable $\mathbf{x}$, then as $t \to \infty$
\begin{equation}
\int d^q \mathbf{x} \varphi(\textbf{x}) e^{2\pi it h(\mathbf{x})} \sim \sum_{\textbf{x}_* \in \boldsymbol{\Sigma}} c_{\mathbf{x}_*} t^{-q/2} e^{2\pi it f(\mathbf{x}_*)},
\end{equation}
for some constants $c_{\mathbf{x}_*}$ and where the sum is over the set $\boldsymbol{\Sigma}$ of solutions to $\nabla f(\mathbf{x}) = 0$. Hence, we find that
\begin{align}
\mathcal{I}_n^{(1)} &\sim c_{n}^{(1)} t^{-1/2} \sum_{x_* \in \Sigma} e^{2\pi itn g(x_*)} \\
\mathcal{I}_{n_1,n_2}^{(2)} &\sim c_{n_1, n_2}^{(2)} t^{-1} \sum_{x_*, x_*' \in \Sigma} e^{2\pi it (n_1 g(x_*) + n_2 g(x_*'))}, \\
\mathcal{I}_{n_1,n_2,n_3}^{(3)} &\sim \cdots
\end{align}
where $\Sigma$ is the set of solutions to $g'(x) = 0$. The oscillatory part of each of these expressions is consistent with a quasiperiodic function of $t$ with base frequencies $\omega(x_*) = 2\pi g(x_*)$, $x_* \in \Sigma$.

\bibliography{ref-autobib,ref-manual}

\begin{thebibliography}{124}%
\makeatletter
\providecommand \@ifxundefined [1]{%
 \@ifx{#1\undefined}
}%
\providecommand \@ifnum [1]{%
 \ifnum #1\expandafter \@firstoftwo
 \else \expandafter \@secondoftwo
 \fi
}%
\providecommand \@ifx [1]{%
 \ifx #1\expandafter \@firstoftwo
 \else \expandafter \@secondoftwo
 \fi
}%
\providecommand \natexlab [1]{#1}%
\providecommand \enquote  [1]{``#1''}%
\providecommand \bibnamefont  [1]{#1}%
\providecommand \bibfnamefont [1]{#1}%
\providecommand \citenamefont [1]{#1}%
\providecommand \href@noop [0]{\@secondoftwo}%
\providecommand \href [0]{\begingroup \@sanitize@url \@href}%
\providecommand \@href[1]{\@@startlink{#1}\@@href}%
\providecommand \@@href[1]{\endgroup#1\@@endlink}%
\providecommand \@sanitize@url [0]{\catcode `\\12\catcode `\$12\catcode
  `\&12\catcode `\#12\catcode `\^12\catcode `\_12\catcode `\%12\relax}%
\providecommand \@@startlink[1]{}%
\providecommand \@@endlink[0]{}%
\providecommand \url  [0]{\begingroup\@sanitize@url \@url }%
\providecommand \@url [1]{\endgroup\@href {#1}{\urlprefix }}%
\providecommand \urlprefix  [0]{URL }%
\providecommand \Eprint [0]{\href }%
\providecommand \doibase [0]{http://dx.doi.org/}%
\providecommand \selectlanguage [0]{\@gobble}%
\providecommand \bibinfo  [0]{\@secondoftwo}%
\providecommand \bibfield  [0]{\@secondoftwo}%
\providecommand \translation [1]{[#1]}%
\providecommand \BibitemOpen [0]{}%
\providecommand \bibitemStop [0]{}%
\providecommand \bibitemNoStop [0]{.\EOS\space}%
\providecommand \EOS [0]{\spacefactor3000\relax}%
\providecommand \BibitemShut  [1]{\csname bibitem#1\endcsname}%
\let\auto@bib@innerbib\@empty
%</preamble>
\bibitem [{\citenamefont {Lieb}\ \emph {et~al.}(1961)\citenamefont {Lieb},
  \citenamefont {Schultz},\ and\ \citenamefont {Mattis}}]{Lieb_1961}%
  \BibitemOpen
  \bibfield  {author} {\bibinfo {author} {\bibfnamefont {Elliott}\ \bibnamefont
  {Lieb}}, \bibinfo {author} {\bibfnamefont {Theodore}\ \bibnamefont
  {Schultz}}, \ and\ \bibinfo {author} {\bibfnamefont {Daniel}\ \bibnamefont
  {Mattis}},\ }\bibfield  {title} {\enquote {\bibinfo {title} {{Two soluble
  models of an antiferromagnetic chain}},}\ }\href {\doibase
  10.1016/0003-4916(61)90115-4} {\bibfield  {journal} {\bibinfo  {journal}
  {Ann. Phys.}\ }\textbf {\bibinfo {volume} {16}},\ \bibinfo {pages} {407}
  (\bibinfo {year} {1961})}\BibitemShut {NoStop}%
\bibitem [{\citenamefont {Oshikawa}\ \emph {et~al.}(1997)\citenamefont
  {Oshikawa}, \citenamefont {Yamanaka},\ and\ \citenamefont
  {Affleck}}]{Oshikawa_9610}%
  \BibitemOpen
  \bibfield  {author} {\bibinfo {author} {\bibfnamefont {Masaki}\ \bibnamefont
  {Oshikawa}}, \bibinfo {author} {\bibfnamefont {Masanori}\ \bibnamefont
  {Yamanaka}}, \ and\ \bibinfo {author} {\bibfnamefont {Ian}\ \bibnamefont
  {Affleck}},\ }\bibfield  {title} {\enquote {\bibinfo {title} {Magnetization
  plateaus in spin chains: ``{Haldane} gap'' for half-integer spins},}\ }\href
  {\doibase 10.1103/PhysRevLett.78.1984} {\bibfield  {journal} {\bibinfo
  {journal} {Phys. Rev. Lett.}\ }\textbf {\bibinfo {volume} {78}},\ \bibinfo
  {pages} {1984} (\bibinfo {year} {1997})},\ \Eprint
  {http://arxiv.org/abs/cond-mat/9610168} {arXiv:cond-mat/9610168} \BibitemShut
  {NoStop}%
\bibitem [{\citenamefont {Oshikawa}(2000{\natexlab{a}})}]{Oshikawa_9911}%
  \BibitemOpen
  \bibfield  {author} {\bibinfo {author} {\bibfnamefont {Masaki}\ \bibnamefont
  {Oshikawa}},\ }\bibfield  {title} {\enquote {\bibinfo {title}
  {Commensurability, excitation gap, and topology in quantum many-particle
  systems on a periodic lattice},}\ }\href {\doibase
  10.1103/PhysRevLett.84.1535} {\bibfield  {journal} {\bibinfo  {journal}
  {Phys. Rev. Lett.}\ }\textbf {\bibinfo {volume} {84}},\ \bibinfo {pages}
  {1535} (\bibinfo {year} {2000}{\natexlab{a}})},\ \Eprint
  {http://arxiv.org/abs/cond-mat/9911137} {arXiv:cond-mat/9911137} \BibitemShut
  {NoStop}%
\bibitem [{\citenamefont {Hastings}(2005)}]{Hastings_0411}%
  \BibitemOpen
  \bibfield  {author} {\bibinfo {author} {\bibfnamefont {M.~B}\ \bibnamefont
  {Hastings}},\ }\bibfield  {title} {\enquote {\bibinfo {title} {{Sufficient
  conditions for topological order in insulators}},}\ }\href {\doibase
  10.1209/epl/i2005-10046-x} {\bibfield  {journal} {\bibinfo  {journal}
  {Europhys. Lett.}\ }\textbf {\bibinfo {volume} {70}},\ \bibinfo {pages} {824}
  (\bibinfo {year} {2005})},\ \Eprint {http://arxiv.org/abs/cond-mat/0411094}
  {arXiv:cond-mat/0411094} \BibitemShut {NoStop}%
\bibitem [{\citenamefont {Misguich}\ \emph {et~al.}(2002)\citenamefont
  {Misguich}, \citenamefont {Lhuillier}, \citenamefont {Mambrini},\ and\
  \citenamefont {Sindzingre}}]{Misguich_0112}%
  \BibitemOpen
  \bibfield  {author} {\bibinfo {author} {\bibfnamefont {G.}~\bibnamefont
  {Misguich}}, \bibinfo {author} {\bibfnamefont {C.}~\bibnamefont {Lhuillier}},
  \bibinfo {author} {\bibfnamefont {M.}~\bibnamefont {Mambrini}}, \ and\
  \bibinfo {author} {\bibfnamefont {P.}~\bibnamefont {Sindzingre}},\ }\bibfield
   {title} {\enquote {\bibinfo {title} {{Degeneracy of the ground-state of
  antiferromagnetic spin-1/2 Hamiltonians}},}\ }\href {\doibase
  10.1140/epjb/e20020078} {\bibfield  {journal} {\bibinfo  {journal} {Eur.
  Phys. J. B}\ }\textbf {\bibinfo {volume} {26}},\ \bibinfo {pages} {167}
  (\bibinfo {year} {2002})},\ \Eprint {http://arxiv.org/abs/cond-mat/0112360}
  {arXiv:cond-mat/0112360} \BibitemShut {NoStop}%
\bibitem [{\citenamefont {Luttinger}(1960)}]{Luttinger_1960}%
  \BibitemOpen
  \bibfield  {author} {\bibinfo {author} {\bibfnamefont {J.~M.}\ \bibnamefont
  {Luttinger}},\ }\bibfield  {title} {\enquote {\bibinfo {title} {{Fermi}
  surface and some simple equilibrium properties of a system of interacting
  fermions},}\ }\href {\doibase 10.1103/PhysRev.119.1153} {\bibfield  {journal}
  {\bibinfo  {journal} {Phys. Rev.}\ }\textbf {\bibinfo {volume} {119}},\
  \bibinfo {pages} {1153} (\bibinfo {year} {1960})}\BibitemShut {NoStop}%
\bibitem [{\citenamefont {Oshikawa}(2000{\natexlab{b}})}]{Oshikawa_0002}%
  \BibitemOpen
  \bibfield  {author} {\bibinfo {author} {\bibfnamefont {Masaki}\ \bibnamefont
  {Oshikawa}},\ }\bibfield  {title} {\enquote {\bibinfo {title} {Topological
  approach to {Luttinger}'s theorem and the {Fermi} surface of a kondo
  lattice},}\ }\href {\doibase 10.1103/PhysRevLett.84.3370} {\bibfield
  {journal} {\bibinfo  {journal} {Phys. Rev. Lett.}\ }\textbf {\bibinfo
  {volume} {84}},\ \bibinfo {pages} {3370} (\bibinfo {year}
  {2000}{\natexlab{b}})},\ \Eprint {http://arxiv.org/abs/cond-mat/0002392}
  {arXiv:cond-mat/0002392} \BibitemShut {NoStop}%
\bibitem [{\citenamefont {Senthil}\ \emph {et~al.}(2003)\citenamefont
  {Senthil}, \citenamefont {Sachdev},\ and\ \citenamefont
  {Vojta}}]{Senthil_0209}%
  \BibitemOpen
  \bibfield  {author} {\bibinfo {author} {\bibfnamefont {T.}~\bibnamefont
  {Senthil}}, \bibinfo {author} {\bibfnamefont {Subir}\ \bibnamefont
  {Sachdev}}, \ and\ \bibinfo {author} {\bibfnamefont {Matthias}\ \bibnamefont
  {Vojta}},\ }\bibfield  {title} {\enquote {\bibinfo {title} {Fractionalized
  {Fermi} liquids},}\ }\href {\doibase 10.1103/PhysRevLett.90.216403}
  {\bibfield  {journal} {\bibinfo  {journal} {Phys. Rev. Lett.}\ }\textbf
  {\bibinfo {volume} {90}},\ \bibinfo {pages} {216403} (\bibinfo {year}
  {2003})},\ \Eprint {http://arxiv.org/abs/cond-mat/0209144}
  {arXiv:cond-mat/0209144} \BibitemShut {NoStop}%
\bibitem [{\citenamefont {Senthil}\ \emph {et~al.}(2004)\citenamefont
  {Senthil}, \citenamefont {Vojta},\ and\ \citenamefont
  {Sachdev}}]{Senthil_0305}%
  \BibitemOpen
  \bibfield  {author} {\bibinfo {author} {\bibfnamefont {T.}~\bibnamefont
  {Senthil}}, \bibinfo {author} {\bibfnamefont {Matthias}\ \bibnamefont
  {Vojta}}, \ and\ \bibinfo {author} {\bibfnamefont {Subir}\ \bibnamefont
  {Sachdev}},\ }\bibfield  {title} {\enquote {\bibinfo {title} {{Weak magnetism
  and non-Fermi liquids near heavy-fermion critical points}},}\ }\href
  {\doibase 10.1103/PhysRevB.69.035111} {\bibfield  {journal} {\bibinfo
  {journal} {Phys. Rev. B}\ }\textbf {\bibinfo {volume} {69}},\ \bibinfo
  {pages} {035111} (\bibinfo {year} {2004})},\ \Eprint
  {http://arxiv.org/abs/cond-mat/0305193} {arXiv:cond-mat/0305193} \BibitemShut
  {NoStop}%
\bibitem [{\citenamefont {Paramekanti}\ and\ \citenamefont
  {Vishwanath}(2004)}]{Paramekanti_0406}%
  \BibitemOpen
  \bibfield  {author} {\bibinfo {author} {\bibfnamefont {Arun}\ \bibnamefont
  {Paramekanti}}\ and\ \bibinfo {author} {\bibfnamefont {Ashvin}\ \bibnamefont
  {Vishwanath}},\ }\bibfield  {title} {\enquote {\bibinfo {title} {{Extending
  Luttinger's theorem to ${Z}_{2}$ fractionalized phases of matter}},}\ }\href
  {\doibase 10.1103/PhysRevB.70.245118} {\bibfield  {journal} {\bibinfo
  {journal} {Phys. Rev. B}\ }\textbf {\bibinfo {volume} {70}},\ \bibinfo
  {pages} {245118} (\bibinfo {year} {2004})},\ \Eprint
  {http://arxiv.org/abs/cond-mat/0406619} {arXiv:cond-mat/0406619} \BibitemShut
  {NoStop}%
\bibitem [{\citenamefont {Bonderson}\ \emph {et~al.}()\citenamefont
  {Bonderson}, \citenamefont {Cheng}, \citenamefont {Patel},\ and\
  \citenamefont {Plamadeala}}]{Bonderson_1601}%
  \BibitemOpen
  \bibfield  {author} {\bibinfo {author} {\bibfnamefont {Parsa}\ \bibnamefont
  {Bonderson}}, \bibinfo {author} {\bibfnamefont {Meng}\ \bibnamefont {Cheng}},
  \bibinfo {author} {\bibfnamefont {Kaushal}\ \bibnamefont {Patel}}, \ and\
  \bibinfo {author} {\bibfnamefont {Eugeniu}\ \bibnamefont {Plamadeala}},\
  }\bibfield  {title} {\enquote {\bibinfo {title} {Topological enrichment of
  {Luttinger}'s theorem},}\ }\href@noop {} {\ }\Eprint
  {http://arxiv.org/abs/1601.07902} {arXiv:1601.07902} \BibitemShut {NoStop}%
\bibitem [{\citenamefont {Watanabe}\ \emph {et~al.}(2015)\citenamefont
  {Watanabe}, \citenamefont {Po}, \citenamefont {Vishwanath},\ and\
  \citenamefont {Zaletel}}]{Watanabe_1505}%
  \BibitemOpen
  \bibfield  {author} {\bibinfo {author} {\bibfnamefont {Haruki}\ \bibnamefont
  {Watanabe}}, \bibinfo {author} {\bibfnamefont {Hoi~Chun}\ \bibnamefont {Po}},
  \bibinfo {author} {\bibfnamefont {Ashvin}\ \bibnamefont {Vishwanath}}, \ and\
  \bibinfo {author} {\bibfnamefont {Michael}\ \bibnamefont {Zaletel}},\
  }\bibfield  {title} {\enquote {\bibinfo {title} {{Filling constraints for
  spin-orbit coupled insulators in symmorphic and nonsymmorphic crystals}},}\
  }\href {\doibase 10.1073/pnas.1514665112} {\bibfield  {journal} {\bibinfo
  {journal} {Proc. Natl. Acad. Sci.}\ }\textbf {\bibinfo {volume} {112}},\
  \bibinfo {pages} {14551} (\bibinfo {year} {2015})},\ \Eprint
  {http://arxiv.org/abs/1505.04193} {arXiv:1505.04193} \BibitemShut {NoStop}%
\bibitem [{\citenamefont {Lu}\ \emph {et~al.}(2020)\citenamefont {Lu},
  \citenamefont {Ran},\ and\ \citenamefont {Oshikawa}}]{Lu_1705}%
  \BibitemOpen
  \bibfield  {author} {\bibinfo {author} {\bibfnamefont {Yuan-Ming}\
  \bibnamefont {Lu}}, \bibinfo {author} {\bibfnamefont {Ying}\ \bibnamefont
  {Ran}}, \ and\ \bibinfo {author} {\bibfnamefont {Masaki}\ \bibnamefont
  {Oshikawa}},\ }\bibfield  {title} {\enquote {\bibinfo {title}
  {{Filling-enforced constraint on the quantized Hall conductivity on a
  periodic lattice}},}\ }\href {\doibase 10.1016/j.aop.2019.168060} {\bibfield
  {journal} {\bibinfo  {journal} {Ann. Phys.}\ }\textbf {\bibinfo {volume}
  {413}},\ \bibinfo {pages} {168060} (\bibinfo {year} {2020})},\ \Eprint
  {http://arxiv.org/abs/1705.09298} {arXiv:1705.09298} \BibitemShut {NoStop}%
\bibitem [{\citenamefont {Bultinck}\ and\ \citenamefont
  {Cheng}(2018)}]{Bultinck_1808}%
  \BibitemOpen
  \bibfield  {author} {\bibinfo {author} {\bibfnamefont {Nick}\ \bibnamefont
  {Bultinck}}\ and\ \bibinfo {author} {\bibfnamefont {Meng}\ \bibnamefont
  {Cheng}},\ }\bibfield  {title} {\enquote {\bibinfo {title} {{Filling
  constraints on fermionic topological order in zero magnetic field}},}\ }\href
  {\doibase 10.1103/PhysRevB.98.161119} {\bibfield  {journal} {\bibinfo
  {journal} {Phys. Rev. B}\ }\textbf {\bibinfo {volume} {98}},\ \bibinfo
  {pages} {161119} (\bibinfo {year} {2018})},\ \Eprint
  {http://arxiv.org/abs/1808.00324} {arXiv:1808.00324} \BibitemShut {NoStop}%
\bibitem [{\citenamefont {Cheng}\ \emph {et~al.}(2016)\citenamefont {Cheng},
  \citenamefont {Zaletel}, \citenamefont {Barkeshli}, \citenamefont
  {Vishwanath},\ and\ \citenamefont {Bonderson}}]{Cheng_1511}%
  \BibitemOpen
  \bibfield  {author} {\bibinfo {author} {\bibfnamefont {Meng}\ \bibnamefont
  {Cheng}}, \bibinfo {author} {\bibfnamefont {Michael}\ \bibnamefont
  {Zaletel}}, \bibinfo {author} {\bibfnamefont {Maissam}\ \bibnamefont
  {Barkeshli}}, \bibinfo {author} {\bibfnamefont {Ashvin}\ \bibnamefont
  {Vishwanath}}, \ and\ \bibinfo {author} {\bibfnamefont {Parsa}\ \bibnamefont
  {Bonderson}},\ }\bibfield  {title} {\enquote {\bibinfo {title} {Translational
  symmetry and microscopic constraints on symmetry-enriched topological phases:
  A view from the surface},}\ }\href {\doibase 10.1103/PhysRevX.6.041068}
  {\bibfield  {journal} {\bibinfo  {journal} {Phys. Rev. X}\ }\textbf {\bibinfo
  {volume} {6}},\ \bibinfo {pages} {041068} (\bibinfo {year} {2016})},\ \Eprint
  {http://arxiv.org/abs/1511.02263} {arXiv:1511.02263} \BibitemShut {NoStop}%
\bibitem [{\citenamefont {Cho}\ \emph {et~al.}(2017)\citenamefont {Cho},
  \citenamefont {Hsieh},\ and\ \citenamefont {Ryu}}]{Cho_1705}%
  \BibitemOpen
  \bibfield  {author} {\bibinfo {author} {\bibfnamefont {Gil~Young}\
  \bibnamefont {Cho}}, \bibinfo {author} {\bibfnamefont {Chang-Tse}\
  \bibnamefont {Hsieh}}, \ and\ \bibinfo {author} {\bibfnamefont {Shinsei}\
  \bibnamefont {Ryu}},\ }\bibfield  {title} {\enquote {\bibinfo {title}
  {{Anomaly manifestation of Lieb-Schultz-Mattis theorem and topological
  phases}},}\ }\href {\doibase 10.1103/PhysRevB.96.195105} {\bibfield
  {journal} {\bibinfo  {journal} {Phys. Rev. B}\ }\textbf {\bibinfo {volume}
  {96}},\ \bibinfo {pages} {195105} (\bibinfo {year} {2017})},\ \Eprint
  {http://arxiv.org/abs/1705.03892} {arXiv:1705.03892} \BibitemShut {NoStop}%
\bibitem [{\citenamefont {Jian}\ \emph {et~al.}(2018)\citenamefont {Jian},
  \citenamefont {Bi},\ and\ \citenamefont {Xu}}]{Jian_1705}%
  \BibitemOpen
  \bibfield  {author} {\bibinfo {author} {\bibfnamefont {Chao-Ming}\
  \bibnamefont {Jian}}, \bibinfo {author} {\bibfnamefont {Zhen}\ \bibnamefont
  {Bi}}, \ and\ \bibinfo {author} {\bibfnamefont {Cenke}\ \bibnamefont {Xu}},\
  }\bibfield  {title} {\enquote {\bibinfo {title} {{Lieb-Schultz-Mattis theorem
  and its generalizations from the perspective of the symmetry-protected
  topological phase}},}\ }\href {\doibase 10.1103/PhysRevB.97.054412}
  {\bibfield  {journal} {\bibinfo  {journal} {Phys. Rev. B}\ }\textbf {\bibinfo
  {volume} {97}},\ \bibinfo {pages} {054412} (\bibinfo {year} {2018})},\
  \Eprint {http://arxiv.org/abs/1705.00012} {arXiv:1705.00012} \BibitemShut
  {NoStop}%
\bibitem [{\citenamefont {Metlitski}\ and\ \citenamefont
  {Thorngren}(2018)}]{Metlitski_1707}%
  \BibitemOpen
  \bibfield  {author} {\bibinfo {author} {\bibfnamefont {Max~A.}\ \bibnamefont
  {Metlitski}}\ and\ \bibinfo {author} {\bibfnamefont {Ryan}\ \bibnamefont
  {Thorngren}},\ }\bibfield  {title} {\enquote {\bibinfo {title} {{Intrinsic
  and emergent anomalies at deconfined critical points}},}\ }\href {\doibase
  10.1103/PhysRevB.98.085140} {\bibfield  {journal} {\bibinfo  {journal} {Phys.
  Rev. B}\ }\textbf {\bibinfo {volume} {98}},\ \bibinfo {pages} {085140}
  (\bibinfo {year} {2018})},\ \Eprint {http://arxiv.org/abs/1707.07686}
  {arXiv:1707.07686} \BibitemShut {NoStop}%
\bibitem [{\citenamefont {Yao}\ and\ \citenamefont
  {Oshikawa}(2020)}]{Yao_1906}%
  \BibitemOpen
  \bibfield  {author} {\bibinfo {author} {\bibfnamefont {Yuan}\ \bibnamefont
  {Yao}}\ and\ \bibinfo {author} {\bibfnamefont {Masaki}\ \bibnamefont
  {Oshikawa}},\ }\bibfield  {title} {\enquote {\bibinfo {title} {Generalized
  boundary condition applied to lieb-schultz-mattis-type ingappabilities and
  many-body chern numbers},}\ }\href {\doibase 10.1103/PhysRevX.10.031008}
  {\bibfield  {journal} {\bibinfo  {journal} {Phys. Rev. X}\ }\textbf {\bibinfo
  {volume} {10}},\ \bibinfo {pages} {031008} (\bibinfo {year} {2020})},\
  \Eprint {http://arxiv.org/abs/1906.11662} {arXiv:1906.11662} \BibitemShut
  {NoStop}%
\bibitem [{\citenamefont {Song}\ \emph {et~al.}()\citenamefont {Song},
  \citenamefont {He}, \citenamefont {Vishwanath},\ and\ \citenamefont
  {Wang}}]{Song_1909}%
  \BibitemOpen
  \bibfield  {author} {\bibinfo {author} {\bibfnamefont {Xue-Yang}\
  \bibnamefont {Song}}, \bibinfo {author} {\bibfnamefont {Yin-Chen}\
  \bibnamefont {He}}, \bibinfo {author} {\bibfnamefont {Ashvin}\ \bibnamefont
  {Vishwanath}}, \ and\ \bibinfo {author} {\bibfnamefont {Chong}\ \bibnamefont
  {Wang}},\ }\bibfield  {title} {\enquote {\bibinfo {title} {{Electric
  polarization as a nonquantized topological response and boundary Luttinger
  theorem}},}\ }\href@noop {} {\ }\Eprint {http://arxiv.org/abs/1909.08637}
  {arXiv:1909.08637} \BibitemShut {NoStop}%
\bibitem [{\citenamefont {Coleman}\ \emph {et~al.}(2005)\citenamefont
  {Coleman}, \citenamefont {Paul},\ and\ \citenamefont
  {Rech}}]{coleman2005sum}%
  \BibitemOpen
  \bibfield  {author} {\bibinfo {author} {\bibfnamefont {P.}~\bibnamefont
  {Coleman}}, \bibinfo {author} {\bibfnamefont {I.}~\bibnamefont {Paul}}, \
  and\ \bibinfo {author} {\bibfnamefont {J.}~\bibnamefont {Rech}},\ }\bibfield
  {title} {\enquote {\bibinfo {title} {{Sum rules and Ward identities in the
  Kondo lattice}},}\ }\href {\doibase 10.1103/PhysRevB.72.094430} {\bibfield
  {journal} {\bibinfo  {journal} {Phys. Rev. B}\ }\textbf {\bibinfo {volume}
  {72}},\ \bibinfo {pages} {094430} (\bibinfo {year} {2005})},\ \Eprint
  {http://arxiv.org/abs/cond-mat/0503001} {arXiv:cond-mat/0503001} \BibitemShut
  {NoStop}%
\bibitem [{\citenamefont {Powell}\ \emph {et~al.}(2005)\citenamefont {Powell},
  \citenamefont {Sachdev},\ and\ \citenamefont
  {Büchler}}]{powell2005depletion}%
  \BibitemOpen
  \bibfield  {author} {\bibinfo {author} {\bibfnamefont {Stephen}\ \bibnamefont
  {Powell}}, \bibinfo {author} {\bibfnamefont {Subir}\ \bibnamefont {Sachdev}},
  \ and\ \bibinfo {author} {\bibfnamefont {Hans~Peter}\ \bibnamefont
  {Büchler}},\ }\bibfield  {title} {\enquote {\bibinfo {title} {{Depletion of
  the Bose-Einstein condensate in Bose-Fermi mixtures}},}\ }\href {\doibase
  10.1103/PhysRevB.72.024534} {\bibfield  {journal} {\bibinfo  {journal} {Phys.
  Rev. B}\ }\textbf {\bibinfo {volume} {72}},\ \bibinfo {pages} {024534}
  (\bibinfo {year} {2005})},\ \Eprint {http://arxiv.org/abs/cond-mat/0502299}
  {arXiv:cond-mat/0502299} \BibitemShut {NoStop}%
\bibitem [{\citenamefont {Huijse}\ and\ \citenamefont
  {Sachdev}(2011)}]{huijse2011fermi}%
  \BibitemOpen
  \bibfield  {author} {\bibinfo {author} {\bibfnamefont {Liza}\ \bibnamefont
  {Huijse}}\ and\ \bibinfo {author} {\bibfnamefont {Subir}\ \bibnamefont
  {Sachdev}},\ }\bibfield  {title} {\enquote {\bibinfo {title} {{Fermi surfaces
  and gauge-gravity duality}},}\ }\href {\doibase 10.1103/PhysRevD.84.026001}
  {\bibfield  {journal} {\bibinfo  {journal} {Phys. Rev. D}\ }\textbf {\bibinfo
  {volume} {84}},\ \bibinfo {pages} {026001} (\bibinfo {year} {2011})},\
  \Eprint {http://arxiv.org/abs/1104.5022} {arXiv:1104.5022} \BibitemShut
  {NoStop}%
\bibitem [{\citenamefont {Potthoff}(2006)}]{potthoff2004non}%
  \BibitemOpen
  \bibfield  {author} {\bibinfo {author} {\bibfnamefont {M.}~\bibnamefont
  {Potthoff}},\ }\bibfield  {title} {\enquote {\bibinfo {title}
  {Non-perturbative construction of the luttinger-ward functional},}\ }\href
  {\doibase 10.5488/CMP.9.3.557} {\bibfield  {journal} {\bibinfo  {journal}
  {Condens. Matt. Phys.}\ }\textbf {\bibinfo {volume} {9}},\ \bibinfo {pages}
  {557} (\bibinfo {year} {2006})},\ \Eprint
  {http://arxiv.org/abs/cond-mat/0406671} {cond-mat/0406671} \BibitemShut
  {NoStop}%
\bibitem [{\citenamefont {Gaiotto}\ \emph {et~al.}(2015)\citenamefont
  {Gaiotto}, \citenamefont {Kapustin}, \citenamefont {Seiberg},\ and\
  \citenamefont {Willett}}]{Gaiotto_1412}%
  \BibitemOpen
  \bibfield  {author} {\bibinfo {author} {\bibfnamefont {Davide}\ \bibnamefont
  {Gaiotto}}, \bibinfo {author} {\bibfnamefont {Anton}\ \bibnamefont
  {Kapustin}}, \bibinfo {author} {\bibfnamefont {Nathan}\ \bibnamefont
  {Seiberg}}, \ and\ \bibinfo {author} {\bibfnamefont {Brian}\ \bibnamefont
  {Willett}},\ }\bibfield  {title} {\enquote {\bibinfo {title} {{Generalized
  global symmetries}},}\ }\href {\doibase 10.1007/JHEP02(2015)172} {\bibfield
  {journal} {\bibinfo  {journal} {J. High Energy Phys.}\ }\textbf {\bibinfo
  {volume} {2015}},\ \bibinfo {pages} {172} (\bibinfo {year} {2015})},\ \Eprint
  {http://arxiv.org/abs/1412.5148} {arXiv:1412.5148} \BibitemShut {NoStop}%
\bibitem [{\citenamefont {Chen}\ \emph
  {et~al.}(2011{\natexlab{a}})\citenamefont {Chen}, \citenamefont {Liu},\ and\
  \citenamefont {Wen}}]{Chen_1106_4752}%
  \BibitemOpen
  \bibfield  {author} {\bibinfo {author} {\bibfnamefont {Xie}\ \bibnamefont
  {Chen}}, \bibinfo {author} {\bibfnamefont {Zheng-Xin}\ \bibnamefont {Liu}}, \
  and\ \bibinfo {author} {\bibfnamefont {Xiao-Gang}\ \bibnamefont {Wen}},\
  }\bibfield  {title} {\enquote {\bibinfo {title} {{Two-dimensional
  symmetry-protected topological orders and their protected gapless edge
  excitations}},}\ }\href {\doibase 10.1103/PhysRevB.84.235141} {\bibfield
  {journal} {\bibinfo  {journal} {Phys. Rev. B}\ }\textbf {\bibinfo {volume}
  {84}},\ \bibinfo {pages} {235141} (\bibinfo {year} {2011}{\natexlab{a}})},\
  \Eprint {http://arxiv.org/abs/1106.4752} {arXiv:1106.4752} \BibitemShut
  {NoStop}%
\bibitem [{\citenamefont {Else}\ and\ \citenamefont {Nayak}(2014)}]{Else_1409}%
  \BibitemOpen
  \bibfield  {author} {\bibinfo {author} {\bibfnamefont {Dominic~V.}\
  \bibnamefont {Else}}\ and\ \bibinfo {author} {\bibfnamefont {Chetan}\
  \bibnamefont {Nayak}},\ }\bibfield  {title} {\enquote {\bibinfo {title}
  {{Classifying symmetry-protected topological phases through the anomalous
  action of the symmetry on the edge}},}\ }\href {\doibase
  10.1103/PhysRevB.90.235137} {\bibfield  {journal} {\bibinfo  {journal} {Phys.
  Rev. B}\ }\textbf {\bibinfo {volume} {90}},\ \bibinfo {pages} {235137}
  (\bibinfo {year} {2014})},\ \Eprint {http://arxiv.org/abs/1409.5436}
  {arXiv:1409.5436} \BibitemShut {NoStop}%
\bibitem [{\citenamefont {'t~Hooft}(1980)}]{tHooftAnomaly}%
  \BibitemOpen
  \bibfield  {author} {\bibinfo {author} {\bibfnamefont {G.}~\bibnamefont
  {'t~Hooft}},\ }\bibfield  {title} {\enquote {\bibinfo {title} {Naturalness,
  chiral symmetry, and spontaneous chiral symmetry breaking},}\ }in\ \href
  {\doibase 10.1007\%2F978-1-4684-7571-5_9} {\emph {\bibinfo {booktitle} {NATO
  Advanced Study Institute Seminar Series}}},\ Vol.~\bibinfo {volume} {59}\
  (\bibinfo {year} {1980})\ p.\ \bibinfo {pages} {135}\BibitemShut {NoStop}%
\bibitem [{\citenamefont {Kapustin}\ and\ \citenamefont
  {Thorngren}(2014)}]{Kapustin_1403_0617}%
  \BibitemOpen
  \bibfield  {author} {\bibinfo {author} {\bibfnamefont {Anton}\ \bibnamefont
  {Kapustin}}\ and\ \bibinfo {author} {\bibfnamefont {Ryan}\ \bibnamefont
  {Thorngren}},\ }\bibfield  {title} {\enquote {\bibinfo {title} {Anomalous
  discrete symmetries in three dimensions and group cohomology},}\ }\href
  {\doibase 10.1103/PhysRevLett.112.231602} {\bibfield  {journal} {\bibinfo
  {journal} {Phys. Rev. Lett.}\ }\textbf {\bibinfo {volume} {112}},\ \bibinfo
  {pages} {231602} (\bibinfo {year} {2014})},\ \Eprint
  {http://arxiv.org/abs/1403.0617} {arXiv:1403.0617} \BibitemShut {NoStop}%
\bibitem [{\citenamefont {Gu}\ and\ \citenamefont {Wen}(2009)}]{Gu_0903}%
  \BibitemOpen
  \bibfield  {author} {\bibinfo {author} {\bibfnamefont {Zheng-Cheng}\
  \bibnamefont {Gu}}\ and\ \bibinfo {author} {\bibfnamefont {Xiao-Gang}\
  \bibnamefont {Wen}},\ }\bibfield  {title} {\enquote {\bibinfo {title}
  {{Tensor-entanglement-filtering renormalization approach and
  symmetry-protected topological order}},}\ }\href {\doibase
  10.1103/PhysRevB.80.155131} {\bibfield  {journal} {\bibinfo  {journal} {Phys.
  Rev. B}\ }\textbf {\bibinfo {volume} {80}},\ \bibinfo {pages} {155131}
  (\bibinfo {year} {2009})},\ \Eprint {http://arxiv.org/abs/0903.1069}
  {arXiv:0903.1069} \BibitemShut {NoStop}%
\bibitem [{\citenamefont {Senthil}(2015)}]{Senthil_1405}%
  \BibitemOpen
  \bibfield  {author} {\bibinfo {author} {\bibfnamefont {T.}~\bibnamefont
  {Senthil}},\ }\bibfield  {title} {\enquote {\bibinfo {title}
  {Symmetry-protected topological phases of quantum matter},}\ }\href {\doibase
  10.1146/annurev-conmatphys-031214-014740} {\bibfield  {journal} {\bibinfo
  {journal} {Annu. Rev. Condens. Matter Phys.}\ }\textbf {\bibinfo {volume}
  {6}},\ \bibinfo {pages} {299} (\bibinfo {year} {2015})},\ \Eprint
  {http://arxiv.org/abs/1405.4015} {arXiv:1405.4015} \BibitemShut {NoStop}%
\bibitem [{\citenamefont {Wen}(2017)}]{Wen_1610}%
  \BibitemOpen
  \bibfield  {author} {\bibinfo {author} {\bibfnamefont {Xiao-Gang}\
  \bibnamefont {Wen}},\ }\bibfield  {title} {\enquote {\bibinfo {title}
  {{Colloquium: Zoo of quantum-topological phases of matter}},}\ }\href
  {\doibase 10.1103/RevModPhys.89.041004} {\bibfield  {journal} {\bibinfo
  {journal} {Rev. Mod. Phys.}\ }\textbf {\bibinfo {volume} {89}},\ \bibinfo
  {pages} {041004} (\bibinfo {year} {2017})},\ \Eprint
  {http://arxiv.org/abs/1610.03911} {arXiv:1610.03911} \BibitemShut {NoStop}%
\bibitem [{\citenamefont {Hasan}\ and\ \citenamefont
  {Kane}(2010)}]{Hasan_1002}%
  \BibitemOpen
  \bibfield  {author} {\bibinfo {author} {\bibfnamefont {M.~Z.}\ \bibnamefont
  {Hasan}}\ and\ \bibinfo {author} {\bibfnamefont {C.~L.}\ \bibnamefont
  {Kane}},\ }\bibfield  {title} {\enquote {\bibinfo {title} {Colloquium:
  Topological insulators},}\ }\href {\doibase 10.1103/RevModPhys.82.3045}
  {\bibfield  {journal} {\bibinfo  {journal} {Rev. Mod. Phys.}\ }\textbf
  {\bibinfo {volume} {82}},\ \bibinfo {pages} {3045} (\bibinfo {year}
  {2010})},\ \Eprint {http://arxiv.org/abs/1002.3895} {arXiv:1002.3895}
  \BibitemShut {NoStop}%
\bibitem [{\citenamefont {Haldane}(1983{\natexlab{a}})}]{Haldane_1983a}%
  \BibitemOpen
  \bibfield  {author} {\bibinfo {author} {\bibfnamefont {F.D.M.}\ \bibnamefont
  {Haldane}},\ }\bibfield  {title} {\enquote {\bibinfo {title} {{Continuum
  dynamics of the 1-D Heisenberg antiferromagnet: Identification with the O(3)
  nonlinear sigma model}},}\ }\href {\doibase 10.1016/0375-9601(83)90631-X}
  {\bibfield  {journal} {\bibinfo  {journal} {Phys. Lett. A}\ }\textbf
  {\bibinfo {volume} {93}},\ \bibinfo {pages} {464} (\bibinfo {year}
  {1983}{\natexlab{a}})}\BibitemShut {NoStop}%
\bibitem [{\citenamefont {Haldane}(1983{\natexlab{b}})}]{Haldane_1983b}%
  \BibitemOpen
  \bibfield  {author} {\bibinfo {author} {\bibfnamefont {F.~D.~M.}\
  \bibnamefont {Haldane}},\ }\bibfield  {title} {\enquote {\bibinfo {title}
  {Nonlinear field theory of large-spin heisenberg antiferromagnets:
  Semiclassically quantized solitons of the one-dimensional easy-axis {Néel}
  state},}\ }\href {\doibase 10.1103/PhysRevLett.50.1153} {\bibfield  {journal}
  {\bibinfo  {journal} {Phys. Rev. Lett.}\ }\textbf {\bibinfo {volume} {50}},\
  \bibinfo {pages} {1153} (\bibinfo {year} {1983}{\natexlab{b}})}\BibitemShut
  {NoStop}%
\bibitem [{\citenamefont {Pollmann}\ \emph {et~al.}(2010)\citenamefont
  {Pollmann}, \citenamefont {Turner}, \citenamefont {Berg},\ and\ \citenamefont
  {Oshikawa}}]{Pollmann_0910}%
  \BibitemOpen
  \bibfield  {author} {\bibinfo {author} {\bibfnamefont {Frank}\ \bibnamefont
  {Pollmann}}, \bibinfo {author} {\bibfnamefont {Ari~M.}\ \bibnamefont
  {Turner}}, \bibinfo {author} {\bibfnamefont {Erez}\ \bibnamefont {Berg}}, \
  and\ \bibinfo {author} {\bibfnamefont {Masaki}\ \bibnamefont {Oshikawa}},\
  }\bibfield  {title} {\enquote {\bibinfo {title} {{Entanglement spectrum of a
  topological phase in one dimension}},}\ }\href {\doibase
  10.1103/PhysRevB.81.064439} {\bibfield  {journal} {\bibinfo  {journal} {Phys.
  Rev. B}\ }\textbf {\bibinfo {volume} {81}},\ \bibinfo {pages} {064439}
  (\bibinfo {year} {2010})},\ \Eprint {http://arxiv.org/abs/0910.1811}
  {arXiv:0910.1811} \BibitemShut {NoStop}%
\bibitem [{\citenamefont {Qi}\ \emph {et~al.}(2008)\citenamefont {Qi},
  \citenamefont {Hughes},\ and\ \citenamefont {Zhang}}]{Qi_0802}%
  \BibitemOpen
  \bibfield  {author} {\bibinfo {author} {\bibfnamefont {Xiao-Liang}\
  \bibnamefont {Qi}}, \bibinfo {author} {\bibfnamefont {Taylor~L.}\
  \bibnamefont {Hughes}}, \ and\ \bibinfo {author} {\bibfnamefont {Shou-Cheng}\
  \bibnamefont {Zhang}},\ }\bibfield  {title} {\enquote {\bibinfo {title}
  {{Topological field theory of time-reversal invariant insulators}},}\ }\href
  {\doibase 10.1103/PhysRevB.78.195424} {\bibfield  {journal} {\bibinfo
  {journal} {Phys. Rev. B}\ }\textbf {\bibinfo {volume} {78}},\ \bibinfo
  {pages} {195424} (\bibinfo {year} {2008})},\ \Eprint
  {http://arxiv.org/abs/0802.3537} {arXiv:0802.3537} \BibitemShut {NoStop}%
\bibitem [{\citenamefont {Pollmann}\ \emph {et~al.}(2012)\citenamefont
  {Pollmann}, \citenamefont {Berg}, \citenamefont {Turner},\ and\ \citenamefont
  {Oshikawa}}]{Pollmann_0909}%
  \BibitemOpen
  \bibfield  {author} {\bibinfo {author} {\bibfnamefont {Frank}\ \bibnamefont
  {Pollmann}}, \bibinfo {author} {\bibfnamefont {Erez}\ \bibnamefont {Berg}},
  \bibinfo {author} {\bibfnamefont {Ari~M.}\ \bibnamefont {Turner}}, \ and\
  \bibinfo {author} {\bibfnamefont {Masaki}\ \bibnamefont {Oshikawa}},\
  }\bibfield  {title} {\enquote {\bibinfo {title} {{Symmetry protection of
  topological phases in one-dimensional quantum spin systems}},}\ }\href
  {\doibase 10.1103/PhysRevB.85.075125} {\bibfield  {journal} {\bibinfo
  {journal} {Phys. Rev. B}\ }\textbf {\bibinfo {volume} {85}},\ \bibinfo
  {pages} {075125} (\bibinfo {year} {2012})},\ \Eprint
  {http://arxiv.org/abs/0909.4059} {arXiv:0909.4059} \BibitemShut {NoStop}%
\bibitem [{\citenamefont {Chen}\ \emph
  {et~al.}(2011{\natexlab{b}})\citenamefont {Chen}, \citenamefont {Gu},\ and\
  \citenamefont {Wen}}]{Chen_1008}%
  \BibitemOpen
  \bibfield  {author} {\bibinfo {author} {\bibfnamefont {Xie}\ \bibnamefont
  {Chen}}, \bibinfo {author} {\bibfnamefont {Zheng-Cheng}\ \bibnamefont {Gu}},
  \ and\ \bibinfo {author} {\bibfnamefont {Xiao-Gang}\ \bibnamefont {Wen}},\
  }\bibfield  {title} {\enquote {\bibinfo {title} {{Classification of gapped
  symmetric phases in one-dimensional spin systems}},}\ }\href {\doibase
  10.1103/PhysRevB.83.035107} {\bibfield  {journal} {\bibinfo  {journal} {Phys.
  Rev. B}\ }\textbf {\bibinfo {volume} {83}},\ \bibinfo {pages} {035107}
  (\bibinfo {year} {2011}{\natexlab{b}})},\ \Eprint
  {http://arxiv.org/abs/1008.3745} {arXiv:1008.3745} \BibitemShut {NoStop}%
\bibitem [{\citenamefont {Schuch}\ \emph {et~al.}(2011)\citenamefont {Schuch},
  \citenamefont {Pérez-García},\ and\ \citenamefont {Cirac}}]{Schuch_1010}%
  \BibitemOpen
  \bibfield  {author} {\bibinfo {author} {\bibfnamefont {Norbert}\ \bibnamefont
  {Schuch}}, \bibinfo {author} {\bibfnamefont {David}\ \bibnamefont
  {Pérez-García}}, \ and\ \bibinfo {author} {\bibfnamefont {Ignacio}\
  \bibnamefont {Cirac}},\ }\bibfield  {title} {\enquote {\bibinfo {title}
  {{Classifying quantum phases using matrix product states and projected
  entangled pair states}},}\ }\href {\doibase 10.1103/PhysRevB.84.165139}
  {\bibfield  {journal} {\bibinfo  {journal} {Phys. Rev. B}\ }\textbf {\bibinfo
  {volume} {84}},\ \bibinfo {pages} {165139} (\bibinfo {year} {2011})},\
  \Eprint {http://arxiv.org/abs/1010.3732} {arXiv:1010.3732} \BibitemShut
  {NoStop}%
\bibitem [{\citenamefont {Chen}\ \emph
  {et~al.}(2011{\natexlab{c}})\citenamefont {Chen}, \citenamefont {Gu},\ and\
  \citenamefont {Wen}}]{Chen_1103}%
  \BibitemOpen
  \bibfield  {author} {\bibinfo {author} {\bibfnamefont {Xie}\ \bibnamefont
  {Chen}}, \bibinfo {author} {\bibfnamefont {Zheng-Cheng}\ \bibnamefont {Gu}},
  \ and\ \bibinfo {author} {\bibfnamefont {Xiao-Gang}\ \bibnamefont {Wen}},\
  }\bibfield  {title} {\enquote {\bibinfo {title} {{Complete classification of
  one-dimensional gapped quantum phases in interacting spin systems}},}\ }\href
  {\doibase 10.1103/PhysRevB.84.235128} {\bibfield  {journal} {\bibinfo
  {journal} {Phys. Rev. B}\ }\textbf {\bibinfo {volume} {84}},\ \bibinfo
  {pages} {235128} (\bibinfo {year} {2011}{\natexlab{c}})},\ \Eprint
  {http://arxiv.org/abs/1103.3323} {arXiv:1103.3323} \BibitemShut {NoStop}%
\bibitem [{\citenamefont {Chen}\ \emph {et~al.}(2013)\citenamefont {Chen},
  \citenamefont {Gu}, \citenamefont {Liu},\ and\ \citenamefont
  {Wen}}]{Chen_1106_4772}%
  \BibitemOpen
  \bibfield  {author} {\bibinfo {author} {\bibfnamefont {Xie}\ \bibnamefont
  {Chen}}, \bibinfo {author} {\bibfnamefont {Zheng-Cheng}\ \bibnamefont {Gu}},
  \bibinfo {author} {\bibfnamefont {Zheng-Xin}\ \bibnamefont {Liu}}, \ and\
  \bibinfo {author} {\bibfnamefont {Xiao-Gang}\ \bibnamefont {Wen}},\
  }\bibfield  {title} {\enquote {\bibinfo {title} {{Symmetry protected
  topological orders and the group cohomology of their symmetry group}},}\
  }\href {\doibase 10.1103/PhysRevB.87.155114} {\bibfield  {journal} {\bibinfo
  {journal} {Phys. Rev. B}\ }\textbf {\bibinfo {volume} {87}},\ \bibinfo
  {pages} {155114} (\bibinfo {year} {2013})},\ \Eprint
  {http://arxiv.org/abs/1106.4772} {arXiv:1106.4772} \BibitemShut {NoStop}%
\bibitem [{\citenamefont {Gu}\ and\ \citenamefont {Wen}(2014)}]{Gu_1201}%
  \BibitemOpen
  \bibfield  {author} {\bibinfo {author} {\bibfnamefont {Zheng-Cheng}\
  \bibnamefont {Gu}}\ and\ \bibinfo {author} {\bibfnamefont {Xiao-Gang}\
  \bibnamefont {Wen}},\ }\bibfield  {title} {\enquote {\bibinfo {title}
  {{Symmetry-protected topological orders for interacting fermions: Fermionic
  topological nonlinear $\ensuremath{\sigma}$ models and a special group
  supercohomology theory}},}\ }\href {\doibase 10.1103/PhysRevB.90.115141}
  {\bibfield  {journal} {\bibinfo  {journal} {Phys. Rev. B}\ }\textbf {\bibinfo
  {volume} {90}},\ \bibinfo {pages} {115141} (\bibinfo {year} {2014})},\
  \Eprint {http://arxiv.org/abs/1201.2648} {arXiv:1201.2648} \BibitemShut
  {NoStop}%
\bibitem [{\citenamefont {Levin}\ and\ \citenamefont
  {Gu}(2012)}]{levin2012braiding}%
  \BibitemOpen
  \bibfield  {author} {\bibinfo {author} {\bibfnamefont {Michael}\ \bibnamefont
  {Levin}}\ and\ \bibinfo {author} {\bibfnamefont {Zheng-Cheng}\ \bibnamefont
  {Gu}},\ }\bibfield  {title} {\enquote {\bibinfo {title} {{Braiding statistics
  approach to symmetry-protected topological phases}},}\ }\href {\doibase
  10.1103/PhysRevB.86.115109} {\bibfield  {journal} {\bibinfo  {journal} {Phys.
  Rev. B}\ }\textbf {\bibinfo {volume} {86}},\ \bibinfo {pages} {115109}
  (\bibinfo {year} {2012})},\ \Eprint {http://arxiv.org/abs/1202.3120}
  {arXiv:1202.3120} \BibitemShut {NoStop}%
\bibitem [{\citenamefont {Lu}\ and\ \citenamefont
  {Vishwanath}(2012)}]{Lu_1205}%
  \BibitemOpen
  \bibfield  {author} {\bibinfo {author} {\bibfnamefont {Yuan-Ming}\
  \bibnamefont {Lu}}\ and\ \bibinfo {author} {\bibfnamefont {Ashvin}\
  \bibnamefont {Vishwanath}},\ }\bibfield  {title} {\enquote {\bibinfo {title}
  {{Theory and classification of interacting integer topological phases in two
  dimensions: A Chern-Simons approach}},}\ }\href {\doibase
  10.1103/PhysRevB.86.125119} {\bibfield  {journal} {\bibinfo  {journal} {Phys.
  Rev. B}\ }\textbf {\bibinfo {volume} {86}},\ \bibinfo {pages} {125119}
  (\bibinfo {year} {2012})},\ \Eprint {http://arxiv.org/abs/1205.3156}
  {arXiv:1205.3156} \BibitemShut {NoStop}%
\bibitem [{\citenamefont {Vishwanath}\ and\ \citenamefont
  {Senthil}(2013)}]{Vishwanath_1209}%
  \BibitemOpen
  \bibfield  {author} {\bibinfo {author} {\bibfnamefont {Ashvin}\ \bibnamefont
  {Vishwanath}}\ and\ \bibinfo {author} {\bibfnamefont {T.}~\bibnamefont
  {Senthil}},\ }\bibfield  {title} {\enquote {\bibinfo {title} {Physics of
  three-dimensional bosonic topological insulators: Surface-deconfined
  criticality and quantized magnetoelectric effect},}\ }\href {\doibase
  10.1103/PhysRevX.3.011016} {\bibfield  {journal} {\bibinfo  {journal} {Phys.
  Rev. X}\ }\textbf {\bibinfo {volume} {3}},\ \bibinfo {pages} {011016}
  (\bibinfo {year} {2013})},\ \Eprint {http://arxiv.org/abs/1209.3058}
  {arXiv:1209.3058} \BibitemShut {NoStop}%
\bibitem [{\citenamefont {Fidkowski}\ \emph {et~al.}(2013)\citenamefont
  {Fidkowski}, \citenamefont {Chen},\ and\ \citenamefont
  {Vishwanath}}]{fidkowski2013non}%
  \BibitemOpen
  \bibfield  {author} {\bibinfo {author} {\bibfnamefont {Lukasz}\ \bibnamefont
  {Fidkowski}}, \bibinfo {author} {\bibfnamefont {Xie}\ \bibnamefont {Chen}}, \
  and\ \bibinfo {author} {\bibfnamefont {Ashvin}\ \bibnamefont {Vishwanath}},\
  }\bibfield  {title} {\enquote {\bibinfo {title} {Non-abelian topological
  order on the surface of a 3d topological superconductor from an exactly
  solved model},}\ }\href {\doibase 10.1103/PhysRevX.3.041016} {\bibfield
  {journal} {\bibinfo  {journal} {Phys. Rev. X}\ }\textbf {\bibinfo {volume}
  {3}},\ \bibinfo {pages} {041016} (\bibinfo {year} {2013})},\ \Eprint
  {http://arxiv.org/abs/1305.5851} {arXiv:1305.5851} \BibitemShut {NoStop}%
\bibitem [{\citenamefont {Wang}\ \emph {et~al.}(2014)\citenamefont {Wang},
  \citenamefont {Potter},\ and\ \citenamefont
  {Senthil}}]{wang2014classification}%
  \BibitemOpen
  \bibfield  {author} {\bibinfo {author} {\bibfnamefont {C.}~\bibnamefont
  {Wang}}, \bibinfo {author} {\bibfnamefont {A.~C.}\ \bibnamefont {Potter}}, \
  and\ \bibinfo {author} {\bibfnamefont {T.}~\bibnamefont {Senthil}},\
  }\bibfield  {title} {\enquote {\bibinfo {title} {Classification of
  interacting electronic topological insulators in three dimensions},}\ }\href
  {\doibase 10.1126/science.1243326} {\bibfield  {journal} {\bibinfo  {journal}
  {Science}\ }\textbf {\bibinfo {volume} {343}},\ \bibinfo {pages} {629}
  (\bibinfo {year} {2014})},\ \Eprint {http://arxiv.org/abs/1306.3238}
  {arXiv:1306.3238} \BibitemShut {NoStop}%
\bibitem [{\citenamefont {Wang}\ and\ \citenamefont
  {Senthil}(2014)}]{wang2014interacting}%
  \BibitemOpen
  \bibfield  {author} {\bibinfo {author} {\bibfnamefont {Chong}\ \bibnamefont
  {Wang}}\ and\ \bibinfo {author} {\bibfnamefont {T.}~\bibnamefont {Senthil}},\
  }\bibfield  {title} {\enquote {\bibinfo {title} {{Interacting fermionic
  topological insulators/superconductors in three dimensions}},}\ }\href
  {\doibase 10.1103/PhysRevB.89.195124} {\bibfield  {journal} {\bibinfo
  {journal} {Phys. Rev. B}\ }\textbf {\bibinfo {volume} {89}},\ \bibinfo
  {pages} {195124} (\bibinfo {year} {2014})},\ \Eprint
  {http://arxiv.org/abs/1401.1142} {arXiv:1401.1142} \BibitemShut {NoStop}%
\bibitem [{\citenamefont {Kapustin}()}]{Kapustin_1403_1467}%
  \BibitemOpen
  \bibfield  {author} {\bibinfo {author} {\bibfnamefont {Anton}\ \bibnamefont
  {Kapustin}},\ }\bibfield  {title} {\enquote {\bibinfo {title} {Symmetry
  protected topological phases, anomalies, and cobordisms: Beyond group
  cohomology},}\ }\href@noop {} {\ }\Eprint {http://arxiv.org/abs/1403.1467}
  {arXiv:1403.1467} \BibitemShut {NoStop}%
\bibitem [{\citenamefont {Metlitski}\ \emph {et~al.}()\citenamefont
  {Metlitski}, \citenamefont {Fidkowski}, \citenamefont {Chen},\ and\
  \citenamefont {Vishwanath}}]{metlitski2014interaction}%
  \BibitemOpen
  \bibfield  {author} {\bibinfo {author} {\bibfnamefont {Max~A.}\ \bibnamefont
  {Metlitski}}, \bibinfo {author} {\bibfnamefont {Lukasz}\ \bibnamefont
  {Fidkowski}}, \bibinfo {author} {\bibfnamefont {Xie}\ \bibnamefont {Chen}}, \
  and\ \bibinfo {author} {\bibfnamefont {Ashvin}\ \bibnamefont {Vishwanath}},\
  }\bibfield  {title} {\enquote {\bibinfo {title} {{Interaction effects on 3D
  topological superconductors: surface topological order from vortex
  condensation, the 16 fold way and fermionic Kramers doublets}},}\ }\href@noop
  {} {\ }\Eprint {http://arxiv.org/abs/1406.3032} {arXiv:1406.3032}
  \BibitemShut {NoStop}%
\bibitem [{\citenamefont {Kapustin}\ \emph {et~al.}(2015)\citenamefont
  {Kapustin}, \citenamefont {Thorngren}, \citenamefont {Turzillo},\ and\
  \citenamefont {Wang}}]{Kapustin_1406}%
  \BibitemOpen
  \bibfield  {author} {\bibinfo {author} {\bibfnamefont {Anton}\ \bibnamefont
  {Kapustin}}, \bibinfo {author} {\bibfnamefont {Ryan}\ \bibnamefont
  {Thorngren}}, \bibinfo {author} {\bibfnamefont {Alex}\ \bibnamefont
  {Turzillo}}, \ and\ \bibinfo {author} {\bibfnamefont {Zitao}\ \bibnamefont
  {Wang}},\ }\bibfield  {title} {\enquote {\bibinfo {title} {{Fermionic
  symmetry protected topological phases and cobordisms}},}\ }\href {\doibase
  10.1007/JHEP12(2015)052} {\bibfield  {journal} {\bibinfo  {journal} {J. High
  Energy Phys.}\ }\textbf {\bibinfo {volume} {2015}},\ \bibinfo {pages} {1}
  (\bibinfo {year} {2015})},\ \Eprint {http://arxiv.org/abs/1406.7329}
  {arXiv:1406.7329} \BibitemShut {NoStop}%
\bibitem [{\citenamefont {Cheng}\ \emph {et~al.}(2018)\citenamefont {Cheng},
  \citenamefont {Bi}, \citenamefont {You},\ and\ \citenamefont
  {Gu}}]{Cheng_1501}%
  \BibitemOpen
  \bibfield  {author} {\bibinfo {author} {\bibfnamefont {Meng}\ \bibnamefont
  {Cheng}}, \bibinfo {author} {\bibfnamefont {Zhen}\ \bibnamefont {Bi}},
  \bibinfo {author} {\bibfnamefont {Yi-Zhuang}\ \bibnamefont {You}}, \ and\
  \bibinfo {author} {\bibfnamefont {Zheng-Cheng}\ \bibnamefont {Gu}},\
  }\bibfield  {title} {\enquote {\bibinfo {title} {{Classification of
  symmetry-protected phases for interacting fermions in two dimensions}},}\
  }\href {\doibase 10.1103/PhysRevB.97.205109} {\bibfield  {journal} {\bibinfo
  {journal} {Phys. Rev. B}\ }\textbf {\bibinfo {volume} {97}},\ \bibinfo
  {pages} {205109} (\bibinfo {year} {2018})},\ \Eprint
  {http://arxiv.org/abs/1501.01313} {arXiv:1501.01313} \BibitemShut {NoStop}%
\bibitem [{\citenamefont {Kitaev}(2015)}]{KitaevIPAM}%
  \BibitemOpen
  \bibfield  {author} {\bibinfo {author} {\bibfnamefont {Alexei}\ \bibnamefont
  {Kitaev}},\ }\href@noop {} {} (\bibinfo {year} {2015}),\ \bibinfo {note}
  {{I}nstitute for Pure Applied Mathematics, UCLA.
  \url{http://www.ipam.ucla.edu/abstract/?tid=12389\&pcode=STQ2015
  }}\BibitemShut {NoStop}%
\bibitem [{\citenamefont {Witten}(2016)}]{witten2016fermion}%
  \BibitemOpen
  \bibfield  {author} {\bibinfo {author} {\bibfnamefont {Edward}\ \bibnamefont
  {Witten}},\ }\bibfield  {title} {\enquote {\bibinfo {title} {{Fermion path
  integrals and topological phases}},}\ }\href {\doibase
  10.1103/RevModPhys.88.035001} {\bibfield  {journal} {\bibinfo  {journal}
  {Rev. Mod. Phys.}\ }\textbf {\bibinfo {volume} {88}},\ \bibinfo {pages}
  {035001} (\bibinfo {year} {2016})},\ \Eprint
  {http://arxiv.org/abs/1508.04715} {arXiv:1508.04715} \BibitemShut {NoStop}%
\bibitem [{\citenamefont {Freed}\ and\ \citenamefont {Hopkins}()}]{Freed_1604}%
  \BibitemOpen
  \bibfield  {author} {\bibinfo {author} {\bibfnamefont {Daniel~S.}\
  \bibnamefont {Freed}}\ and\ \bibinfo {author} {\bibfnamefont {Michael~J.}\
  \bibnamefont {Hopkins}},\ }\bibfield  {title} {\enquote {\bibinfo {title}
  {{Reflection positivity and invertible topological phases}},}\ }\href@noop {}
  {\ }\Eprint {http://arxiv.org/abs/1604.06527} {arXiv:1604.06527} \BibitemShut
  {NoStop}%
\bibitem [{\citenamefont {Xiong}(2018)}]{Xiong_1701}%
  \BibitemOpen
  \bibfield  {author} {\bibinfo {author} {\bibfnamefont {Charles~Zhaoxi}\
  \bibnamefont {Xiong}},\ }\bibfield  {title} {\enquote {\bibinfo {title}
  {{Minimalist approach to the classification of symmetry protected topological
  phases}},}\ }\href {\doibase 10.1088/1751-8121/aae0b1} {\bibfield  {journal}
  {\bibinfo  {journal} {J. Phys. A: Math. Theor.}\ }\textbf {\bibinfo {volume}
  {51}},\ \bibinfo {pages} {445001} (\bibinfo {year} {2018})},\ \Eprint
  {http://arxiv.org/abs/1701.00004} {arXiv:1701.00004} \BibitemShut {NoStop}%
\bibitem [{\citenamefont {Kapustin}\ and\ \citenamefont
  {Thorngren}(2017)}]{Kapustin_1701}%
  \BibitemOpen
  \bibfield  {author} {\bibinfo {author} {\bibfnamefont {Anton}\ \bibnamefont
  {Kapustin}}\ and\ \bibinfo {author} {\bibfnamefont {Ryan}\ \bibnamefont
  {Thorngren}},\ }\bibfield  {title} {\enquote {\bibinfo {title} {{Fermionic
  SPT phases in higher dimensions and bosonization}},}\ }\href {\doibase
  10.1007/JHEP10(2017)080} {\bibfield  {journal} {\bibinfo  {journal} {J. High
  Energy Phys.}\ }\textbf {\bibinfo {volume} {2017}},\ \bibinfo {pages} {80}
  (\bibinfo {year} {2017})},\ \Eprint {http://arxiv.org/abs/1701.08264}
  {arXiv:1701.08264} \BibitemShut {NoStop}%
\bibitem [{\citenamefont {Wang}\ and\ \citenamefont {Gu}(2018)}]{Wang_1703}%
  \BibitemOpen
  \bibfield  {author} {\bibinfo {author} {\bibfnamefont {Qing-Rui}\
  \bibnamefont {Wang}}\ and\ \bibinfo {author} {\bibfnamefont {Zheng-Cheng}\
  \bibnamefont {Gu}},\ }\bibfield  {title} {\enquote {\bibinfo {title} {Towards
  a complete classification of symmetry-protected topological phases for
  interacting fermions in three dimensions and a general group supercohomology
  theory},}\ }\href {\doibase 10.1103/PhysRevX.8.011055} {\bibfield  {journal}
  {\bibinfo  {journal} {Phys. Rev. X}\ }\textbf {\bibinfo {volume} {8}},\
  \bibinfo {pages} {011055} (\bibinfo {year} {2018})},\ \Eprint
  {http://arxiv.org/abs/1703.10937} {arXiv:1703.10937} \BibitemShut {NoStop}%
\bibitem [{\citenamefont {Gaiotto}\ and\ \citenamefont
  {Johnson-Freyd}(2019)}]{Gaiotto_1712}%
  \BibitemOpen
  \bibfield  {author} {\bibinfo {author} {\bibfnamefont {Davide}\ \bibnamefont
  {Gaiotto}}\ and\ \bibinfo {author} {\bibfnamefont {Theo}\ \bibnamefont
  {Johnson-Freyd}},\ }\bibfield  {title} {\enquote {\bibinfo {title} {{Symmetry
  protected topological phases and generalized cohomology}},}\ }\href {\doibase
  10.1007/JHEP05(2019)007} {\bibfield  {journal} {\bibinfo  {journal} {J. High
  Energy Phys.}\ }\textbf {\bibinfo {volume} {2019}},\ \bibinfo {pages} {7}
  (\bibinfo {year} {2019})},\ \Eprint {http://arxiv.org/abs/1712.07950}
  {arXiv:1712.07950} \BibitemShut {NoStop}%
\bibitem [{\citenamefont {Wang}\ and\ \citenamefont {Gu}(2020)}]{Wang_1811}%
  \BibitemOpen
  \bibfield  {author} {\bibinfo {author} {\bibfnamefont {Qing-Rui}\
  \bibnamefont {Wang}}\ and\ \bibinfo {author} {\bibfnamefont {Zheng-Cheng}\
  \bibnamefont {Gu}},\ }\bibfield  {title} {\enquote {\bibinfo {title}
  {Construction and classification of symmetry-protected topological phases in
  interacting fermion systems},}\ }\href {\doibase 10.1103/PhysRevX.10.031055}
  {\bibfield  {journal} {\bibinfo  {journal} {Phys. Rev. X}\ }\textbf {\bibinfo
  {volume} {10}},\ \bibinfo {pages} {031055} (\bibinfo {year} {2020})},\
  \Eprint {http://arxiv.org/abs/1811.00536} {arXiv:1811.00536} \BibitemShut
  {NoStop}%
\bibitem [{\citenamefont {Laughlin}(1981)}]{Laughlin_1981}%
  \BibitemOpen
  \bibfield  {author} {\bibinfo {author} {\bibfnamefont {R.~B.}\ \bibnamefont
  {Laughlin}},\ }\bibfield  {title} {\enquote {\bibinfo {title} {{Quantized
  Hall conductivity in two dimensions}},}\ }\href {\doibase
  10.1103/PhysRevB.23.5632} {\bibfield  {journal} {\bibinfo  {journal} {Phys.
  Rev. B}\ }\textbf {\bibinfo {volume} {23}},\ \bibinfo {pages} {5632}
  (\bibinfo {year} {1981})}\BibitemShut {NoStop}%
\bibitem [{\citenamefont {Haldane}(1992)}]{Haldane_0505}%
  \BibitemOpen
  \bibfield  {author} {\bibinfo {author} {\bibfnamefont {F.~D.~M.}\
  \bibnamefont {Haldane}},\ }\bibfield  {title} {\enquote {\bibinfo {title}
  {Luttinger's theorem and bosonization of the {Fermi} surface},}\ }in\
  \href@noop {} {\emph {\bibinfo {booktitle} {Proceedings of the International
  School of Physics ``Enrico Fermi''}}}\ (\bibinfo {year} {1992})\ p.~\bibinfo
  {pages} {5},\ \Eprint {http://arxiv.org/abs/cond-mat/0505529}
  {arXiv:cond-mat/0505529} \BibitemShut {NoStop}%
\bibitem [{\citenamefont {Pressley}\ and\ \citenamefont
  {Segal}()}]{LoopGroupsBook}%
  \BibitemOpen
  \bibfield  {author} {\bibinfo {author} {\bibfnamefont {Andrew}\ \bibnamefont
  {Pressley}}\ and\ \bibinfo {author} {\bibfnamefont {Graeme}\ \bibnamefont
  {Segal}},\ }\href@noop {} {\emph {\bibinfo {title} {Loop groups}}}\ (\bibinfo
   {publisher} {Oxford University Press},\ \bibinfo {address}
  {Oxford})\BibitemShut {NoStop}%
\bibitem [{\citenamefont {Chang}\ and\ \citenamefont {Niu}(1996)}]{Chang_9511}%
  \BibitemOpen
  \bibfield  {author} {\bibinfo {author} {\bibfnamefont {Ming-Che}\
  \bibnamefont {Chang}}\ and\ \bibinfo {author} {\bibfnamefont {Qian}\
  \bibnamefont {Niu}},\ }\bibfield  {title} {\enquote {\bibinfo {title} {{Berry
  phase, hyperorbits, and the Hofstadter spectrum: Semiclassical dynamics in
  magnetic Bloch bands}},}\ }\href {\doibase 10.1103/PhysRevB.53.7010}
  {\bibfield  {journal} {\bibinfo  {journal} {Phys. Rev. B}\ }\textbf {\bibinfo
  {volume} {53}},\ \bibinfo {pages} {7010} (\bibinfo {year} {1996})},\ \Eprint
  {http://arxiv.org/abs/cond-mat/9511014} {arXiv:cond-mat/9511014} \BibitemShut
  {NoStop}%
\bibitem [{\citenamefont {Barci}\ \emph {et~al.}(2018)\citenamefont {Barci},
  \citenamefont {Fradkin},\ and\ \citenamefont {Ribeiro}}]{Barci_1805}%
  \BibitemOpen
  \bibfield  {author} {\bibinfo {author} {\bibfnamefont {Daniel~G.}\
  \bibnamefont {Barci}}, \bibinfo {author} {\bibfnamefont {Eduardo}\
  \bibnamefont {Fradkin}}, \ and\ \bibinfo {author} {\bibfnamefont {Leonardo}\
  \bibnamefont {Ribeiro}},\ }\bibfield  {title} {\enquote {\bibinfo {title}
  {{Bosonization of Fermi liquids in a weak magnetic field}},}\ }\href
  {\doibase 10.1103/PhysRevB.98.155146} {\bibfield  {journal} {\bibinfo
  {journal} {Phys. Rev. B}\ }\textbf {\bibinfo {volume} {98}},\ \bibinfo
  {pages} {155146} (\bibinfo {year} {2018})},\ \Eprint
  {http://arxiv.org/abs/1805.05337} {arXiv:1805.05337} \BibitemShut {NoStop}%
\bibitem [{\citenamefont {Golkar}\ \emph {et~al.}(2016)\citenamefont {Golkar},
  \citenamefont {Nguyen}, \citenamefont {Roberts},\ and\ \citenamefont
  {Son}}]{Golkar_1602}%
  \BibitemOpen
  \bibfield  {author} {\bibinfo {author} {\bibfnamefont {Siavash}\ \bibnamefont
  {Golkar}}, \bibinfo {author} {\bibfnamefont {Dung~Xuan}\ \bibnamefont
  {Nguyen}}, \bibinfo {author} {\bibfnamefont {Matthew~M.}\ \bibnamefont
  {Roberts}}, \ and\ \bibinfo {author} {\bibfnamefont {Dam~Thanh}\ \bibnamefont
  {Son}},\ }\bibfield  {title} {\enquote {\bibinfo {title} {Higher-spin theory
  of the magnetorotons},}\ }\href {\doibase 10.1103/PhysRevLett.117.216403}
  {\bibfield  {journal} {\bibinfo  {journal} {Phys. Rev. Lett.}\ }\textbf
  {\bibinfo {volume} {117}},\ \bibinfo {pages} {216403} (\bibinfo {year}
  {2016})},\ \Eprint {http://arxiv.org/abs/1602.08499} {arXiv:1602.08499}
  \BibitemShut {NoStop}%
\bibitem [{\citenamefont {Holstein}\ \emph {et~al.}(1973)\citenamefont
  {Holstein}, \citenamefont {Norton},\ and\ \citenamefont
  {Pincus}}]{holstein1973haas}%
  \BibitemOpen
  \bibfield  {author} {\bibinfo {author} {\bibfnamefont {T.}~\bibnamefont
  {Holstein}}, \bibinfo {author} {\bibfnamefont {R.~E.}\ \bibnamefont
  {Norton}}, \ and\ \bibinfo {author} {\bibfnamefont {P.}~\bibnamefont
  {Pincus}},\ }\bibfield  {title} {\enquote {\bibinfo {title} {de {Haas}-van
  {Alphen} effect and the specific heat of an electron gas},}\ }\href {\doibase
  10.1103/physrevb.8.2649} {\bibfield  {journal} {\bibinfo  {journal} {Phys.
  Rev. B}\ }\textbf {\bibinfo {volume} {8}},\ \bibinfo {pages} {2649} (\bibinfo
  {year} {1973})}\BibitemShut {NoStop}%
\bibitem [{\citenamefont {Lee}\ \emph {et~al.}(2006)\citenamefont {Lee},
  \citenamefont {Nagaosa},\ and\ \citenamefont {Wen}}]{lee2006doping}%
  \BibitemOpen
  \bibfield  {author} {\bibinfo {author} {\bibfnamefont {Patrick~A.}\
  \bibnamefont {Lee}}, \bibinfo {author} {\bibfnamefont {Naoto}\ \bibnamefont
  {Nagaosa}}, \ and\ \bibinfo {author} {\bibfnamefont {Xiao-Gang}\ \bibnamefont
  {Wen}},\ }\bibfield  {title} {\enquote {\bibinfo {title} {{Doping a Mott
  insulator: Physics of high-temperature superconductivity}},}\ }\href
  {\doibase 10.1103/revmodphys.78.17} {\bibfield  {journal} {\bibinfo
  {journal} {Rev. Mod. Phys.}\ }\textbf {\bibinfo {volume} {78}},\ \bibinfo
  {pages} {17} (\bibinfo {year} {2006})},\ \Eprint
  {http://arxiv.org/abs/cond-mat/0410445} {arXiv:cond-mat/0410445} \BibitemShut
  {NoStop}%
\bibitem [{\citenamefont {Halperin}\ \emph {et~al.}(1993)\citenamefont
  {Halperin}, \citenamefont {Lee},\ and\ \citenamefont
  {Read}}]{halperin1993theory}%
  \BibitemOpen
  \bibfield  {author} {\bibinfo {author} {\bibfnamefont {B.~I.}\ \bibnamefont
  {Halperin}}, \bibinfo {author} {\bibfnamefont {Patrick~A.}\ \bibnamefont
  {Lee}}, \ and\ \bibinfo {author} {\bibfnamefont {Nicholas}\ \bibnamefont
  {Read}},\ }\bibfield  {title} {\enquote {\bibinfo {title} {{Theory of the
  half-filled Landau level}},}\ }\href {\doibase 10.1103/physrevb.47.7312}
  {\bibfield  {journal} {\bibinfo  {journal} {Phys. Rev. B}\ }\textbf {\bibinfo
  {volume} {47}},\ \bibinfo {pages} {7312} (\bibinfo {year}
  {1993})}\BibitemShut {NoStop}%
\bibitem [{\citenamefont {Polchinski}(1994)}]{polchinski1994low}%
  \BibitemOpen
  \bibfield  {author} {\bibinfo {author} {\bibfnamefont {Joseph}\ \bibnamefont
  {Polchinski}},\ }\bibfield  {title} {\enquote {\bibinfo {title} {{Low-energy
  dynamics of the spinon-gauge system}},}\ }\href {\doibase
  10.1016/0550-3213(94)90449-9} {\bibfield  {journal} {\bibinfo  {journal}
  {Nucl. Phys. B}\ }\textbf {\bibinfo {volume} {422}},\ \bibinfo {pages} {617}
  (\bibinfo {year} {1994})},\ \Eprint {http://arxiv.org/abs/cond-mat/9303037}
  {arXiv:cond-mat/9303037} \BibitemShut {NoStop}%
\bibitem [{\citenamefont {Nayak}\ and\ \citenamefont
  {Wilczek}(1994)}]{nayak1994non}%
  \BibitemOpen
  \bibfield  {author} {\bibinfo {author} {\bibfnamefont {Chetan}\ \bibnamefont
  {Nayak}}\ and\ \bibinfo {author} {\bibfnamefont {Frank}\ \bibnamefont
  {Wilczek}},\ }\bibfield  {title} {\enquote {\bibinfo {title} {{Non-Fermi
  liquid fixed point in 2 + 1 dimensions}},}\ }\href {\doibase
  10.1016/0550-3213(94)90477-4} {\bibfield  {journal} {\bibinfo  {journal}
  {Nucl. Phys. B}\ }\textbf {\bibinfo {volume} {417}},\ \bibinfo {pages} {359}
  (\bibinfo {year} {1994})},\ \Eprint {http://arxiv.org/abs/cond-mat/9312086}
  {arXiv:cond-mat/9312086} \BibitemShut {NoStop}%
\bibitem [{\citenamefont {Altshuler}\ \emph {et~al.}(1994)\citenamefont
  {Altshuler}, \citenamefont {Ioffe},\ and\ \citenamefont
  {Millis}}]{altshuler1994low}%
  \BibitemOpen
  \bibfield  {author} {\bibinfo {author} {\bibfnamefont {B.~L.}\ \bibnamefont
  {Altshuler}}, \bibinfo {author} {\bibfnamefont {L.~B.}\ \bibnamefont
  {Ioffe}}, \ and\ \bibinfo {author} {\bibfnamefont {A.~J.}\ \bibnamefont
  {Millis}},\ }\bibfield  {title} {\enquote {\bibinfo {title} {{Low-energy
  properties of fermions with singular interactions}},}\ }\href {\doibase
  10.1103/PhysRevB.50.14048} {\bibfield  {journal} {\bibinfo  {journal} {Phys.
  Rev. B}\ }\textbf {\bibinfo {volume} {50}},\ \bibinfo {pages} {14048}
  (\bibinfo {year} {1994})},\ \Eprint {http://arxiv.org/abs/cond-mat/9406024}
  {arXiv:cond-mat/9406024} \BibitemShut {NoStop}%
\bibitem [{\citenamefont {Lee}(2009)}]{lee2009low}%
  \BibitemOpen
  \bibfield  {author} {\bibinfo {author} {\bibfnamefont {Sung-Sik}\
  \bibnamefont {Lee}},\ }\bibfield  {title} {\enquote {\bibinfo {title}
  {{Low-energy effective theory of Fermi surface coupled with U(1) gauge field
  in $2+1$ dimensions}},}\ }\href {\doibase 10.1103/PhysRevB.80.165102}
  {\bibfield  {journal} {\bibinfo  {journal} {Phys. Rev. B}\ }\textbf {\bibinfo
  {volume} {80}},\ \bibinfo {pages} {165102} (\bibinfo {year} {2009})},\
  \Eprint {http://arxiv.org/abs/0905.4532} {arXiv:0905.4532} \BibitemShut
  {NoStop}%
\bibitem [{\citenamefont {Metlitski}\ and\ \citenamefont
  {Sachdev}(2010)}]{metlitski2010quantum}%
  \BibitemOpen
  \bibfield  {author} {\bibinfo {author} {\bibfnamefont {Max~A.}\ \bibnamefont
  {Metlitski}}\ and\ \bibinfo {author} {\bibfnamefont {Subir}\ \bibnamefont
  {Sachdev}},\ }\bibfield  {title} {\enquote {\bibinfo {title} {{Quantum phase
  transitions of metals in two spatial dimensions. I. Ising-nematic order}},}\
  }\href {\doibase 10.1103/PhysRevB.82.075127} {\bibfield  {journal} {\bibinfo
  {journal} {Phys. Rev. B}\ }\textbf {\bibinfo {volume} {82}},\ \bibinfo
  {pages} {075127} (\bibinfo {year} {2010})},\ \Eprint
  {http://arxiv.org/abs/1001.1153} {arXiv:1001.1153} \BibitemShut {NoStop}%
\bibitem [{\citenamefont {Mross}\ \emph {et~al.}(2010)\citenamefont {Mross},
  \citenamefont {McGreevy}, \citenamefont {Liu},\ and\ \citenamefont
  {Senthil}}]{mross2010controlled}%
  \BibitemOpen
  \bibfield  {author} {\bibinfo {author} {\bibfnamefont {David~F.}\
  \bibnamefont {Mross}}, \bibinfo {author} {\bibfnamefont {John}\ \bibnamefont
  {McGreevy}}, \bibinfo {author} {\bibfnamefont {Hong}\ \bibnamefont {Liu}}, \
  and\ \bibinfo {author} {\bibfnamefont {T.}~\bibnamefont {Senthil}},\
  }\bibfield  {title} {\enquote {\bibinfo {title} {{Controlled expansion for
  certain non-Fermi-liquid metals}},}\ }\href {\doibase
  10.1103/PhysRevB.82.045121} {\bibfield  {journal} {\bibinfo  {journal} {Phys.
  Rev. B}\ }\textbf {\bibinfo {volume} {82}},\ \bibinfo {pages} {045121}
  (\bibinfo {year} {2010})},\ \Eprint {http://arxiv.org/abs/1003.0894}
  {arXiv:1003.0894} \BibitemShut {NoStop}%
\bibitem [{\citenamefont {Schattner}\ \emph {et~al.}(2016)\citenamefont
  {Schattner}, \citenamefont {Lederer}, \citenamefont {Kivelson},\ and\
  \citenamefont {Berg}}]{schattner2016ising}%
  \BibitemOpen
  \bibfield  {author} {\bibinfo {author} {\bibfnamefont {Yoni}\ \bibnamefont
  {Schattner}}, \bibinfo {author} {\bibfnamefont {Samuel}\ \bibnamefont
  {Lederer}}, \bibinfo {author} {\bibfnamefont {Steven~A.}\ \bibnamefont
  {Kivelson}}, \ and\ \bibinfo {author} {\bibfnamefont {Erez}\ \bibnamefont
  {Berg}},\ }\bibfield  {title} {\enquote {\bibinfo {title} {Ising nematic
  quantum critical point in a metal: A {Monte} {Carlo} study},}\ }\href
  {\doibase 10.1103/PhysRevX.6.031028} {\bibfield  {journal} {\bibinfo
  {journal} {Phys. Rev. X}\ }\textbf {\bibinfo {volume} {6}},\ \bibinfo {pages}
  {031028} (\bibinfo {year} {2016})},\ \Eprint
  {http://arxiv.org/abs/1511.03282} {arXiv:1511.03282} \BibitemShut {NoStop}%
\bibitem [{\citenamefont {Dalidovich}\ and\ \citenamefont
  {Lee}(2013)}]{dalidovich2013perturbative}%
  \BibitemOpen
  \bibfield  {author} {\bibinfo {author} {\bibfnamefont {Denis}\ \bibnamefont
  {Dalidovich}}\ and\ \bibinfo {author} {\bibfnamefont {Sung-Sik}\ \bibnamefont
  {Lee}},\ }\bibfield  {title} {\enquote {\bibinfo {title} {{Perturbative
  non-Fermi liquids from dimensional regularization}},}\ }\href {\doibase
  10.1103/PhysRevB.88.245106} {\bibfield  {journal} {\bibinfo  {journal} {Phys.
  Rev. B}\ }\textbf {\bibinfo {volume} {88}},\ \bibinfo {pages} {245106}
  (\bibinfo {year} {2013})},\ \Eprint {http://arxiv.org/abs/1307.3170}
  {arXiv:1307.3170} \BibitemShut {NoStop}%
\bibitem [{\citenamefont {Lee}(2018)}]{lee2018recent}%
  \BibitemOpen
  \bibfield  {author} {\bibinfo {author} {\bibfnamefont {Sung-Sik}\
  \bibnamefont {Lee}},\ }\bibfield  {title} {\enquote {\bibinfo {title} {Recent
  developments in non-{Fermi} liquid theory},}\ }\href {\doibase
  10.1146/annurev-conmatphys-031016-025531} {\bibfield  {journal} {\bibinfo
  {journal} {Annu. Rev. Condens. Matter Phys.}\ }\textbf {\bibinfo {volume}
  {9}},\ \bibinfo {pages} {227} (\bibinfo {year} {2018})},\ \Eprint
  {http://arxiv.org/abs/1703.08172} {arXiv:1703.08172} \BibitemShut {NoStop}%
\bibitem [{\citenamefont {Metlitski}\ \emph {et~al.}(2015)\citenamefont
  {Metlitski}, \citenamefont {Mross}, \citenamefont {Sachdev},\ and\
  \citenamefont {Senthil}}]{metlitski2015cooper}%
  \BibitemOpen
  \bibfield  {author} {\bibinfo {author} {\bibfnamefont {Max~A.}\ \bibnamefont
  {Metlitski}}, \bibinfo {author} {\bibfnamefont {David~F.}\ \bibnamefont
  {Mross}}, \bibinfo {author} {\bibfnamefont {Subir}\ \bibnamefont {Sachdev}},
  \ and\ \bibinfo {author} {\bibfnamefont {T.}~\bibnamefont {Senthil}},\
  }\bibfield  {title} {\enquote {\bibinfo {title} {{Cooper pairing in non-Fermi
  liquids}},}\ }\href {\doibase 10.1103/PhysRevB.91.115111} {\bibfield
  {journal} {\bibinfo  {journal} {Phys. Rev. B}\ }\textbf {\bibinfo {volume}
  {91}},\ \bibinfo {pages} {115111} (\bibinfo {year} {2015})},\ \Eprint
  {http://arxiv.org/abs/1403.3694} {arXiv:1403.3694} \BibitemShut {NoStop}%
\bibitem [{\citenamefont {Senthil}(2008)}]{senthil2008critical}%
  \BibitemOpen
  \bibfield  {author} {\bibinfo {author} {\bibfnamefont {T.}~\bibnamefont
  {Senthil}},\ }\bibfield  {title} {\enquote {\bibinfo {title} {{Critical Fermi
  surfaces and non-Fermi liquid metals}},}\ }\href {\doibase
  10.1103/PhysRevB.78.035103} {\bibfield  {journal} {\bibinfo  {journal} {Phys.
  Rev. B}\ }\textbf {\bibinfo {volume} {78}},\ \bibinfo {pages} {035103}
  (\bibinfo {year} {2008})},\ \Eprint {http://arxiv.org/abs/0803.4009}
  {arXiv:0803.4009} \BibitemShut {NoStop}%
\bibitem [{\citenamefont {Shoenberg}(1984)}]{QuantumOscillationsBook}%
  \BibitemOpen
  \bibfield  {author} {\bibinfo {author} {\bibfnamefont {D.}~\bibnamefont
  {Shoenberg}},\ }\href {\doibase 10.1017/CBO9780511897870} {\emph {\bibinfo
  {title} {Magnetic Oscillations in Metals}}}\ (\bibinfo  {publisher}
  {Cambridge University Press},\ \bibinfo {address} {Cambridge},\ \bibinfo
  {year} {1984})\BibitemShut {NoStop}%
\bibitem [{\citenamefont {Chen}\ and\ \citenamefont {Son}(2017)}]{Chen_1604}%
  \BibitemOpen
  \bibfield  {author} {\bibinfo {author} {\bibfnamefont {Jing-Yuan}\
  \bibnamefont {Chen}}\ and\ \bibinfo {author} {\bibfnamefont {Dam~Thanh}\
  \bibnamefont {Son}},\ }\bibfield  {title} {\enquote {\bibinfo {title} {{Berry
  Fermi liquid theory}},}\ }\href {\doibase 10.1016/j.aop.2016.12.017}
  {\bibfield  {journal} {\bibinfo  {journal} {Ann. Phys.}\ }\textbf {\bibinfo
  {volume} {377}},\ \bibinfo {pages} {345} (\bibinfo {year} {2017})},\ \Eprint
  {http://arxiv.org/abs/1604.07857} {arXiv:1604.07857} \BibitemShut {NoStop}%
\bibitem [{\citenamefont {Son}\ and\ \citenamefont
  {Yamamoto}(2012)}]{Son_1203}%
  \BibitemOpen
  \bibfield  {author} {\bibinfo {author} {\bibfnamefont {Dam~Thanh}\
  \bibnamefont {Son}}\ and\ \bibinfo {author} {\bibfnamefont {Naoki}\
  \bibnamefont {Yamamoto}},\ }\bibfield  {title} {\enquote {\bibinfo {title}
  {Berry curvature, triangle anomalies, and the chiral magnetic effect in
  {Fermi} liquids},}\ }\href {\doibase 10.1103/PhysRevLett.109.181602}
  {\bibfield  {journal} {\bibinfo  {journal} {Phys. Rev. Lett.}\ }\textbf
  {\bibinfo {volume} {109}},\ \bibinfo {pages} {181602} (\bibinfo {year}
  {2012})},\ \Eprint {http://arxiv.org/abs/1203.2697} {arXiv:1203.2697}
  \BibitemShut {NoStop}%
\bibitem [{\citenamefont {Son}(2015)}]{Son_1502}%
  \BibitemOpen
  \bibfield  {author} {\bibinfo {author} {\bibfnamefont {Dam~Thanh}\
  \bibnamefont {Son}},\ }\bibfield  {title} {\enquote {\bibinfo {title} {Is the
  composite fermion a {Dirac} particle?}}\ }\href {\doibase
  10.1103/PhysRevX.5.031027} {\bibfield  {journal} {\bibinfo  {journal} {Phys.
  Rev. X}\ }\textbf {\bibinfo {volume} {5}},\ \bibinfo {pages} {031027}
  (\bibinfo {year} {2015})},\ \Eprint {http://arxiv.org/abs/1502.03446}
  {arXiv:1502.03446} \BibitemShut {NoStop}%
\bibitem [{\citenamefont {Seiberg}\ \emph {et~al.}(2016)\citenamefont
  {Seiberg}, \citenamefont {Senthil}, \citenamefont {Wang},\ and\ \citenamefont
  {Witten}}]{Seiberg_1606}%
  \BibitemOpen
  \bibfield  {author} {\bibinfo {author} {\bibfnamefont {Nathan}\ \bibnamefont
  {Seiberg}}, \bibinfo {author} {\bibfnamefont {T.}~\bibnamefont {Senthil}},
  \bibinfo {author} {\bibfnamefont {Chong}\ \bibnamefont {Wang}}, \ and\
  \bibinfo {author} {\bibfnamefont {Edward}\ \bibnamefont {Witten}},\
  }\bibfield  {title} {\enquote {\bibinfo {title} {{A duality web in 2+1
  dimensions and condensed matter physics}},}\ }\href {\doibase
  10.1016/j.aop.2016.08.007} {\bibfield  {journal} {\bibinfo  {journal} {Ann.
  Phys.}\ }\textbf {\bibinfo {volume} {374}},\ \bibinfo {pages} {395} (\bibinfo
  {year} {2016})},\ \Eprint {http://arxiv.org/abs/1606.01989}
  {arXiv:1606.01989} \BibitemShut {NoStop}%
\bibitem [{\citenamefont {Balram}\ \emph {et~al.}(2015)\citenamefont {Balram},
  \citenamefont {Tőke},\ and\ \citenamefont {Jain}}]{Balram_1506}%
  \BibitemOpen
  \bibfield  {author} {\bibinfo {author} {\bibfnamefont {Ajit~C.}\ \bibnamefont
  {Balram}}, \bibinfo {author} {\bibfnamefont {Csaba}\ \bibnamefont {Tőke}}, \
  and\ \bibinfo {author} {\bibfnamefont {J.~K.}\ \bibnamefont {Jain}},\
  }\bibfield  {title} {\enquote {\bibinfo {title} {{Luttinger} theorem for the
  strongly correlated {Fermi} liquid of composite fermions},}\ }\href {\doibase
  10.1103/PhysRevLett.115.186805} {\bibfield  {journal} {\bibinfo  {journal}
  {Phys. Rev. Lett.}\ }\textbf {\bibinfo {volume} {115}},\ \bibinfo {pages}
  {186805} (\bibinfo {year} {2015})},\ \Eprint
  {http://arxiv.org/abs/1506.02747} {arXiv:1506.02747} \BibitemShut {NoStop}%
\bibitem [{\citenamefont {Balram}\ and\ \citenamefont
  {Jain}(2017)}]{Balram_1707}%
  \BibitemOpen
  \bibfield  {author} {\bibinfo {author} {\bibfnamefont {Ajit~C.}\ \bibnamefont
  {Balram}}\ and\ \bibinfo {author} {\bibfnamefont {J.~K.}\ \bibnamefont
  {Jain}},\ }\bibfield  {title} {\enquote {\bibinfo {title} {{Fermi wave vector
  for the partially spin-polarized composite-fermion Fermi sea}},}\ }\href
  {\doibase 10.1103/PhysRevB.96.235102} {\bibfield  {journal} {\bibinfo
  {journal} {Phys. Rev. B}\ }\textbf {\bibinfo {volume} {96}},\ \bibinfo
  {pages} {235102} (\bibinfo {year} {2017})},\ \Eprint
  {http://arxiv.org/abs/1707.08623} {arXiv:1707.08623} \BibitemShut {NoStop}%
\bibitem [{\citenamefont {Keimer}\ \emph {et~al.}(2015)\citenamefont {Keimer},
  \citenamefont {Kivelson}, \citenamefont {Norman}, \citenamefont {Uchida},\
  and\ \citenamefont {Zaanen}}]{keimer2015quantum}%
  \BibitemOpen
  \bibfield  {author} {\bibinfo {author} {\bibfnamefont {B.}~\bibnamefont
  {Keimer}}, \bibinfo {author} {\bibfnamefont {S.~A.}\ \bibnamefont
  {Kivelson}}, \bibinfo {author} {\bibfnamefont {M.~R.}\ \bibnamefont
  {Norman}}, \bibinfo {author} {\bibfnamefont {S.}~\bibnamefont {Uchida}}, \
  and\ \bibinfo {author} {\bibfnamefont {J.}~\bibnamefont {Zaanen}},\
  }\bibfield  {title} {\enquote {\bibinfo {title} {{From quantum matter to
  high-temperature superconductivity in copper oxides}},}\ }\href {\doibase
  10.1038/nature14165} {\bibfield  {journal} {\bibinfo  {journal} {Nature}\
  }\textbf {\bibinfo {volume} {518}},\ \bibinfo {pages} {179} (\bibinfo {year}
  {2015})},\ \Eprint {http://arxiv.org/abs/1409.4673} {arXiv:1409.4673}
  \BibitemShut {NoStop}%
\bibitem [{\citenamefont {Wan}\ \emph {et~al.}(2011)\citenamefont {Wan},
  \citenamefont {Turner}, \citenamefont {Vishwanath},\ and\ \citenamefont
  {Savrasov}}]{Wan_2011}%
  \BibitemOpen
  \bibfield  {author} {\bibinfo {author} {\bibfnamefont {Xiangang}\
  \bibnamefont {Wan}}, \bibinfo {author} {\bibfnamefont {Ari~M.}\ \bibnamefont
  {Turner}}, \bibinfo {author} {\bibfnamefont {Ashvin}\ \bibnamefont
  {Vishwanath}}, \ and\ \bibinfo {author} {\bibfnamefont {Sergey~Y.}\
  \bibnamefont {Savrasov}},\ }\bibfield  {title} {\enquote {\bibinfo {title}
  {{Topological semimetal and Fermi-arc surface states in the electronic
  structure of pyrochlore iridates}},}\ }\href {\doibase
  10.1103/PhysRevB.83.205101} {\bibfield  {journal} {\bibinfo  {journal} {Phys.
  Rev. B}\ }\textbf {\bibinfo {volume} {83}},\ \bibinfo {pages} {205101}
  (\bibinfo {year} {2011})}\BibitemShut {NoStop}%
\bibitem [{\citenamefont {Read}\ and\ \citenamefont
  {Sachdev}(1991)}]{read1991large}%
  \BibitemOpen
  \bibfield  {author} {\bibinfo {author} {\bibfnamefont {N.}~\bibnamefont
  {Read}}\ and\ \bibinfo {author} {\bibfnamefont {Subir}\ \bibnamefont
  {Sachdev}},\ }\bibfield  {title} {\enquote {\bibinfo {title} {{Large-N
  expansion for frustrated quantum antiferromagnets}},}\ }\href {\doibase
  10.1103/physrevlett.66.1773} {\bibfield  {journal} {\bibinfo  {journal}
  {Phys. Rev. Lett.}\ }\textbf {\bibinfo {volume} {66}},\ \bibinfo {pages}
  {1773} (\bibinfo {year} {1991})}\BibitemShut {NoStop}%
\bibitem [{\citenamefont {Wen}(1991)}]{wen1991mean}%
  \BibitemOpen
  \bibfield  {author} {\bibinfo {author} {\bibfnamefont {X.~G.}\ \bibnamefont
  {Wen}},\ }\bibfield  {title} {\enquote {\bibinfo {title} {{Mean-field theory
  of spin-liquid states with finite energy gap and topological orders}},}\
  }\href {\doibase 10.1103/physrevb.44.2664} {\bibfield  {journal} {\bibinfo
  {journal} {Phys. Rev. B}\ }\textbf {\bibinfo {volume} {44}},\ \bibinfo
  {pages} {2664} (\bibinfo {year} {1991})}\BibitemShut {NoStop}%
\bibitem [{\citenamefont {Senthil}\ and\ \citenamefont
  {Fisher}(2000)}]{senthil2000z}%
  \BibitemOpen
  \bibfield  {author} {\bibinfo {author} {\bibfnamefont {T.}~\bibnamefont
  {Senthil}}\ and\ \bibinfo {author} {\bibfnamefont {Matthew P.~A.}\
  \bibnamefont {Fisher}},\ }\bibfield  {title} {\enquote {\bibinfo {title}
  {{${Z}_{2}$ gauge theory of electron fractionalization in strongly correlated
  systems}},}\ }\href {\doibase 10.1103/PhysRevB.62.7850} {\bibfield  {journal}
  {\bibinfo  {journal} {Phys. Rev. B}\ }\textbf {\bibinfo {volume} {62}},\
  \bibinfo {pages} {7850} (\bibinfo {year} {2000})},\ \Eprint
  {http://arxiv.org/abs/cond-mat/9910224} {arXiv:cond-mat/9910224} \BibitemShut
  {NoStop}%
\bibitem [{\citenamefont {Jalabert}\ and\ \citenamefont
  {Sachdev}(1991)}]{jalabert1991spontaneous}%
  \BibitemOpen
  \bibfield  {author} {\bibinfo {author} {\bibfnamefont {Rodolfo~A.}\
  \bibnamefont {Jalabert}}\ and\ \bibinfo {author} {\bibfnamefont {Subir}\
  \bibnamefont {Sachdev}},\ }\bibfield  {title} {\enquote {\bibinfo {title}
  {{Spontaneous alignment of frustrated bonds in an anisotropic,
  three-dimensional Ising model}},}\ }\href {\doibase 10.1103/physrevb.44.686}
  {\bibfield  {journal} {\bibinfo  {journal} {Phys. Rev. B}\ }\textbf {\bibinfo
  {volume} {44}},\ \bibinfo {pages} {686} (\bibinfo {year} {1991})}\BibitemShut
  {NoStop}%
\bibitem [{\citenamefont {Sachdev}\ and\ \citenamefont
  {Vojta}()}]{sachdev1999translational}%
  \BibitemOpen
  \bibfield  {author} {\bibinfo {author} {\bibfnamefont {Subir}\ \bibnamefont
  {Sachdev}}\ and\ \bibinfo {author} {\bibfnamefont {Matthias}\ \bibnamefont
  {Vojta}},\ }\bibfield  {title} {\enquote {\bibinfo {title} {Translational
  symmetry breaking in two-dimensional antiferromagnets and superconductors},}\
  }in\ \href@noop {} {\emph {\bibinfo {booktitle} {Proceedings of the
  International Workshop on Magnetic Excitations in Strongly Correlated
  Electrons, Hammatsu, Japan: August 19-22, 1999}}},\ \Eprint
  {http://arxiv.org/abs/cond-mat/9910231} {arXiv:cond-mat/9910231} \BibitemShut
  {NoStop}%
\bibitem [{\citenamefont {Etingof}\ \emph {et~al.}()\citenamefont {Etingof},
  \citenamefont {Nikshych}, \citenamefont {Ostrik},\ and\ \citenamefont {with
  an appendix~by Ehud~Meir}}]{Etingof_0909}%
  \BibitemOpen
  \bibfield  {author} {\bibinfo {author} {\bibfnamefont {Pavel}\ \bibnamefont
  {Etingof}}, \bibinfo {author} {\bibfnamefont {Dmitri}\ \bibnamefont
  {Nikshych}}, \bibinfo {author} {\bibfnamefont {Victor}\ \bibnamefont
  {Ostrik}}, \ and\ \bibinfo {author} {\bibnamefont {with an appendix~by
  Ehud~Meir}},\ }\bibfield  {title} {\enquote {\bibinfo {title} {{Fusion
  categories and homotopy theory}},}\ }\href@noop {} {\ }\Eprint
  {http://arxiv.org/abs/0909.3140} {arXiv:0909.3140} \BibitemShut {NoStop}%
\bibitem [{\citenamefont {Barkeshli}\ \emph {et~al.}(2019)\citenamefont
  {Barkeshli}, \citenamefont {Bonderson}, \citenamefont {Cheng},\ and\
  \citenamefont {Wang}}]{Barkeshli_1410}%
  \BibitemOpen
  \bibfield  {author} {\bibinfo {author} {\bibfnamefont {Maissam}\ \bibnamefont
  {Barkeshli}}, \bibinfo {author} {\bibfnamefont {Parsa}\ \bibnamefont
  {Bonderson}}, \bibinfo {author} {\bibfnamefont {Meng}\ \bibnamefont {Cheng}},
  \ and\ \bibinfo {author} {\bibfnamefont {Zhenghan}\ \bibnamefont {Wang}},\
  }\bibfield  {title} {\enquote {\bibinfo {title} {{Symmetry fractionalization,
  defects, and gauging of topological phases}},}\ }\href {\doibase
  10.1103/PhysRevB.100.115147} {\bibfield  {journal} {\bibinfo  {journal}
  {Phys. Rev. B}\ }\textbf {\bibinfo {volume} {100}},\ \bibinfo {pages}
  {115147} (\bibinfo {year} {2019})},\ \Eprint {http://arxiv.org/abs/1410.4540}
  {arXiv:1410.4540} \BibitemShut {NoStop}%
\bibitem [{\citenamefont {Benini}\ \emph {et~al.}(2019)\citenamefont {Benini},
  \citenamefont {Córdova},\ and\ \citenamefont {Hsin}}]{Benini_1803}%
  \BibitemOpen
  \bibfield  {author} {\bibinfo {author} {\bibfnamefont {Francesco}\
  \bibnamefont {Benini}}, \bibinfo {author} {\bibfnamefont {Clay}\ \bibnamefont
  {Córdova}}, \ and\ \bibinfo {author} {\bibfnamefont {Po-Shen}\ \bibnamefont
  {Hsin}},\ }\bibfield  {title} {\enquote {\bibinfo {title} {{On 2-group global
  symmetries and their anomalies}},}\ }\href {\doibase 10.1007/JHEP03(2019)118}
  {\bibfield  {journal} {\bibinfo  {journal} {J. High Energy Phys.}\ }\textbf
  {\bibinfo {volume} {2019}},\ \bibinfo {pages} {118} (\bibinfo {year}
  {2019})},\ \Eprint {http://arxiv.org/abs/1803.09336} {arXiv:1803.09336}
  \BibitemShut {NoStop}%
\bibitem [{\citenamefont {Hsin}\ and\ \citenamefont
  {Turzillo}(2020)}]{Hsin_1904}%
  \BibitemOpen
  \bibfield  {author} {\bibinfo {author} {\bibfnamefont {Po-Shen}\ \bibnamefont
  {Hsin}}\ and\ \bibinfo {author} {\bibfnamefont {Alex}\ \bibnamefont
  {Turzillo}},\ }\bibfield  {title} {\enquote {\bibinfo {title}
  {{Symmetry-enriched quantum spin liquids in (3 + 1)d}},}\ }\href {\doibase
  10.1007/JHEP09(2020)022} {\bibfield  {journal} {\bibinfo  {journal} {J. High
  Energy Phys.}\ }\textbf {\bibinfo {volume} {2020}},\ \bibinfo {pages} {22}
  (\bibinfo {year} {2020})},\ \Eprint {http://arxiv.org/abs/1904.11550}
  {arXiv:1904.11550} \BibitemShut {NoStop}%
\bibitem [{\citenamefont {Baez}\ and\ \citenamefont {Lauda}()}]{Baez_0307}%
  \BibitemOpen
  \bibfield  {author} {\bibinfo {author} {\bibfnamefont {John~C.}\ \bibnamefont
  {Baez}}\ and\ \bibinfo {author} {\bibfnamefont {Aaron~D.}\ \bibnamefont
  {Lauda}},\ }\bibfield  {title} {\enquote {\bibinfo {title}
  {Higher-dimensional algebra {V:} 2-groups},}\ }\href@noop {} {\ }\Eprint
  {http://arxiv.org/abs/math.QA/0307200} {arXiv:math.QA/0307200} \BibitemShut
  {NoStop}%
\bibitem [{\citenamefont {Gukov}\ and\ \citenamefont
  {Kapustin}()}]{Gukov_1307}%
  \BibitemOpen
  \bibfield  {author} {\bibinfo {author} {\bibfnamefont {Sergei}\ \bibnamefont
  {Gukov}}\ and\ \bibinfo {author} {\bibfnamefont {Anton}\ \bibnamefont
  {Kapustin}},\ }\bibfield  {title} {\enquote {\bibinfo {title} {Topological
  quantum field theory, nonlocal operators, and gapped phases of gauge
  theories},}\ }\href@noop {} {\ }\Eprint {http://arxiv.org/abs/1307.4793}
  {arXiv:1307.4793} \BibitemShut {NoStop}%
\bibitem [{\citenamefont {Kapustin}\ and\ \citenamefont
  {Thorngren}()}]{Kapustin_1309}%
  \BibitemOpen
  \bibfield  {author} {\bibinfo {author} {\bibfnamefont {Anton}\ \bibnamefont
  {Kapustin}}\ and\ \bibinfo {author} {\bibfnamefont {Ryan}\ \bibnamefont
  {Thorngren}},\ }\bibfield  {title} {\enquote {\bibinfo {title} {{Higher
  symmetry and gapped phases of gauge theories}},}\ }\href@noop {} {\ }\Eprint
  {http://arxiv.org/abs/1309.4721} {arXiv:1309.4721} \BibitemShut {NoStop}%
\bibitem [{\citenamefont {Lannert}\ \emph {et~al.}(2001)\citenamefont
  {Lannert}, \citenamefont {Fisher},\ and\ \citenamefont
  {Senthil}}]{lannert2001quantum}%
  \BibitemOpen
  \bibfield  {author} {\bibinfo {author} {\bibfnamefont {C.}~\bibnamefont
  {Lannert}}, \bibinfo {author} {\bibfnamefont {Matthew P.~A.}\ \bibnamefont
  {Fisher}}, \ and\ \bibinfo {author} {\bibfnamefont {T.}~\bibnamefont
  {Senthil}},\ }\bibfield  {title} {\enquote {\bibinfo {title} {{Quantum
  confinement transition in a d-wave superconductor}},}\ }\href {\doibase
  10.1103/PhysRevB.63.134510} {\bibfield  {journal} {\bibinfo  {journal} {Phys.
  Rev. B}\ }\textbf {\bibinfo {volume} {63}},\ \bibinfo {pages} {134510}
  (\bibinfo {year} {2001})},\ \Eprint {http://arxiv.org/abs/cond-mat/0007002}
  {arXiv:cond-mat/0007002} \BibitemShut {NoStop}%
\bibitem [{\citenamefont {Balents}\ \emph {et~al.}(2005)\citenamefont
  {Balents}, \citenamefont {Bartosch}, \citenamefont {Burkov}, \citenamefont
  {Sachdev},\ and\ \citenamefont {Sengupta}}]{balents2005putting}%
  \BibitemOpen
  \bibfield  {author} {\bibinfo {author} {\bibfnamefont {Leon}\ \bibnamefont
  {Balents}}, \bibinfo {author} {\bibfnamefont {Lorenz}\ \bibnamefont
  {Bartosch}}, \bibinfo {author} {\bibfnamefont {Anton}\ \bibnamefont
  {Burkov}}, \bibinfo {author} {\bibfnamefont {Subir}\ \bibnamefont {Sachdev}},
  \ and\ \bibinfo {author} {\bibfnamefont {Krishnendu}\ \bibnamefont
  {Sengupta}},\ }\bibfield  {title} {\enquote {\bibinfo {title} {{Putting
  competing orders in their place near the Mott transition}},}\ }\href
  {\doibase 10.1103/PhysRevB.71.144508} {\bibfield  {journal} {\bibinfo
  {journal} {Phys. Rev. B}\ }\textbf {\bibinfo {volume} {71}},\ \bibinfo
  {pages} {144508} (\bibinfo {year} {2005})},\ \Eprint
  {http://arxiv.org/abs/cond-mat/0408329} {arXiv:cond-mat/0408329} \BibitemShut
  {NoStop}%
\bibitem [{\citenamefont {Dzyaloshinskii}(2003)}]{Dzyaloshinskii__2003}%
  \BibitemOpen
  \bibfield  {author} {\bibinfo {author} {\bibfnamefont {Igor}\ \bibnamefont
  {Dzyaloshinskii}},\ }\bibfield  {title} {\enquote {\bibinfo {title} {{Some
  consequences of the Luttinger theorem: The Luttinger surfaces in non-Fermi
  liquids and Mott insulators}},}\ }\href {\doibase 10.1103/PhysRevB.68.085113}
  {\bibfield  {journal} {\bibinfo  {journal} {Phys. Rev. B}\ }\textbf {\bibinfo
  {volume} {68}},\ \bibinfo {pages} {085113} (\bibinfo {year}
  {2003})}\BibitemShut {NoStop}%
\bibitem [{\citenamefont {Stanescu}\ and\ \citenamefont
  {Kotliar}(2006)}]{Stanescu_0508}%
  \BibitemOpen
  \bibfield  {author} {\bibinfo {author} {\bibfnamefont {Tudor~D.}\
  \bibnamefont {Stanescu}}\ and\ \bibinfo {author} {\bibfnamefont {Gabriel}\
  \bibnamefont {Kotliar}},\ }\bibfield  {title} {\enquote {\bibinfo {title}
  {{Fermi arcs and hidden zeros of the Green function in the pseudogap
  state}},}\ }\href {\doibase 10.1103/PhysRevB.74.125110} {\bibfield  {journal}
  {\bibinfo  {journal} {Phys. Rev. B}\ }\textbf {\bibinfo {volume} {74}},\
  \bibinfo {pages} {125110} (\bibinfo {year} {2006})},\ \Eprint
  {http://arxiv.org/abs/cond-mat/0508302} {arXiv:cond-mat/0508302} \BibitemShut
  {NoStop}%
\bibitem [{\citenamefont {Yang}\ \emph {et~al.}(2006)\citenamefont {Yang},
  \citenamefont {Rice},\ and\ \citenamefont {Zhang}}]{Yang_0602}%
  \BibitemOpen
  \bibfield  {author} {\bibinfo {author} {\bibfnamefont {Kai-Yu}\ \bibnamefont
  {Yang}}, \bibinfo {author} {\bibfnamefont {T.~M.}\ \bibnamefont {Rice}}, \
  and\ \bibinfo {author} {\bibfnamefont {Fu-Chun}\ \bibnamefont {Zhang}},\
  }\bibfield  {title} {\enquote {\bibinfo {title} {{Phenomenological theory of
  the pseudogap state}},}\ }\href {\doibase 10.1103/PhysRevB.73.174501}
  {\bibfield  {journal} {\bibinfo  {journal} {Phys. Rev. B}\ }\textbf {\bibinfo
  {volume} {73}},\ \bibinfo {pages} {174501} (\bibinfo {year} {2006})},\
  \Eprint {http://arxiv.org/abs/cond-mat/0602164} {arXiv:cond-mat/0602164}
  \BibitemShut {NoStop}%
\bibitem [{\citenamefont {Rosch}(2007)}]{Rosch_0602}%
  \BibitemOpen
  \bibfield  {author} {\bibinfo {author} {\bibfnamefont {A.}~\bibnamefont
  {Rosch}},\ }\bibfield  {title} {\enquote {\bibinfo {title} {{Breakdown of
  Luttinger's theorem in two-orbital Mott insulators}},}\ }\href {\doibase
  10.1140/epjb/e2007-00312-3} {\bibfield  {journal} {\bibinfo  {journal} {Eur.
  Phys. J. B}\ }\textbf {\bibinfo {volume} {59}},\ \bibinfo {pages} {495}
  (\bibinfo {year} {2007})},\ \Eprint {http://arxiv.org/abs/cond-mat/0602656}
  {arXiv:cond-mat/0602656} \BibitemShut {NoStop}%
\bibitem [{\citenamefont {Dave}\ \emph {et~al.}(2013)\citenamefont {Dave},
  \citenamefont {Phillips},\ and\ \citenamefont {Kane}}]{Dave_1207}%
  \BibitemOpen
  \bibfield  {author} {\bibinfo {author} {\bibfnamefont {Kiaran~B.}\
  \bibnamefont {Dave}}, \bibinfo {author} {\bibfnamefont {Philip~W.}\
  \bibnamefont {Phillips}}, \ and\ \bibinfo {author} {\bibfnamefont
  {Charles~L.}\ \bibnamefont {Kane}},\ }\bibfield  {title} {\enquote {\bibinfo
  {title} {Absence of {Luttinger}'s theorem due to zeros in the single-particle
  green function},}\ }\href {\doibase 10.1103/PhysRevLett.110.090403}
  {\bibfield  {journal} {\bibinfo  {journal} {Phys. Rev. Lett.}\ }\textbf
  {\bibinfo {volume} {110}},\ \bibinfo {pages} {090403} (\bibinfo {year}
  {2013})},\ \Eprint {http://arxiv.org/abs/1207.4201} {arXiv:1207.4201}
  \BibitemShut {NoStop}%
\bibitem [{\citenamefont {Hartnoll}\ \emph {et~al.}(2011)\citenamefont
  {Hartnoll}, \citenamefont {Hofman},\ and\ \citenamefont
  {Tavanfar}}]{Hartnoll_1011}%
  \BibitemOpen
  \bibfield  {author} {\bibinfo {author} {\bibfnamefont {S.~A.}\ \bibnamefont
  {Hartnoll}}, \bibinfo {author} {\bibfnamefont {D.~M.}\ \bibnamefont
  {Hofman}}, \ and\ \bibinfo {author} {\bibfnamefont {A.}~\bibnamefont
  {Tavanfar}},\ }\bibfield  {title} {\enquote {\bibinfo {title}
  {{Holographically smeared Fermi surface: Quantum oscillations and Luttinger
  count in electron stars}},}\ }\href {\doibase 10.1209/0295-5075/95/31002}
  {\bibfield  {journal} {\bibinfo  {journal} {Europhys. Lett.}\ }\textbf
  {\bibinfo {volume} {95}},\ \bibinfo {pages} {31002} (\bibinfo {year}
  {2011})},\ \Eprint {http://arxiv.org/abs/1011.2502} {arXiv:1011.2502}
  \BibitemShut {NoStop}%
\bibitem [{\citenamefont {Hartnoll}(2012)}]{HartnollBlackHoles}%
  \BibitemOpen
  \bibfield  {author} {\bibinfo {author} {\bibfnamefont {Sean~A.}\ \bibnamefont
  {Hartnoll}},\ }\bibfield  {title} {\enquote {\bibinfo {title} {Horizons,
  holography and condensed matter},}\ }in\ \href@noop {} {\emph {\bibinfo
  {booktitle} {Black Holes in Higher Dimensions}}},\ \bibinfo {editor} {edited
  by\ \bibinfo {editor} {\bibfnamefont {Gary~T.}\ \bibnamefont {Horowitz}}}\
  (\bibinfo  {publisher} {Cambridge University Press},\ \bibinfo {address} {New
  York},\ \bibinfo {year} {2012})\ pp.\ \bibinfo {pages} {387--419},\ \Eprint
  {http://arxiv.org/abs/1106.4324} {arXiv:1106.4324} \BibitemShut {NoStop}%
\bibitem [{\citenamefont {Iqbal}\ and\ \citenamefont {Liu}(2012)}]{Iqbal_1112}%
  \BibitemOpen
  \bibfield  {author} {\bibinfo {author} {\bibfnamefont {Nabil}\ \bibnamefont
  {Iqbal}}\ and\ \bibinfo {author} {\bibfnamefont {Hong}\ \bibnamefont {Liu}},\
  }\bibfield  {title} {\enquote {\bibinfo {title} {{Luttinger's theorem,
  superfluid vortices and holography}},}\ }\href {\doibase
  10.1088/0264-9381/29/19/194004} {\bibfield  {journal} {\bibinfo  {journal}
  {Classical Quantum Gravity}\ }\textbf {\bibinfo {volume} {29}},\ \bibinfo
  {pages} {194004} (\bibinfo {year} {2012})},\ \Eprint
  {http://arxiv.org/abs/1112.3671} {arXiv:1112.3671} \BibitemShut {NoStop}%
\bibitem [{\citenamefont {Hashimoto}\ and\ \citenamefont
  {Iizuka}(2012)}]{Hashimoto_1203}%
  \BibitemOpen
  \bibfield  {author} {\bibinfo {author} {\bibfnamefont {Koji}\ \bibnamefont
  {Hashimoto}}\ and\ \bibinfo {author} {\bibfnamefont {Norihiro}\ \bibnamefont
  {Iizuka}},\ }\bibfield  {title} {\enquote {\bibinfo {title} {{A comment on
  holographic Luttinger theorem}},}\ }\href {\doibase 10.1007/JHEP07(2012)064}
  {\bibfield  {journal} {\bibinfo  {journal} {J. High Energy Phys.}\ }\textbf
  {\bibinfo {volume} {2012}},\ \bibinfo {pages} {64} (\bibinfo {year}
  {2012})},\ \Eprint {http://arxiv.org/abs/1203.5388} {arXiv:1203.5388}
  \BibitemShut {NoStop}%
\bibitem [{\citenamefont {Putikka}\ \emph {et~al.}(1998)\citenamefont
  {Putikka}, \citenamefont {Luchini},\ and\ \citenamefont
  {Singh}}]{Putikka_9803}%
  \BibitemOpen
  \bibfield  {author} {\bibinfo {author} {\bibfnamefont {W.~O.}\ \bibnamefont
  {Putikka}}, \bibinfo {author} {\bibfnamefont {M.~U.}\ \bibnamefont
  {Luchini}}, \ and\ \bibinfo {author} {\bibfnamefont {R.~R.~P.}\ \bibnamefont
  {Singh}},\ }\bibfield  {title} {\enquote {\bibinfo {title} {{Violation of
  Luttinger's Theorem in the Two-Dimensional $\mathit{t}$- $\mathit{J}$
  Model}},}\ }\href {\doibase 10.1103/PhysRevLett.81.2966} {\bibfield
  {journal} {\bibinfo  {journal} {Phys. Rev. Lett.}\ }\textbf {\bibinfo
  {volume} {81}},\ \bibinfo {pages} {2966} (\bibinfo {year} {1998})},\ \Eprint
  {http://arxiv.org/abs/cond-mat/9803140} {arXiv:cond-mat/9803140} \BibitemShut
  {NoStop}%
\bibitem [{\citenamefont {Cappelluti}\ and\ \citenamefont
  {Zeyher}(1999)}]{Cappelluti_9906}%
  \BibitemOpen
  \bibfield  {author} {\bibinfo {author} {\bibfnamefont {Emmanuele}\
  \bibnamefont {Cappelluti}}\ and\ \bibinfo {author} {\bibfnamefont {Roland}\
  \bibnamefont {Zeyher}},\ }\bibfield  {title} {\enquote {\bibinfo {title}
  {Violation of luttinger's theorem in strongly correlated electronic systems
  within a {1/N} expansion},}\ }\href {\doibase 10.1142/S021797929900254X}
  {\bibfield  {journal} {\bibinfo  {journal} {Int. J. Mod Phys B}\ }\textbf
  {\bibinfo {volume} {13}},\ \bibinfo {pages} {2607} (\bibinfo {year}
  {1999})},\ \Eprint {http://arxiv.org/abs/cond-mat/9906349}
  {arXiv:cond-mat/9906349} \BibitemShut {NoStop}%
\bibitem [{\citenamefont {Else}\ and\ \citenamefont {Senthil}()}]{Else_2010}%
  \BibitemOpen
  \bibfield  {author} {\bibinfo {author} {\bibfnamefont {Dominic~V.}\
  \bibnamefont {Else}}\ and\ \bibinfo {author} {\bibfnamefont {T.}~\bibnamefont
  {Senthil}},\ }\bibfield  {title} {\enquote {\bibinfo {title} {{Strange metals
  as ersatz Fermi liquids}},}\ }\href@noop {} {\ }\Eprint
  {http://arxiv.org/abs/2010.10523} {arXiv:2010.10523} \BibitemShut {NoStop}%
\bibitem [{\citenamefont {Kitaev}(2009)}]{Kitaev_0901}%
  \BibitemOpen
  \bibfield  {author} {\bibinfo {author} {\bibfnamefont {Alexei}\ \bibnamefont
  {Kitaev}},\ }\bibfield  {title} {\enquote {\bibinfo {title} {Periodic table
  for topological insulators},}\ }\href {\doibase
  https://aip.scitation.org/doi/abs/10.1063/1.3149495} {\bibfield  {journal}
  {\bibinfo  {journal} {AIP Conference Proceedings}\ }\textbf {\bibinfo
  {volume} {1134}},\ \bibinfo {pages} {22} (\bibinfo {year} {2009})},\ \Eprint
  {http://arxiv.org/abs/0901.2686} {arXiv:0901.2686} \BibitemShut {NoStop}%
\bibitem [{\citenamefont {Thorngren}\ and\ \citenamefont
  {Else}(2018)}]{Thorngren_1612}%
  \BibitemOpen
  \bibfield  {author} {\bibinfo {author} {\bibfnamefont {Ryan}\ \bibnamefont
  {Thorngren}}\ and\ \bibinfo {author} {\bibfnamefont {Dominic~V.}\
  \bibnamefont {Else}},\ }\bibfield  {title} {\enquote {\bibinfo {title}
  {Gauging spatial symmetries and the classification of topological crystalline
  phases},}\ }\href {\doibase 10.1103/PhysRevX.8.011040} {\bibfield  {journal}
  {\bibinfo  {journal} {Phys. Rev. X}\ }\textbf {\bibinfo {volume} {8}},\
  \bibinfo {pages} {011040} (\bibinfo {year} {2018})},\ \Eprint
  {http://arxiv.org/abs/1612.00846} {arXiv:1612.00846} \BibitemShut {NoStop}%
\bibitem [{\citenamefont {Else}\ \emph {et~al.}(2020)\citenamefont {Else},
  \citenamefont {Ho},\ and\ \citenamefont {Dumitrescu}}]{Else_1910}%
  \BibitemOpen
  \bibfield  {author} {\bibinfo {author} {\bibfnamefont {Dominic~V.}\
  \bibnamefont {Else}}, \bibinfo {author} {\bibfnamefont {Wen~Wei}\
  \bibnamefont {Ho}}, \ and\ \bibinfo {author} {\bibfnamefont {Philipp~T.}\
  \bibnamefont {Dumitrescu}},\ }\bibfield  {title} {\enquote {\bibinfo {title}
  {Long-lived interacting phases of matter protected by multiple
  time-translation symmetries in quasiperiodically driven systems},}\ }\href
  {\doibase 10.1103/PhysRevX.10.021032} {\bibfield  {journal} {\bibinfo
  {journal} {Phys. Rev. X}\ }\textbf {\bibinfo {volume} {10}},\ \bibinfo
  {pages} {021032} (\bibinfo {year} {2020})},\ \Eprint
  {http://arxiv.org/abs/1910.03584} {arXiv:1910.03584} \BibitemShut {NoStop}%
\bibitem [{\citenamefont {Duval}\ \emph {et~al.}(2006)\citenamefont {Duval},
  \citenamefont {Horváth}, \citenamefont {Horváthy}, \citenamefont
  {Martini},\ and\ \citenamefont {Stichel}}]{Duval_0506}%
  \BibitemOpen
  \bibfield  {author} {\bibinfo {author} {\bibfnamefont {C.}~\bibnamefont
  {Duval}}, \bibinfo {author} {\bibfnamefont {Z.}~\bibnamefont {Horváth}},
  \bibinfo {author} {\bibfnamefont {P.~A.}\ \bibnamefont {Horváthy}}, \bibinfo
  {author} {\bibfnamefont {L.}~\bibnamefont {Martini}}, \ and\ \bibinfo
  {author} {\bibfnamefont {P.~C.}\ \bibnamefont {Stichel}},\ }\bibfield
  {title} {\enquote {\bibinfo {title} {Berry phase correction to electron
  density and ``exotic'' dynamics},}\ }\href {\doibase
  10.1142/S0217984906010573} {\bibfield  {journal} {\bibinfo  {journal} {Mod.
  Phys. Lett. B}\ }\textbf {\bibinfo {volume} {20}},\ \bibinfo {pages} {373}
  (\bibinfo {year} {2006})},\ \Eprint {http://arxiv.org/abs/cond-mat/0506051}
  {arXiv:cond-mat/0506051} \BibitemShut {NoStop}%
\bibitem [{\citenamefont {Lee}(2013)}]{ManifoldBook}%
  \BibitemOpen
  \bibfield  {author} {\bibinfo {author} {\bibfnamefont {John~M.}\ \bibnamefont
  {Lee}},\ }\href@noop {} {\emph {\bibinfo {title} {Introduction to smooth
  manifolds}}}\ (\bibinfo  {publisher} {Springer},\ \bibinfo {address} {New
  York},\ \bibinfo {year} {2013})\BibitemShut {NoStop}%
\bibitem [{Mat()}]{MathOverflow}%
  \BibitemOpen
  \href@noop {} {}\bibinfo {note}
  {\url{https://mathoverflow.net/questions/361149/does-a-compact-lie-group-have-finitely-many-conjugacy-classes-of-maximal-abelian}}\BibitemShut
  {NoStop}%
\bibitem [{\citenamefont {Montgomery}\ and\ \citenamefont
  {Zippin}(1942)}]{Montgomery_1942}%
  \BibitemOpen
  \bibfield  {author} {\bibinfo {author} {\bibfnamefont {Deane}\ \bibnamefont
  {Montgomery}}\ and\ \bibinfo {author} {\bibfnamefont {Leo}\ \bibnamefont
  {Zippin}},\ }\bibfield  {title} {\enquote {\bibinfo {title} {{A theorem of
  Lie groups}},}\ }\href {\doibase 10.1090/S0002-9904-1942-07699-3} {\bibfield
  {journal} {\bibinfo  {journal} {Bulletin of the American Mathematical
  Society}\ }\textbf {\bibinfo {volume} {48}},\ \bibinfo {pages} {448}
  (\bibinfo {year} {1942})}\BibitemShut {NoStop}%
\bibitem [{\citenamefont {Thorngren}()}]{thorngren2018anomalies}%
  \BibitemOpen
  \bibfield  {author} {\bibinfo {author} {\bibfnamefont {Ryan}\ \bibnamefont
  {Thorngren}},\ }\bibfield  {title} {\enquote {\bibinfo {title} {Anomalies and
  bosonization},}\ }\href@noop {} {\ }\Eprint {http://arxiv.org/abs/1810.04414}
  {arXiv:1810.04414} \BibitemShut {NoStop}%
\end{thebibliography}%
\end{document}